\documentclass[10pt]{article}
\usepackage[top=2.5cm,bottom=2.5cm,left=2.5cm,right=2.5cm]{geometry}

\usepackage{amsmath, amssymb, amsfonts, amsthm}
\usepackage{latexsym}
\usepackage{graphicx}
\usepackage{color}
\usepackage{dsfont}
\usepackage[mathscr]{eucal} 
\usepackage{mathtools}
\usepackage{tikz}
\usetikzlibrary{positioning}

\numberwithin{equation}{section}

\theoremstyle{plain}
\newtheorem{thm}{Theorem}[section]
\newtheorem{prp}{Proposition}[section]
\newtheorem{lem}{Lemma}[section]
\newtheorem{cor}{Corollary}[section]

\theoremstyle{definition}
\newtheorem{defn}{Definition}[section]
\newtheorem{rmk}{Remark}[section]

\newtheorem{exm}{Example}[section]
\newtheorem{exr}{Exercise}[section]

\theoremstyle{remark}

\newcommand{\mf}{\mathfrak}
\newcommand{\mc}{\mathcal}

\newcommand{\bb}{\mathbb}

\newcommand{\sV}{\mathscr{V}}
\newcommand{\sW}{\mathscr{W}}

\newcommand{\C}{\mathbb{C}}

\newcommand{\Z}{\mathbb{Z}}
\newcommand{\F}{\mathbb{F}}
\newcommand{\bZ}{\mathbb{Z}}

\newcommand{\blu}{\color{blue}}

\newcommand{\End}{{\rm End}}
\newcommand{\Der}{{\rm Der}}
\newcommand{\Lie}{{\rm Lie}}
\newcommand{\Ker}{{\rm Ker}}
\newcommand{\Res}{{\rm Res\,}}
\newcommand{\Vir}{{\rm Vir}}
\newcommand{\ad}{{\rm ad}}
\newcommand{\tc}{{\rm tc}}
\newcommand{\Span}{{\rm span}}
\newcommand{\vac}{{|0\rangle}}
\newcommand{\Zhu}{{\rm Zhu}}
\newcommand{\1}{\mathds{1}}

\newcommand{\one}{|1\rangle}

\begin{document}


\begin{center}
{\Large{\textbf{Introduction to vertex algebras, Poisson vertex algebras,\\[.5em]
and integrable Hamiltonian PDE}}} 
\bigskip

Victor Kac 
\end{center}

\begin{abstract}
\noindent
These lectures were given in Session 1: ``Vertex algebras,
W-algebras, and applications'' of INdAM Intensive research period
``Perspectives in Lie Theory'' at the  Centro di Ricerca Matematica Ennio De Giorgi,  Pisa, Italy, December 9, 2014 -- February 28, 2015.
\end{abstract}

\tableofcontents

\newpage
\section*{Notation}

\begin{itemize}

\item $U[z]$: polynomials with coefficients in a vector space $U$.

\item $U[z,z^{-1}]$: Laurent polynomials.

\item $U[[z]]$: formal power series.

\item $U((z))$: formal Laurent series.

\item $U[[z,z^{-1}]]$: bilateral series.

\item $\Z_+ = \{0,1,2,\ldots\}$.

\item $\F$: the base field, a field of characteristic 0. All vector spaces are considered over $\F$.

\end{itemize}

\section*{About the LaTeX'ing of these notes}
These lecture notes were typeset by Vidas Regelskis (Lectures 1 and 6), Tam\'{a}s F. G\"orbe (2), Xiao He (3), Biswajit Ransingh (4) and Laura Fedele (5 and 6). 

The author would like to thank all the above mentioned scribes for their work, especially Laura Fedele and Vidas Regelskis for many corrections to the edited manuscript.

\newpage

\section{Lecture 1 (December 9, 2014)} \label{sec:1}

In the first lecture we give the definition of a vertex algebra and explain calculus of formal distributions. We end the lecture by giving two examples of non-commutative vertex algebras: the free boson and the free fermion. 

\subsection{Definition of a vertex algebra}
\label{subsec:1.1}


From a physicist's point of view, a vertex algebra can be understood as an algebra of chiral fields of a 2-dimensional conformal field theory. This point of view is explained in my book \cite{VAB}.

From a mathematician's point of view a vertex algebra can be understood as a natural ``infinite'' analogue of a unital commutative associative differential algebra. Recall that a differential algebra is an algebra $ V $ with a derivation $ T. $ A simple, but important remark is that a unital algebra $ V $ is commutative and associative if and only if 
\begin{equation}
\label{e1.1}
\hat{a}\hat{b} =\hat{b} \hat{a}, \quad a,b \in V,
\end{equation}
where $\hat{a}$ denotes the operator of left multiplication by $a\in V$.
\begin{exr}
\label{exer1.1}
Prove this remark.
\end{exr}

A vertex algebra is roughly a unital differential algebra  with a product, depending on a parameter $ z, $ satisfying a locality axiom, similar to \eqref{e1.1}. To be more precise, let me first introduce the notion of a $ z $-algebra. (Sorry for the awkward name, but I was unable to find a better one.)
\begin{defn}
\label{def1.1}
A \emph{$ z $-algebra} is a vector space $ V $ endowed with a bilinear (over $ \mathbb{F} $) product, valued in $ V((z)), a \otimes b \mapsto a(z)b, $ endowed with a derivation $ T $ of this product:
\begin{enumerate}
\item[(i)] $ T(a(z)b) = (Ta)(z)b + a(z)Tb, $
 
such that the following consistency property holds:
\item[(ii)] $ (Ta)(z) = \partial_z a(z).$
\end{enumerate}

Here and further on we denote by $ a(z) $ the operator of left multiplication by $ a\in V $ in the $ z $-algebra $ V. $ Using the standard notation 
\begin{equation}
\label{e1.2}
a(z)b = \sum_{n \in \mathbb{Z}} (a_{(n)} b) z^{-n-1}, 
\end{equation}
\noindent we can write
\begin{equation}
\label{e1.3}
a(z) = \sum_{n \in \mathbb{Z}} a_{(n)}  z^{-n-1}, \mbox{ where } a_{(n)} \in \End V.
\end{equation}
\end{defn}
The bilinear (over $ \mathbb{F} $) product $ a_{(n)}b $ is called the \emph{$ n $-th product}. Note that $ a(z) $ is an \emph{$\End V$-valued distribution}, i.e., an element of $(\End V)[[z,z^{-1}]]$.
Moreover, $a(z)$ is a \emph{quantum field}, i.e., $a_{(n)}b=0$ for  $b\in V$ and sufficiently large $n$ (depending on $b$). 
\begin{rmk}
\label{rem1.1}
Axioms (i) and (ii) of a $ z $-algebra imply the following \emph{translation covariance} property of $ a(z): $
\begin{equation}
[T,a(z)]=\partial_za(z),\quad\text{i.e.,}\quad [T,a_{(n)}]=-na_{(n-1)},
\quad\forall n\in\Z .
\label{1.4}
\end{equation}
Moreover, the translation covariance of $ a(z) $ and either of the axioms (i) or (ii) in Definition \ref{def1.1} imply the other axiom. 
\end{rmk}

Next, we define a unital $z$-algabra. 
\begin{defn}
\label{def1.2}
A \emph{unit element} of a $ z $-algebra $ V $ is a non-zero vector $ 1 \in V $, such that 
\[ 1(z)a = a, \,\,\hbox{and}\,\,a(z)1 = a \mod z V[[z]] .\]
\end{defn}
\begin{lem}
\label{L:2}
Let $V$ be a vector space, let $1 \in V$ and $T\in \End V$ be such that $T1 =0$. Then
\begin{enumerate}
\item[(a)] For any translation covariant (with respect to $T$) quantum field $a(z)$, we have $a(z)1 \in V[[z]]$.
\item[(b)] \begin{equation}
a(z)\one=e^{zT}a\;\;(=\sum_{n=0}^\infty\frac{z^n}{n!}\,T^n(a)), \mbox{ where } a = a_{(-1)}1.
\label{e1.5}
\end{equation}
\end{enumerate}
\end{lem}
\begin{proof}
For (a) we have to prove that $a_{(n)}1 =0$ for all $n\in\Z_+$. Since $a(z)$ is a quantum field,
$a_{(n)}1=0$ for $n\geq N$ with some $N\in\Z_+$. Also by translation covariance
we have $[T,a_{(n)}]=-na_{(n-1)}$ for all $n\in\Z$. Apply both sides of the last equality to $1$:
\begin{equation}
[T,a_{(n)}]1=Ta_{(n)}1-a_{(n)}T1=Ta_{(n)}1=-na_{(n-1)}1.
\label{1.5}
\end{equation}
Therefore $a_{(n)}1=0$ for $n>0$ implies $a_{(n-1)}1=0$. Hence $a_{(n)}1=0$
for all $n\in\Z_+$ and $a(z)1\in V[[z]]$.

Now we prove (b). By (a), the LHS of \eqref{e1.5} lies in $ V[[z]]. $ Both sides are solutions of the differential equation
\begin{equation}
\label{e1.7}
\frac{df}{dz} = T f(z), \quad f(z) \in V[[z]].
\end{equation}
For the RHS it is obvious, and for the LHS it follows from \eqref{1.4} and 
$ T 1 = 0 $:
\begin{equation}
\label{e1.8}
\partial_z a(z) 1 = Ta(z)1 - a(z) T 1 = Ta(z) 1.
\end{equation}
Both sides obviously satisfy the same initial condition $ f(0) = a, $ hence they are equal. 
\end{proof}
Since,  $1_{(-1)}1=1$ and $T$ is a derivation of $n$-th products, we have in a unital $z$-algebra:  
\begin{equation}
\label{e1.9} 
T1 = 0.
\end{equation}
Note that Lemma \ref{L:2}(a) implies that $ a(z)1 \in V[[z]], $  and by Lemma \ref{L:2}(b) one actually has (\ref{e1.5}).
Lemma \ref{L:2}(b) implies that
\begin{equation}
\label{e1.10}
Ta = a_{(-2)}1,
\end{equation}
\noindent so that the derivation $ T $ is ``built in'' the product of a unital $ z $-algebra.



Now we can define a vertex algebra.
\begin{defn}
\label{def1.3}
\begin{enumerate}
\item[(a)] A $ z $-algebra is called \emph{local} if 
\begin{equation}
\label{e1.11}
(z-w)^{N_{ab}} a(z) b(w) = (z-w)^{N_{ab}} b(w) a(z), \mbox{ for some } N_{ab} \in \bZ_+ (\mbox{depending on } a,b \in V).
\end{equation}
\item[(b)] A \emph{vertex algebra} is a local unital $ z $-algebra.
\end{enumerate}
\end{defn}
A frequently asked question is: why one cannot cancel $ (z-w)^{N_{ab}} $ on both sides of \eqref{e1.11}? As we will see in a moment, the answer is: due to the existence of the delta function. In fact, the case $ N_{ab} = 0 $ for all $ a,b \in V $ is not very interesting, since all such vertex algebras correspond bijectively to unital commutative associate differential algebras, as 
Exercise \ref{exr1.2} below demonstrates.  
\begin{exm} \label{ex1} 
A commutative vertex algebra, i.e., $[a(z),b(w)]=0$ for all $a,b\in V$, can 
be constructed
as follows. Take $V$ to be a unital commutative associative algebra with a derivation
$T$ . Then $V$ is a commutative vertex algebra with the product
$a(z)b=(e^{zT}a)b$.
\end{exm}
\begin{exr}
\label{exr1.2}
Check that the above example is indeed a commutative vertex algebra.
Using Lemma \ref{L:2}, prove that all commutative vertex algebras are of the
form given in Example \ref{ex1}.
\end{exr}


\begin{rmk}
\label{remark1.2}  
A unital $z$-algebra $V$ is a vector space with unit element 1 and  bilinear products $a_{(n)}b$, $n\in\Z$. (Recall that $ T $ is obtained by \eqref{e1.9}.) Through these bilinear products we can naturally define $z$-algebra homomorphisms/isomorphisms, and subalgebras/ideals. Namely, a \emph{homomorphism} between two $z$-algebras $V$ and $V^\prime$ is a linear map $f$ such that $ f(1) =1 $ and $f(a)_{(n)}f(b)=f(a_{(n)}b), \, \forall a, b \in V$ and $\forall n\in \Z$. It is an \emph{isomorphism} if it is a homomorphism of $z$-algebras and also an isomorphism as vector spaces. A \emph{subalgebra} is a subspace $W$ of $V$ which contains $1$, such that $a_{(n)}b\in W, \, \forall a, b \in W$ and $\forall n\in \Z$. And an \emph{ideal} is a subspace $I$ such that $a_{(n)}b, b_{(n)}a\in I, \, \forall a \in V,\, \forall b \in I$ and $\forall n\in \Z$. Note that both a subalgebra and an ideal are $T$-invariant due to \eqref{e1.9}, and if an ideal $I$ contains $1$, then it must be the whole vertex algebra $V$ .
\end{rmk}

Now I will give another definition of a vertex algebra, in the spirit of quantum field theory, using language closer to physics: a unit element is called a vacuum vector, element of a vector space is called a state, etc. 
\begin{defn}
\label{def1.4}
A vertex algebra is a vector space $ V $ (the space of states) with a non-zero vector $ \vac $ (the vacuum vector) and a linear map from $ V $  to the space of $ \End V $-valued quantum fields (the state-field correspondence) $ a \mapsto a(z), $ satisfying the following axioms:
\begin{itemize}
\item[] vacuum axiom: $ \vac(z) = I_V, \ a(z) \vac = a + (Ta) z + \ldots ,  $ where $ T \in \End V $;
\item[] translation covariance axiom \eqref{1.4}; 
\item[] locality axiom \eqref{e1.11}.
\end{itemize}
\end{defn}

Remark \ref{rem1.1} demonstrates that a vertex algebra defined in the spirit of differential algebra is a vertex algebra defined in the spirit of quantum field theory. However, in order to prove the converse, one has to show that axiom (ii) of Definition \ref{def1.1} holds. This will follow from the proof of the Extension theorem in Lecture 3. 
\begin{defn}
\label{def1.5}
Given a vertex algebra $V$, the map of the space of its quantum fields to $V$, defined by
\begin{equation}
\mathit{fs}\colon a(z)\mapsto a(z)\vac_{z=0}=a_{(-1)}\vac=a,
\label{1.12}
\end{equation}
is called the \emph{field-state correspondence}. This map is obviously surjective. If this map is also injective, then the inverse map
\begin{equation}
\mathit{sf}\colon a\mapsto a(z)
\label{1.13}
\end{equation}
is called the \emph{state-field correspondence}.
\end{defn}
The first fundamental theorem, which allows one to construct non-commutative
vertex algebras, is the so-called Extension theorem. 
\begin{thm} (Extension theorem).
\label{Th1.1}
Let $V$ be a vector space, $\vac\in V$ a non-zero vector, $T\in\End V$ and 
\begin{equation}
\mc{F}=\bigg\{a^j(z)=\sum_{n\in\Z}a^j_{(n)}z^{-n-1}\bigg\}_{j\in J}
\label{1.14}
\end{equation}
a collection of $\End V$-valued quantum fields indexed by a set $J$. Suppose that the following properties hold:
\begin{itemize}
\item[(i)](vacuum axiom) $T\vac=0$,
\item[(ii)](translation covariance) $[T,a^j(z)]=\partial_z a^j(z)$ for all $j\in J$,
\item[(iii)](locality) $(z-w)^{N_{ij}}[a^i(z),a^j(w)]=0$ for all $i,j\in J$
with some $N_{ij}\in\Z_+$,
\item[(iv)](completeness) $\Span\{a^{j_1}_{(n_1)}\cdots a^{j_s}_{(n_s)}\vac\mid
j_i\in J,\ n_i\in\Z,\ s\in\Z_+\}=V$.
\end{itemize}
Let $\mc{F}_{\max}$ denote the set of all translation covariant quantum fields $a(z)$
such that $a(z)$, $a^j(z)$ is a local pair for all $j\in J$. Then the map
\begin{equation}
\mathit{fs}\colon\mc{F}_{\max}\to V,\quad a(z) 
\mapsto a(z) \vac_{z=0}  
\label{1.15}
\end{equation}
is bijective and the inverse map $\mathit{sf}\colon V\to\mc{F}_{\max}$ 
 endows $V$ with a structure of a vertex algebra (in the sense of Definition 
\ref{def1.4}) 
with vacuum vector $\vac$ and translation operator $T$.
\end{thm}
\begin{rmk}
By conditions (ii) and (iii) 
we have
$\mc{F}\subset\mc{F}_{\max}$, hence the name Extension theorem.  
\end{rmk}
\medskip\noindent
Some historical remarks:
%
\begin{itemize}
\item Vertex algebras first appeared implicitly in the paper of Belavin, Polyakov, Zamolodchikov
\cite{BPZ} in 1984.
\item The first definition of vertex algebras was given by Borcherds \cite{B1}
in  1986.
\item The Extension Theorem was proved in \cite{DSK06}.
In \cite{VAB} a weaker version was given.
\item Connection to physics (Wightman axioms of a quantum field theory \cite{W}
in the 1950's) is discussed e.g. in \cite{VAB}.
\end{itemize}
\begin{rmk}[Super version]
A vertex superalgebra $V$ is a local unital $z$-superalgebra $V$,
cf. Definition \ref{def1.3}. Namely $V$ is a vector superspace
\begin{equation}
V=V_{\bar 0}\oplus V_{\bar 1}, \,\,\{\bar 0,\bar 1\}=\Z/2\Z,
\label{1.16}
\end{equation}
 $a(z)b\in V_{\alpha +\beta}((z))$ if $a\in V_\alpha , b\in V_\beta$,
and $TV_{\alpha}\subset V_\alpha$, $\alpha, \beta\in\Z/2\Z$.
An element $a\in V$ has \emph{parity} $p(a)=\alpha$ if $a\in V_\alpha$.
Finally, the locality axiom (\ref{e1.11}) is written as $(z-w)^{N_{ab}}[a(z),b(w)]=0$, where the commutator is understood in the ``super'' sense, i.e. 
\[[a(z),b(w)]=a(z) b(w)
-(-1)^{p(a)p(b)}b(w)a(z).\]
All the identities in the ``super'' case are obtained from the respective identities in the purely even case by the Koszul-Quillen rule: there is a sign change if the order of two odd elements is reversed; no change otherwise.
It is a general convention to drop the adjective ``super'' in the case of
vertex superalgebras.
\end{rmk}

\subsection{Calculus of formal distributions}
\label{subsec:1.2}

\begin{defn}
Let $U$ be a vector space. A \emph{$U$-valued formal distribution} $a(z)$ is an
element of ${U}[[z,z^{-1}]]$:
\begin{equation}
a(z)=\sum_{n\in\Z}a_{n}z^{n},\quad a_{n}\in U.
\label{1.17}
\end{equation}
The \emph{residue} of $a(z)$ is
\begin{equation}
\Res a(z)dz=a_{-1}.
\label{1.18}
\end{equation}
Most often one uses a different indexing of coefficients:
\begin{equation}
a(z)=\sum_{n\in\Z}a_{(n)}z^{-n-1},\quad\text{so that}\quad a_{(n)}=\Res a(z)z^ndz.
\label{1.19}
\end{equation}
\end{defn}

Note that $a(z)$ is a linear function on the space of test functions $\F[z,z^{-1}]$:
\begin{equation}
\langle a(z),\varphi(z)\rangle=\Res a(z)\varphi(z)dz\in{U},\quad
\forall\varphi(z)\in
\F[z,z^{-1}],
\label{1.20}
\end{equation}
and it is easy to see that one thus gets all linear functions on 
$\F[z,z^{-1}]$.

A formal distribution in two variables $z$ and $w$ is an element $a(z,w)\in U[[z,z^{-1},w,w^{-1}]]$.
\begin{defn}
A formal distribution $a(z,w)$ is called \emph{local} if $(z-w)^{N}a(z,w)=0$ for some $N\in\Z_+$.
\end{defn}
\begin{exm}
The \emph{formal delta function} $\delta(z,w)$, defined by
\begin{equation}
\delta(z,w)=\sum_{n\in\Z}z^{-n-1}w^{n},
\label{1.21}
\end{equation}
is an example of an $\F$-valued formal distribution in two variables.
It is local since $(z-w)\delta(z,w)=0$. In fact, one can write $\delta(z,w)$ as
\begin{equation}
\delta(z,w)=i_{z,w}\frac{1}{z-w}-i_{w,z}\frac{1}{z-w},
\label{1.22}
\end{equation}
where $i_{z,w}$ denotes the expansion in the domain $|z|>|w|$ and $i_{w,z}$ denotes
the expansion in the domain $|z|<|w|$, i.e.
\begin{equation}
i_{z,w}\frac{1}{z-w}=z^{-1}\frac{1}{1-\dfrac{w}{z}}=\sum_{n\geq 0}z^{-n-1}w^n,
\quad\text{and}\quad
i_{w,z}\frac{1}{z-w}=-w^{-1}\frac{1}{1-\dfrac{z}{w}}=-\sum_{n<0}z^{-n-1} w^n.
\label{1.23}
\end{equation}
\end{exm}

\noindent
The following formula, which is derived by differentiating 
\eqref{1.21} and \eqref{1.22} $n\in \Z_+$ times, 
will be useful:
\begin{equation}
\frac{\partial^n_w\delta(z,w)}{n!}=i_{z,w}\frac{1}{(z-w)^{n+1}}
-i_{w,z}\frac{1}{(z-w)^{n+1}}=\sum_{j\in\Z}{j\choose n}w^{j-n}z^{-j-1}.
\label{1.24}
\end{equation}

\noindent
Let us list some properties of the formal delta function, which are straightforward by \eqref{1.24}:
\begin{itemize}
\item[(1)] $(z-w)^m\dfrac{\partial_w^n\delta(z,w)}{n!}
=\begin{cases}\dfrac{\partial^{n-m}_w \delta(z,w)}{(n-m)!}&\text{if}\ n\geq m\geq 0,
\\0,&\text{if}\ m>n,\end{cases}$
\item[(2)] $\delta(z,w)=\delta(w,z)$,
\item[(3)] $\partial_z \delta(z,w)=-\partial_w \delta(z,w)$,
\item[(4)] $a(z)\delta(z,w)=a(w)\delta(z,w)$, where $a(z)$ is any formal distribution,
\item[(5)] $\Res a(z)\delta(z,w)dz=a(w)$.
\end{itemize}

\begin{thm}[Decomposition theorem]
Any local formal distribution $a(z,w)$ can be uniquely decomposed as a finite 
sum
of derivatives of the formal delta function with formal distributions
in $w$ as coefficients:
\begin{equation}
a(z,w)=\sum_{j\geq 0}c^j(w)\frac{\partial^j_w \delta(z,w)}{j!}.
\label{1.25}
\end{equation}
Moreover
\begin{equation}
c^j(w) = \Res a(z,w)(z-w)^jdz.
\label{1.26}
\end{equation}
\end{thm}
\begin{proof}
Multiply both sides of \eqref{1.25} by $(z-w)^j$ and take residues. Using properties of the delta function listed above we obtain \eqref{1.26}. To show \eqref{1.25} we set
\begin{equation}
b(z,w)=a(z,w)-\sum_{j\geq 0}c^j(w)\frac{\partial^j_w \delta(z,w)}{j!}
\label{1.26+}
\end{equation}
with $c^j(w)$ given by \eqref{1.26}. It is immediate that
\begin{equation}
\Res b(z,w)(z-w)^jdz=0 \quad\text{for all}\quad j\in\Z_+ ,
\label{1.26++}
\end{equation}
hence
\begin{equation}
b(z,w)=\sum_{n\geq 0}b_n(w)z^n .
\label{1.26+++}
\end{equation}
By definition $b(z,w)$ is local, therefore \eqref{1.26+++} implies that $b(z,w)=0$.
\end{proof}

\begin{rmk}
If we have a local pair $a(z),b(z)\in\mf{g}[[z,z^{-1}]]$, where $\mf{g}$ is a Lie (super)algebra (i.e. 
$[a(z), b(w)]$ is a local formal distribution in $z$ and $w$), 
then, by the Decomposition theorem, we have:
\begin{equation}
[a(z),b(w)]=\sum_{j\geq 0}(a(w)_{(j)} b(w))\frac{\partial^j_w\delta(z,w)}{j!},
\label{1.27}
\end{equation}
where the sum over $j$ is finite, and
\begin{equation}
\mf{g}[[w,w^{-1}]]\ni a(w)_{(j)}b(w):=\Res(z-w)^j [a(z),b(w)]dz\;\;(=c^j(w)).
\label{1.28}
\end{equation}
Using \eqref{1.24} and comparing the coefficients of $z^m w^n$ on both sides 
of \eqref{1.27}, we find
\begin{equation}
[a_{(m)},b_{(n)}]=\sum_{j\geq 0}{m\choose j}(a_{(j)}b)_{(m+n-j)},\quad \forall m,n\in\Z.
\label{1.29}
\end{equation}
\end{rmk}

\subsection{Free boson and free fermion vertex algebras}
\label{subsec:1.3}

\begin{exm}[Free boson] \label{Ex:Free-Boson}
Let $B=\F[x_1,x_2,\ldots]$, $\vac=1$,
$T=\sum_{j\geq 2}jx_j\frac{\partial}{\partial x_{j-1}}$ and 
\begin{equation}
\mc{F}=\bigg\{a(z)=\sum_{n\in\Z}a_{(n)}z^{-n-1}\bigg\},\quad\text{where}\quad
a_{(n)}=\begin{cases}\dfrac{\partial}{\partial x_n},&\text{if}\ n>0,\\
-nx_{-n},&\text{if}\ n<0,\\[4pt]
0,&\text{if}\ n=0,
\end{cases}
\label{1.30}
\end{equation}
so that
\begin{equation}
[a_{(m)},a_{(n)}]=m\delta_{m,-n},\quad \forall m,n\in\Z.
\label{1.31}
\end{equation}
The quantum
field $a(z)$ is called the \emph{free boson} field. Since \eqref{1.31} is equivalent to
\begin{equation}
[a(z),a(w)]=\partial_w\delta(z,w),
\label{1.32}
\end{equation}
we have
\begin{equation}
(z-w)^2[a(z),a(w)]=0,
\label{1.33}
\end{equation}
i.e., $a(z)$ is local with itself.

The translation covariance of the free boson field $a(z)$, that is $[T,a_{(n)}]=-na_{(n-1)}$, $\forall n\in\Z$,
can be verified directly. Vacuum axiom and completeness are obviously satisfied.
Locality is \eqref{1.33}. So, by the Extension theorem, $B$ carries a vertex algebra
structure.
\end{exm}

\begin{exm}[Free fermion]
Let $F=\Lambda[\xi_1,\xi_2,\ldots]$ be a Grassmann superalgebra, i.e.,
\[\xi_i\xi_j=-\xi_j\xi_i,\,\,\, p(\xi_i)=\bar{1} \, . \] 
Let $\vac=1$ and
$T=\sum_{j\geq 1}j\xi_{j+1}\frac{\partial}{\partial\xi_j}$, where 
$\frac{\partial}{\partial\xi_j}$
is an odd derivation of the superalgebra $F$
(i.e. $\frac{\partial (ab)}{\partial\xi_j}=
\frac{\partial a}{\partial\xi_j}b+(-1)^{p(a)}\frac{\partial b}{\partial\xi_j}$), such that
\begin{equation}
\frac{\partial}{\partial\xi_j}\xi_i=\delta_{ij}.
\label{1.35}
\end{equation}
Set
\begin{equation}
\mc{F}=\bigg\{\varphi(z)=\sum_{n\in\Z}\varphi_{(n)}z^{-n-1}\bigg\},
\quad\text{where}\quad
\varphi_{(n)}=\begin{cases}\dfrac{\partial}{\partial\xi_{n+1}},&\text{if}\ n\geq 0,\\
\xi_{-n},&\text{if}\ n<0,
\end{cases} 
\label{1.36}
\end{equation}
then
\begin{equation}
[\varphi_{(m)},\varphi_{(n)}]=\delta_{m,-n-1},\quad \forall m,n\in\Z.
\label{1.37}
\end{equation}
The odd quantum field  $\varphi(z)$ is called the \emph{free fermion} field. Since \eqref{1.37} is equivalent to
\begin{equation}
[\varphi(z),\varphi(w)]=\sum_{n\in\Z}z^{-n-1}w^n=\delta(z,w),
\label{1.38}
\end{equation}
$\varphi(z)$ is local to itself. As in Example \ref{Ex:Free-Boson}, the 
vacuum axiom 
and completeness are immediate. Translation covariance follows from the
exercise below.  
\end{exm}

\begin{exr}
Show that the free fermion field is translation covariant, i.e.,
\begin{equation}
[T,\varphi_{(n)}]=-n\varphi_{(n-1)},\quad \forall n\in\Z.
\label{1.39}
\end{equation}
\end{exr}

\newpage

\section{Lecture 2 (December 11, 2014)}
\label{sec:2}

In the first lecture we discussed the two simplest examples of non-commutative
vertex algebras (see Examples 1.3 and 1.4). In this lecture we will consider further
important examples, among them a generalization of those two mentioned previously.
First, we need to introduce the necessary notions.

\subsection{Formal distribution Lie algebras and their universal vertex algebras.}
\label{subsec:2.1}

\begin{defn}
A \emph{formal distribution Lie (super)algebra} is a pair $(\mf{g},\mc{F})$,
where $\mf{g}$ is a Lie (super)algebra and $\mc{F}$ is a collection of pairwise
local $\mf{g}$-valued formal distributions $a^j(z)=\sum_{n\in\Z}a^j_{(n)}z^{-n-1}$,
$j\in J$, such that the coefficients $\{a^j_{(n)}\mid j\in J,\ n\in\Z\}$ span
$\mf{g}$. A formal distribution Lie (super)algebra $(\mf{g},\mc{F})$ is called
\emph{regular} if:
\begin{itemize}
\item[(i)] the $\bb{F}[\partial_z]$-span of $\mc{F}$ is closed under all
\emph{$n$-th products} for $ n\in\Z_+$,
\begin{equation}
a(z)_{(n)}b(z):=\Res(w-z)^n[a(w),b(z)]dw,
\label{2.1}
\end{equation}
i.e., if $a(z)$ and $b(z)$ are elements of the form $\sum_{j\in J} f_j(\partial_z)a^j(z)$,
where $f_j(\partial_z)\in \F[\partial_z]$, and only finitely many $f_j(\partial_z)\neq 0$, then their $n$-th product for $n\in \Z_+$ is still an element of the same form.
\item[(ii)] there exists a derivation $T\in\Der\,\mf{g}$ such that
\begin{equation}
T(a^j(z))=\partial_za^j(z),\ \text{i.e.},\ T(a^j_{(n)})=-na^j_{(n-1)}
\quad\forall j\in J.
\label{2.2}
\end{equation}
\end{itemize}
The \emph{annihilation subalgebra} of $\mf{g}$ is
$\mf{g}_-=\Span\{a^j_{(n)}\mid j\in J,\ n\in\Z_+\}$.
\end{defn}

\begin{exr}
Show that $\mf{g}_-$ is a $T$-invariant subalgebra of $\mf{g}$.
(\emph{Hint}: use the commutation formulas \eqref{1.29} and \eqref{2.2}.)
\end{exr}

The following theorem allows one to construct vertex algebras
via the Extension theorem. Let $U(\mf{g})$ denote the universal enveloping algebra of $\mf{g}$.

\begin{thm}
Let $(\mf{g},\mc{F}_0)$ be a regular formal distribution Lie algebra,
and let $\mf{g}_-$ be the annihilation subalgebra. Let
$V=U(\mf{g})/U(\mf{g})\mf{g}_-$ (also known as the induced $\mf{g}$-module 
$\mathrm{Ind}_{\mf{g}_-}^{\mf{g}}(\C)$) and let $\pi$ be the representation
of $\mf{g}$ in $V$ induced via the left multiplication. Let $\vac=\bar{1}$ be the image of $1$ in $V$  and
$T\in\End V$ be the endomorphism of $V$ induced by the derivation of $\mf{g}$.
Let $\mc{F}$ be the collection of $\End V$-valued formal distributions
\begin{equation}
\mc{F}=\bigg\{\pi\big(a^j(z)\big)=\sum_{n\in\Z}\pi\big(a^j_{(n)}\big)z^{-n-1}
\bigg| a^j(z)\in\mc{F}_0,\ j\in J
\bigg\}.
\label{2.3}
\end{equation}
Then $\mc{F}$ consists of quantum fields and $(V,\vac,T,\mc{F})$ satisfies the
conditions of the Extension theorem, hence $V$ is a vertex algebra, which we
denote by $V(\mf{g},\mc{F}_0)$.
\end{thm}
\begin{proof}
The only non-obvious part is to check that all $\pi\big(a^j(z)\big)$ are
quantum fields, i.e., $\pi\big(a^j(z)\big)v\in V((z))$ for each $v\in V$.
Due to the $PBW$ theorem, it is  sufficient to check it for vectors of the following form (we use the same notation for elements in $U(\mf{g})$ and their images in $V$):
\begin{equation}
v=a^{j_1}_{(n_1)}\cdots a^{j_s}_{(n_s)}\vac,\, \hbox{where}\,\, 
j_1,\ldots,j_s\in J. 
\label{2.4}
\end{equation}
We argue by induction on $s$. For $s=0$ we have $v=\vac$, hence
\begin{equation}
\pi\big(a^j(z)\big)\vac=\sum_{n\in\Z}\pi\big(a^j_{(n)}\big)z^{-n-1}\vac
=\sum_{n<0}\pi\big(a^j_{(n)}\big)z^{-n-1}\in V[[z]].
\label{2.5}
\end{equation}
The last equality follows from the fact that $a^j_{(n)}\vac =0$ for $n\geq 0$. We proceed by proving the induction step:
\begin{equation}
\label{quantumfield}
\pi\big(a^j(z)\big)a^{j_1}_{(n_1)}\cdots a^{j_s}_{(n_s)}\vac=[a^j(z), a^{j_1}_{(n_1)}]a^{j_2}_{(n_2)}\cdots a^{j_s}_{(n_s)}\vac+ a^{j_1}_{(n_1)}a^j(z)a^{j_2}_{(n_2)}\cdots a^{j_s}_{(n_s)}\vac.
\end{equation}
By assumption of induction, the second term in the right-hand side is in $V((z))$, so we only need to show that the first term is also in  $V((z))$. Now recall the commutation formula \eqref{1.29}. We have 
\begin{equation}
\label{commu}
[a^j(z),a^{j_1}_{(n_1)}]=\sum_{m\in \Z}\sum_{k\geq 0}{m\choose k}(a^j_{(k)}a^{j_1})_{(m+n_1-k)}z^{-m-1},
\end{equation}
where $(a^j_{(k)}a^{j_1})_{(m+n_1-k)}$ is the Fourier coefficient of the formal distribution $a^j(z)_{(k)}a^{j_1}(z)$. By the regularity property, we know that $a^j(z)_{(k)}a^{j_1}(z)$ is contained in the $\bb{F}[\partial_z]$-span of $\mc{F}$, thus we can assume that 
\begin{equation}
a^j(z)_{(k)}a^{j_1}(z)= \sum_{l\in J} f_l^k(\partial_z)a^l(z).
\end{equation}
Since  $a^j(z), a^{j_1}(z)$ is a local pair, we know that there exists an integer $N\in \Z_+$ such that $a^j(z)_{(k)}a^{j_1}(z)=0$ for $k\geq N$. This allows us to rewrite formula \eqref{commu} as follows,
\begin{equation}
[a^j(z),a^{j_1}_{(n_1)}]=\sum_{0\leq k\leq N} \sum_{m\in \Z}{m \choose k}\big(\sum_{l\in J} f_l^k(\partial_z)a^l(z)\big)_{(m+n_1-k)}z^{-m-1}.
\end{equation}
By assumption of induction, for each $k$, 
\begin{equation}
 \sum_{m\in \Z}\big(\sum_{l\in J} f_l^k(\partial_z)a^l(z)\big)_{(m+n_1-k)}z^{-m-1}a^{j_2}_{(n_2)}\cdots a^{j_s}_{(n_s)}\vac \in V((z))
\end{equation}
thus the first term in the right-hand side of \eqref{quantumfield} is also in $V((z))$.  
\end{proof}

\begin{rmk}
Recall that by the Decomposition theorem for any local pair $a(z),b(w)$ we have
\begin{equation}
[a(z),b(w)]=\sum_{j\geq 0}(a(w)_{(j)} b(w))\frac{\partial^j_w\delta(z,w)}{j!},
\label{2.5+}
\end{equation}
which is equivalent to the commutator formula
\begin{equation}
[a_{(m)},b_{(n)}]=\sum_{j\geq 0}{m\choose j}(a_{(j)}b)_{(m+n-j)},\quad \forall m,n\in\Z,
\label{2.5++}
\end{equation}
where $ (a_{(j)} b)(w) = a(w)_{(j)} b(w) $ is given by (\ref{1.28}). 
This, along with the obvious formula
\begin{equation}
\big(\partial_wa(w)\big)_{(n)}=-na(w)_{(n-1)},
\label{2.6}
\end{equation}
allows us to convert the commutator formula into the decomposition formula,
thereby establishing locality.
\end{rmk}

Let us now discuss the next important example of a non-commutative vertex algebra.
\begin{exm} 
\label{ex2.1}
Let $\mf{g}=\Vir$ be the Virasoro algebra with commutation relations
\begin{equation}
[L_m,L_n]=(m-n)L_{m+n}+\delta_{m,-n}\frac{m^3-m}{12}C,\quad [C,L_m]=0,\quad \forall m,n\in\Z.
\label{2.7}
\end{equation}
Consider the formal distribution
\begin{equation}
L(z)=\sum_{n\in\Z}L_nz^{-n-2},
\label{2.8}
\end{equation}
so that $L_{(n)}=L_{n-1}$. Then the commutation relations \eqref{2.7} can be written
in the equivalent form
\begin{equation}
[L(z),L(w)]=\partial_wL(w)\delta(z,w)+2L(w)\partial_w\delta(z,w)
+\frac{C}{2}\partial_w^3\delta(z,w).
\label{2.9}
\end{equation}
Indeed, by \eqref{2.9} we have: $ L_{(0)}L = \partial L, \  L_{(1)}L = 2 L, \ L_{(3)}L = \tfrac{C}{2},  $ and  $ L_{(j)}L = 0 $ for all other $ j \geq 0. $ Hence by \eqref{2.5++} and \eqref{2.6}, \eqref{2.9} is equivalent to \eqref{2.7}.
It follows that $L(z)$ is local with itself, hence $(\Vir,\{L(z),C\})$ is a formal distribution Lie algebra .
Furthermore, 
it is regular. There are two conditions (i) and (ii) we need to
check: (i) is obvious, for (ii) take
$T=\ad L_{-1}$, then $[L_{-1},L_n]=(-1-n)L_{n-1}$, which gives \eqref{2.2}.
The annihilation subalgebra is
\begin{equation}
\Vir_-=\sum_{n\geq -1}\F L_n.
\label{2.11}
\end{equation}
So, by Theorem 2.1 and the Extension theorem, we get the associated vertex algebra
\begin{equation}
V(\Vir,\{L(z),C\}),
\label{2.12}
\end{equation}
called the \emph{universal Virasoro vertex algebra}. One can make it slightly smaller
by taking $c\in\F$ and factorizing by the ideal generated by $(C-c)$.
Let $V^c$ stand for the corresponding factor vertex algebra, which is called the
\emph{universal Virasoro vertex algebra with central charge $c$}.
\end{exm}
\begin{rmk}
$V^c$ can be non-simple for certain values of $c$. Namely, $V^c$ is non-simple if and only if \cite{GK07}
\begin{equation}
c=1-\frac{6(p-q)^2}{pq},\quad\text{with}\quad p,q\in\Z_{\geq 2}\,\text{ coprime}.
\label{2.13}
\end{equation}
\begin{exr}
The vertex algebra $ V^c $ has a unique maximal ideal $ J^c $.
\end{exr}
Let $ V_c = V^c / J^c $. Since $c$ in \eqref{2.13} is symmetric in $p$ and $q$ we may assume that $p<q$.
The smallest example $p=2$, $q=3$ gives $c=0$; {\blu $V_0$ is the one-dimensional vertex
algebra.}
The next example is $p=3$, $q=4$ when $c=1/2$; the vertex algebra $ V_{\frac{1}{2}} $ is related to the Ising model.
The simple vertex algebras $V_c$ with $c$ of the form \eqref{2.13} are called \emph{discrete series vertex algebras}. They play a fundamental role in conformal field theory 
\cite{BPZ}.
\end{rmk}

\begin{exm}	
\label{ex2.2}
Let $\mf{g}$ be a finite dimensional Lie algebra with a non-degenerate symmetric invariant
bilinear form $(.\vert.)$. Let
$\widehat{\mf{g}}=\mf{g}[t,t^{-1}]+\F K$ be the associated Kac-Moody affinization, with commutation relations
\begin{equation}
[at^m,bt^n]=[a,b]t^{m+n}+m\delta_{m,-n}(a\vert b)K,\quad [K,at^m]=0,
\label{2.14}
\end{equation}
where $a,b\in\mf{g}$, $m,n\in\Z$. Let $a(z)=\sum_{n\in\Z}(at^n)z^{-n-1}$
and $\mc{F}=\{a(z)\}_{a\in\mf{g}}\cup\{K\}$ be an (infinite) collection of formal
distributions. The commutation relations \eqref{2.14} are equivalent to
\begin{equation}
[a(z),b(w)]=[a,b](w)\delta(z,w)+(a\vert b)\partial_w\delta(z,w)K,\quad
[K, a(z)] =0.
\label{2.15}
\end{equation}
Hence $\mc{F}$ is a local family. So $(\widehat{\mf{g}},\mc{F})$ is a formal distribution Lie algebra.  The annihilation subalgebra is 
$\widehat{\mf{g}}_- =\mf{g} [t]$.
\begin{exr}
Show that the formal distribution Lie algebra $(\widehat{\mf{g}},\mc{F})$ defined above is regular with $T=-\partial_t$.
\end{exr}
The associated vertex algebra $V(\widehat{\mf{g}},\mc{F})$ is called the \emph{universal
affine vertex algebra associated to $\big(\mf{g},(.\vert.)\big)$}. Again,
it can be made a little smaller by taking $k\in\F$ and considering
\begin{equation}
V^k(\mf{g})=V(\widehat{\mf{g}},\mc{F})/(K-k)V(\widehat{\mf{g}},\mc{F}),
\label{2.16}
\end{equation}
which is called the \emph{universal affine vertex algebra of level $k$}. There are certain
values of $k$ for which   $V^k(\mf{g})$ is non-simple (it is a known set of rational numbers \cite{GK07}). 

\end{exm}

\begin{exm} 
Let $A$ be a finite dimensional vector superspace with a non-degenerate skewsymmetric
bilinear form $\langle.\vert.\rangle$:
\begin{equation}
\langle a\vert b\rangle=-(-1)^{p(a)p(b)}\langle b\vert a\rangle,\quad a,b\in A.
\label{2.18}
\end{equation}
Take the associated Clifford affinization
\begin{equation}
\hat A=A[t,t^{-1}]+\F K,
\label{2.19}
\end{equation}
with commutation relations
\begin{equation}
[at^m,bt^n]=\delta_{m,-n-1}\langle a\vert b\rangle K,\quad [K,at^m]=0, \ a, b \in A.
\label{2.20}
\end{equation}
Consider the formal distributions
\begin{equation}
a(z)=\sum_{n\in\Z}(at^n)z^{-n-1},\quad a\in A ,
\label{2.21}
\end{equation}
and define $\mc{F}$ to be
\begin{equation}
\mc{F}=\{a(z)\}_{a\in A}\cup\{K\}.
\label{2.22}
\end{equation}
Then the commutation relations \eqref{2.20} are equivalent to
\begin{equation}
[a(z),b(w)]=\langle a\vert b\rangle\delta(z,w)K, \ [K, a(z)] = 0.
\label{2.23}
\end{equation}
Hence $\mc{F}$ is a local 
family, and  $(\hat A,\mc{F})$ is a formal distribution Lie superalgebra. Its annihilation subalgebra is $\hat{A}_-=A[t]$.  
Furthermore, 
$(\hat A,\mc{F})$ 
is regular with
$T=-\partial_t$ and
\begin{equation}
F(A)=V(\hat A,\mc{F})/(K-1)V(\hat A,\mc{F})
\label{2.24}
\end{equation}
is the associated vertex algebra called the \emph{vertex algebra of free
superfermions}. 
\end{exm}
\begin{exr}
\begin{itemize}
\item[(1)]If $A$ is a $1$-dimensional odd superspace we get the free fermion vertex algebra $F= F(A)$.
\item[(2)]If $\mf{g}$ is the $1$-dimensional Lie algebra $\F$, with bilinear form $(a\vert b)=ab$ and level $k=1$, then we get the free boson vertex algebra $B = V^1(\mathbb{F})$.
\end{itemize}
\end{exr}

\begin{exr}
Show that the vertex algebra $F(A)$ is always simple. 
\end{exr}


\subsection{Formal Cauchy formulas and normally ordered product}
\label{subsec:2.2}

We proceed by proving some statements which are analogous to the Cauchy
formula and are true for any formal distribution. Let $U$ be a vector space
and $a(z)=\sum_{n\in\Z}a_{(n)}z^{-n-1}$ be a $U$-valued formal
distribution. We call
\begin{equation}
a(z)_+=\sum_{n<0}a_{(n)}z^{-n-1}
\label{2.26}
\end{equation}
the \emph{creation part} or \emph{``positive'' part} of $a(z)$ and
\begin{equation}
a(z)_-=\sum_{n\geq 0}a_{(n)}z^{-n-1}
\label{2.27}
\end{equation}
the \emph{annihilation part} or \emph{``negative'' part}  of $a(z)$. Note that
$\partial_z\big(a(z)_\pm\big)=\big(\partial_za(z)\big)_\pm$.

\begin{prp}
Formal Cauchy formulas can be written as follows:
\begin{itemize}
\item[(a)] For the ``positive'' and ``negative'' parts of $a(z)$ we have
\begin{equation} \label{2.28}
a(w)_+=\Res a(z)i_{z,w}\frac{1}{z-w}dz, \qquad -a(w)_-=\Res a(z)i_{w,z}\frac{1}{z-w}dz. 
\end{equation}
\item[(b)] For the derivatives of $a(z)_\pm$ we have
\begin{equation} \label{2.30}
\frac{1}{n!}\,\partial_w^na(w)_+=\Res a(z)i_{z,w}\frac{1}{(z-w)^{n+1}}dz, \qquad -\frac{1}{n!}\,\partial_w^na(w)_- =\Res a(z)i_{w,z}\frac{1}{(z-w)^{n+1}}dz. 
\end{equation}
\end{itemize}
\end{prp} 

\begin{proof}
Use property (5) of the delta function and \eqref{1.22} to get
\begin{equation}
a(w)=\Res a(z)\delta(z,w)dz=\Res a(z)\bigg(i_{z,w}\frac{1}{z-w}
-i_{w,z}\frac{1}{z-w}\bigg)dz.
\label{2.32}
\end{equation}
Collect the (non-negative) powers of $w$ on both sides to get (a).
Differentiating (a) by $w$ $n$ times gives (b).
\end{proof}

Multiplying two quantum fields na\"ively would lead to divergences.
The next definition is introduced to circumvent this problem.

\begin{defn}
The \emph{normally ordered product} of $\End V$-valued quantum fields $a(z)$ and $b(z)$
is defined by
\begin{equation}
:a(z)b(z):\,=a(z)_+b(z)+(-1)^{p(a)p(b)}b(z)a(z)_-.
\label{2.33}
\end{equation}
\end{defn}

It must be proved that $:a(z)b(z):$ is an ``honest'' quantum field, i.e., all the 
divergences are removed.

\begin{prp}
If $a(z)$ and $b(z)$ are quantum fields then so is $:a(z)b(z):$.
\end{prp}
\begin{proof}
Apply $:a(z)b(z):$, defined by \eqref{2.33}, to any vector $v\in V$:
\begin{equation}
:a(z)b(z):v=a(z)_+b(z)v+(-1)^{p(a)p(b)}b(z)a(z)_-v.
\label{2.34}
\end{equation}
Since $b(z)$ is assumed to be a quantum field, $b(z)v$ in the first term of the
right-hand side of \eqref{2.34} is a Laurent series by definition. The creation part $a(z)_+$ has
only non-negative powers of $z$, therefore $a(z)_+b(z)v$ is still a Laurent series.
In the second term $a(z)_-v$ consists of finitely many terms with negative powers,
i.e., it is a Laurent polynomial. Now $b(z)a(z)_-v$ is a Laurent series multiplied
by a Laurent polynomial which is still a Laurent series. Hence we proved that
$:a(z)b(z):$ is a sum (or a difference) of two Laurent series, thus it is a Laurent
series.
\end{proof}

\begin{exr}
Let $a(z)$ and $b(z)$ be quantum fields. Show that their $n$-th product $a(z)_{(n)}b(z)$,
$n\in\Z_+$ and derivatives $\partial_za(z),\partial_zb(z)$ are also quantum
fields.
\end{exr}

On the space of quantum fields we have defined $a(w)_{(n)}b(w)$ for $n\geq 0$.
Introduce
\begin{equation}
a(w)_{(-n-1)}b(w)=\frac{1}{n!}:\partial_w^na(w) b(w):,
\label{2.35}
\end{equation}
so that $a(w)_{(-1)}b(w)=\,:a(w)b(w):$. Thus for each $n\in\Z$ we have the $ n$-th product $a(w)_{(n)}b(w)$. Using the formal Cauchy formulas above,  we get the \emph{unified formulas for all $n$-th products} of
quantum fields
\begin{equation}
a(w)_{(n)}b(w)=\Res\big(a(z)b(w)i_{z,w}(z-w)^n
-(-1)^{p(a)p(b)}b(w)a(z)i_{w,z}(z-w)^n\big)dz,\quad
n\in\Z.
\label{2.36}
\end{equation}

\begin{rmk}
For a local pair of quantum fields physicists write
\begin{equation}
a(z)b(w)=\sum_{n\in\Z}\frac{a(w)_{(n)}b(w)}{(z-w)^{n+1}}.
\label{2.37}
\end{equation}
This way of writing is useful but might be confusing, since different parts
of it are expanded in different domains. Therefore it is worth giving a
rigorous interpretation of \eqref{2.37} by writing
\begin{equation}
a(z)b(w)=\sum_{n\geq 0}a(w)_{(n)}b(w)\,i_{z,w}\frac{1}{(z-w)^{n+1}}\,+:a(z)b(w):
\label{2.38}
\end{equation}
and
\begin{equation}
(-1)^{p(a)p(b)}
b(w)a(z)=\sum_{n\geq 0}a(w)_{(n)}b(w)\,i_{w,z}\frac{1}{(z-w)^{n+1}}\,+:a(z)b(w):
\label{2.39}
\end{equation}
By taking the difference \eqref{2.38}--\eqref{2.39} we get
\begin{equation}
[a(z),b(w)]=\sum_{j\in\Z_+}\big(a(w)_{(j)}b(w)\big)\frac{\partial_w^j\delta(z,w)}{j!}.
\label{2.40}
\end{equation}
Conversely, by separating the negative (nesp. non-negative) powers of $z$ in \eqref{2.40}
we get \eqref{2.38} (resp. \eqref{2.39}). We still need to explain \eqref{2.37} for negative $ n $. By Taylor's formula in the domain $|z-w|<|w|$ (\cite{VAB}, (2.4.3)), we have 
\begin{equation}
:a(z)b(w):\,=\sum_{n\geq 0}:\partial_w^na(w)b(w):\frac{(z-w)^n}{n!}
=\sum_{n\geq 0}\big(a(w)_{(-n-1)}b(w)\big)(z-w)^n ,
\label{2.41}
\end{equation}
i.e., the 
$n$-th products for negative $n$ are ``contained'' in the normally ordered product. 
\end{rmk}


\subsection{Bakalov's formula and Dong's lemma.}
\label{subsec:2.4}

Locality of the pair $a(z),b(z)$ of $\End V$-valued quantum fields means that
\begin{equation}
(z-w)^Na(z)b(w)=(-1)^{p(a)p(b)}(z-w)^Nb(w)a(z)\,\,\hbox{for some}\,\,N\in\Z_+. 
\label{2.42}
\end{equation}
Denote either side of this equality by $F(z,w)$.
Then for each $k\in\Z_+$ we have
\begin{equation}
\Res F(z,w)\frac{\partial_w^k\delta(z,w)}{k!}dz
=\Res F(z,w)\,i_{z,w}\frac{1}{(z-w)^{k+1}}dz
-\Res F(z,w)\,i_{w,z}\frac{1}{(z-w)^{k+1}}dz.
\label{2.43}
\end{equation}
The first term of the left-hand side of \eqref{2.43} is
\begin{equation}
\Res a(z)b(w)\,i_{z,w}(z-w)^{N-k-1}dz,
\label{2.44}
\end{equation}
while the second term of the right-hand side of \eqref{2.43} is
\begin{equation}
-\Res a(z)b(w)\,i_{w,z}(z-w)^{N-k-1}dz.
\label{2.45}
\end{equation}
Applying the unified formula \eqref{2.36} the sum of \eqref{2.44} and \eqref{2.45}
can be written as
\begin{equation}
a(w)_{(N-k-1)}b(w).
\label{2.46}
\end{equation}
Hence we obtain \emph{Bakalov's formula}
\begin{equation}
  a(w)_{(N-k-1)}b(w)=
  \Res F(z,w)\frac{\partial_w^k\delta(z,w)}{k!}dz=
  \frac{1}{k!} (\partial_z^kF(z,w))|_{z=w},
\label{2.47}
\end{equation}
which holds for each non-negative integer $k$ and sufficiently large
positive integer $N$. The second equality follows from the first one by properties (3) and (5) of the formal delta function.
\begin{rmk}
Since $a(z)$ and $b(z)$ are quantum fields, it follows from (\ref{2.42}) that $F(z,w)v$ lies in the space $V[[z,w]][z^{-1},w^{-1}]$
for each $v\in V$. Hence (\ref{2.47}) makes sense.
\end{rmk}
\begin{rmk}
\label{rmk2.5}
It follows from (\ref{2.42}) that if we replace $a(z)$ in this equation
by $\partial_z^ka(z)$ for some positive integer $k$, then it still holds with
$N$ replaced by $N+k$. 
\end{rmk}

\begin{lem}[Dong]
If $a(z)$, $b(z)$ and $c(z)$ are pairwise mutually local quantum fields,
then $a(z)_{(n)}b(z)$, $c(z)$ is a local pair for any $n\in\Z$.
\end{lem}
\begin{proof}\cite{Bak15}
It suffices to prove that for $ N $ and $ k $ as in \eqref{2.47} we have for some $ M \in \mathbb{Z}_+: $
\begin{equation}
\label{2.50}
(z_2 - z_3)^M (a(z_2)_{(N-k-1)}b(z_2) ) c(z_3) = \pm (z_2 - z_3)^M c(z_3) a(z_2)_{(N-k-1)} b(z_2),
\end{equation}
where $ \pm $ is the Koszul-Quillen sign, if \eqref{2.42} holds for all three pairs $ (a,b), (a,c) $ and $ (b,c) $. We let $ M = 2N + k. $ By Bakalov's formula \eqref{2.47}, the left-hand side of \eqref{2.50} is equal to
\[ \begin{split}
& \frac{1}{k!} (z_2 - z_3)^{2N + k} \left(  \partial^k_{z_1} ((z_1 - z_2)^N a(z_1)b(z_2)c(z_3))  \right) \bigg{|}_{z_1 = z_2}\\
= & (z_2 - z_3)^{2N+k} \sum_{i=0}^{k} 
\binom{N}{i} 
(z_1-z_2)^{N-i} 
(\frac{\partial^{k-i}_{z_1}}{(k-i)!} 
a(z_1)) b(z_2) c(z_3) \bigg{|}_{z_1 = z_2} \\
= & \sum_{i=0}^{k} 
\binom{N}{i} 
(z_2 - z_3)^N (z_1 - z_3)^{N+k} (z_1 - z_2)^{N-i} 
(\frac{\partial^{k-i}_{z_1}}{(k-i)!}  a(z_1))
b(z_2) c(z_3)\bigg{|}_{z_1 = z_2} \, . \\
\end{split}
 \]
Due to \eqref{2.42} for the pair $ (b,c)$, we can permute $ c(z_3) $ with $ b(z_2) $ (up to the Koszul-Quillen sign), and after that similarly permute $ c(z_3) $ and the $(k-i)$-th derivative of
$ a(z_1),  $ using Remark \ref{rmk2.5}. We thus obtain the right-hand side of \eqref{2.50}.
\end{proof}

\section{Lecture 3 (December 16, 2014)}
\label{sec:3}

In this lecture, we will prove the Extension theorem, the Borcherds identity and the skewsymmetry. We will also introduce the concepts of conformal vector, conformal weight and Hamiltonion operators. In the end, we give some properties of the Formal Fourier Transform.

\subsection{Proof of the Extension Theorem} \label{subsec:3.1}

First of all, let us give a name for the data which appeared in the Extension theorem.
 
\begin{defn}
\label{def:PreVa}
A \emph{pre-vertex algebra} is a quadruple
$\{V,\vac,T,\mc{F}=\{a^j(z)=\sum_{n\in\Z}a^j_{(n)}z^{-n-1}\}_{j\in J}\}$,
where $V$ is $z$-algebra with unit element $\vac$, $T \in \End V$ and $\mc{F}$
is a collection of quantum fields with values in $\End V$ satisfying the following
conditions:
\begin{itemize}
\item[(i)](vacuum axiom) $T\vac=0$,

\item[(ii)](translation covariance) $[T,a^j(z)]=\partial_z a^j(z)$ for all $j\in J$,

\item[(iii)](locality) $(z-w)^{N_{ij}}[a^i(z),a^j(w)]=0$ for all $i,j\in J$
with some $N_{ij}\in\Z_+$,

\item[(iv)](completeness) $\Span\{a^{j_1}_{(n_1)}\cdots a^{j_s}_{(n_s)}\vac\mid
j_i\in J,\ n_i\in\Z,\ s\in\Z_+\}=V$.
\end{itemize}
\end{defn}

Let $\{V,\vac,T,\mc{F}\}$ be a pre-vertex algebra. Define 
\begin{equation}
\mc{F}_{\min}=\Span\bigg\{\big(a^{j_1}(z)_{(n_1)}\big(a^{j_2}(z)_{(n_2)}\cdots
\big(a^{j_s}(z)_{(n_s)}I_{V}\big)\cdots\big)\mid
n_i\in\Z,\ j_i\in J,\ s\in\Z_+\bigg\} ,
\end{equation}
where $I_V$ is the constant field equal to the identity
operator $I_V$ on $V$. Let, as in
 Lecture 1, $\mc{F}_{\max}$ be the set of all translation
covariant quantum fields $a(z)$, such that $a(z), a^j(z)$ is a local pair
for all $j\in J$. The following is a more precise version of the Extension theorem, stated in Lecture  1.
\begin{thm}[Extension theorem]
\label{T:4}
For a pre-vertex algebra  $\{V,\vac,T,\mc{F}\}$, let $\mc{F}_{\min},  \mc{F}_{\max}$ be defined as above, then we have, 
\begin{itemize}
\item[(a)] $\mc{F_{\min}} = \mc{F}_{\max}$,
\item[(b)] The map
\begin{equation}
\mathit{fs}\;:\;\mc{F}_{\max} \longrightarrow V , \quad a(z) \longmapsto a(z)\vac\big|_{z=0}
\end{equation}
is well-defined and bijective. Denote by $\mathit{sf} $ the inverse map. 
\item[(c)] The $z$-product $a(z)b:=\mathit{sf}(a)b$ endows $V$ with a vertex algebra
structure, which extends the pre-vertex algebra structure.
\end{itemize}
\end{thm}

\begin{rmk}
\begin{itemize}
\item[(1)] The map $\mathit{fs}$ is called the \emph {field-state correspondence} since it sends a field to a vector in $V$, called a \textquotedblleft  state\textquotedblright \, in physics. Its inverse map, called the \emph{state-field correspondence}, is denoted by
\begin{equation}
\mathit{sf}\;:\; V\to\mc{F}_{\max},\quad a\mapsto a(z) .
\end{equation}
\item[(2)] Denote by $\mc{F}_{\tc}=\{a(z)\mid[T,a(z)]=\partial_za(z)\}$ the space of translation covariant quantum fields. By Lemma \ref{L:2},
 $a(z)\vac\in V[[z]]$ for $a(z)\in\mc{F}_{\tc}$, hence $\mathit{fs}(a(z))\in V$
is well-defined.
\end{itemize}
\end{rmk}

\begin{lem}
\label{L:4}
$\mc{F}_{\tc}$ contains $I_V$, it is $\partial_z$-invariant and is closed under all
$n$-th product, i.e., $a(z)_{(n)}b(z)\in\mc{F}_{\tc}$ for any $n\in\Z$ if $a(z),b(z)\in\mc{F}_{\tc}$ .
\end{lem}

\begin{proof}
Since $[T, I_V]=0=\partial_zI_V$, we have $I_V \in \mc{F}_{tc}$. Now if $a(z)$ is translation covariant, we need to show that $[T, \partial_za(z)]=\partial_z\partial_za(z)$ and so $\partial_z a(z)$ is also tranlation covariant. But 
\begin{equation}
\begin{split}
[T, \partial_za(z)]&=[T, \sum_{n\in \Z}(-n-1)a_{(n)}z^{-n-2}]=\sum_{n\in \Z}(-n-1)[T, a_{(n)}]z^{-n-2}\\
&= \sum_{n\in \Z}(-n-1)(-n)a_{(n-1)}z^{-n-2}
\end{split}
\end{equation}
and
\begin{equation}
\begin{split}
\partial_z\partial_za(z)&=\partial_z(\sum_{n\in \Z}(-n-1)a_{(n)}z^{-n-2})=\sum_{n\in \Z}(-n-1)(-n-2)a_{(n)}z^{-n-3}\\
&= \sum_{n\in \Z}(-n-1)(-n)a_{(n-1)}z^{-n-2} .
\end{split}
\end{equation}
For the last part of this lemma, let us recall the definition of the $n$-th product,
\begin{equation}
a(w)_{(n)}b(w)=\Res\big(a(z)b(w)i_{z,w}(z-w)^n
-b(w)a(z)i_{w,z}(z-w)^n\big)dz,\quad
n\in\Z.
\end{equation}
We want to prove $[T, a(w)_{(n)}b(w)]=\partial_w(a(w)_{(n)}b(w))$.
Both $T$ and $\partial_w$ commute with $\Res $, moreover, $\partial_w$ commutes with $i_{z, w}$ and  $i_{w, z}$. So we have 
\begin{equation}
\begin{split}
\partial_w(a(w)_{(n)}b(w))&=\Res\big(\partial_w(a(z)b(w)i_{z,w}(z-w)^n)
-\partial_w(b(w)a(z)i_{w,z}(z-w)^n)\big)dz\\
&=\Res\big(a(z)(\partial_wb(w))i_{z,w}(z-w)^n
-(\partial_wb(w))a(z)i_{w,z}(z-w)^n\big)dz\\
&+\Res\big(a(z)b(w)i_{z,w}(\partial_w(z-w)^n)
-b(w)a(z)i_{w,z}(\partial_w(z-w)^n)\big)dz .
\end{split}
\end{equation}
Note that $\partial_w(z-w)^n=-\partial_z(z-w)^n$ and $-\Res a(z)i_{w,z}\partial_z(z-w)^ndz=\Res(\partial_za(z))i_{w,z}(z-w)^ndz$. So
\begin{equation}
\partial_w(a(w)_{(n)}b(w))=a(w)_{(n)}\partial_wb(w)+(\partial_wa(w))_{(n)}b(w) .
\end{equation}
This shows that $\partial_w$ is a derivation for the $n$-th product. Now 
\begin{equation}
\begin{split}
[T, a(w)_{(n)}b(w)] &=\Res\big(Ta(z)b(w)i_{z,w}(z-w)^n-a(z)b(w)Ti_{z,w}(z-w)^n\\
&-Tb(w)a(z)i_{w,z}(z-w)^n+b(w)a(z)Ti_{w,z}(z-w)^n\big)dz\\
&=\Res\big(Ta(z)b(w)i_{z,w}(z-w)^n-a(z)Tb(w)i_{z,w}(z-w)^n\\
&+a(z)Tb(w)i_{z,w}(z-w)^n-a(z)b(w)Ti_{z,w}(z-w)^n\\
&-Tb(w)a(z)i_{w,z}(z-w)^n+b(w)Ta(z)i_{w,z}(z-w)^n\big)dz\\
&-b(w)Ta(z)i_{w,z}(z-w)^n+b(w)a(z)Ti_{w,z}(z-w)^n\big)dz\\
&=\Res\big([T, a(z)]b(w)i_{z,w}(z-w)^n-b(w)[T, a(z)]i_{w, z}(z-w)^n\\
&+\Res\big(a(z)[T, b(w)]i_{z,w}(z-w)^n-[T, b(w)]a(z)i_{w, z}(z-w)^n .
\end{split}
\end{equation}
Since both $a(z), b(z)$ are translation covariant, we have 
\begin{equation}
[T, a(w)_{(n)}b(w)] = a(w)_{(n)} \partial_w b(w) + (\partial_wa(w))_{(n)}b(w) .
\end{equation}
This completes the proof.
\end{proof}

We have inclusions
\begin{equation} 
\mc{F}\subset\mc{F}_{\min}\subset\mc{F}_{\max}\subset\mc{F}_{\tc}.
\end{equation}
The first inclusion is because for any $a(z)\in\mc{F}$, we have
$a(z)_{(-1)}I_V=a(z)\in\mc{F}_{\min}$. 
The second inclusion is by Lemma \ref{L:4} and Dong's Lemma (locality).
The last inclusion is by definition.  

\begin{exr}
Show that the constant field $T$ is translation covariant, but is not local to any non-constant
field.
\end{exr}

\begin{lem}
\label{L:5}
Let $a(z),b(z)\in\mc{F}_{\tc}$, and $a=\mathit{fs}(a(z)),b=\mathit{fs}(b(z))$. Then:
\begin{itemize}
\item[(a)] $\mathit{fs}(I_V)=\vac$,
\item[(b)] $\mathit{fs}(\partial_za(z))=Ta$,
\item[(c)] $\mathit{fs}(a(z)_{(n)}b(z))=a_{(n)}b$.
Here we write $a(z)=\sum_{n\in\Z}a_{(n)}z^{-n-1}$.
\end{itemize}
\end{lem}

\begin{proof}
$(a)$ is obvious. For $(b)$, since $a(z)\vac=e^{zT}a= a+ (Ta)z +\frac {T^2a}{2}z^2+o(z^2)$ we have $\partial_za(z)\vac=Ta+T^2az +o(z)$, so $\mathit{fs}(\partial_za(z))= \partial_za(z)\vac_{z=0}= Ta$. For $(c)$, by definition, we have
\begin{equation}
\mathit{fs}(a(z)_{(n)}b(z))=a(z)_{(n)}b(z)\vac\big|_{z=0} \, ,
\end{equation}
and the right hand side, by definition of the $n$-th product, is equal to 
\begin{equation}
\Res\big(a(w)b(z)i_{w,z}(w-z)^n\vac
-b(z)a(w)i_{z,w}(w-z)^n\vac dw\big)\big|_{z=0} \, .
\end{equation}
Now, since $a(w)\vac \in V[[w]]$ and $i_{z,w}(w-z)^n$ has only non-negative powers of $w$, we have 
 $$\Res b(z)a(w)i_{z,w}(w-z)^ndw\vac =0.$$
For the first term, since $b(z)\vac \in V[[z]]$, we can let $z=0$ before we calculate the residue, which gives 
\begin{equation}
\Res a(w)b(z)i_{w,z}(w-z)^n\vac)dw\big|_{z=0} =\Res a(w)bw^ndw =a_{(n)}b .
\end{equation}
This completes the proof.
\end{proof}

\begin{lem}
\label{L:6}
Let $a(z)\in\mc{F}_{\tc}$. Then $e^{wT}a(z)e^{-wT}=i_{z, w}a(z+w)$.
\end{lem}

\begin{proof}
Both sides are in $(\End V)[[z,z^{-1}]][[w]]$, and both satisfy the differential
equation $\frac{df(w)}{dw}=(\ad T)f(w)$ with the initial condition $f(0)=a(z)$.
\end{proof}

\begin{lem}[Uniqueness Lemma]
\label{L:u}
Let $\mc{F}'\subset\mc{F}_{\tc}$ and let $a(z)$ be some quantum field in $\mc{F}_{\tc}$. Assume that
\begin{itemize}
\item[(i)] $\mathit{fs}(a(z))=0$,
\item[(ii)] $a(z)$ is  local with any element in $\mc{F}'$,
\item[(iii)] $\mathit{fs}(\mc{F}')=V$.
\end{itemize}
Then $a(z)=0$.
\end{lem}

\begin{proof}
Let $b(z)\in\mc{F}'$. By the locality of $a(z)$ and $b(z)$, we have
$(z-w)^N[a(z),b(w)]=0$ for some $N\in\Z_+$. Apply both sides to $\vac$. We get
\begin{equation}
(z-w)^Na(z)b(w)\vac= \pm (z-w)^Nb(w)a(z)\vac .
\end{equation}
By the property (i) we have $a_{(-1)}\vac=0$ and $ a(z) $ is translation covariant, hence by Lemma \ref{L:2}(b), $a(z)\vac =0$. Now, by Lemma \ref{L:2}(a), 
$b(w)\vac\in V[[w]]$, so we can let
$w=0$ and get $z^Na(z)b=0$, which means $a_{(n)}b=0$ for any
$n\in\Z$. This is true for any $b\in V$ by the property (iii). So in fact, we have
$a(z)=0$.
\end{proof}

\begin{proof}[Proof of the Extension Theorem]
\label{pf:ExTh}
We can get the following two statements about the map $\mathit{fs}$:
\begin{itemize}
\item[(i)] the map $\mathit{fs}\;:\;\mc{F}_{\min}\to V$ defined by
$fs(a(z)) = a(z)\vac\big|_{z=0}$ is given by 
\begin{equation}
(a^{j_1}(z)_{(n_1)}( a^{j_2}(z)_{(n_2)} \cdots (a^{j_s}(z)_{(n_s)}I_{V})\cdots)
\mapsto a^{j_1}_{(n_1)} a^{j_2}_{(n_2)} \cdots a^{j_s}_{(n_s)}\vac,
\label{3.5}
\end{equation}
\end{itemize}
\noindent and it is surjective, by (a), (c) of Lemma \ref{L:5} and the completeness axiom;
\begin{itemize}
\item[(ii)] $\mathit{fs}\;:\;\mc{F}_{\max}\to V$ is injective using the Uniqueness
Lemma with $\mc{F}'=\mc{F}_{\min}$.
\end{itemize}
Recall the inclusion $\mc{F}_{\min}\subset\mc{F}_{\max}$. We now have that
$\mathit{fs}\colon\mc{F}_{\min}\to V$ is surjective and
$\mathit{fs}\colon\mc{F}_{\max}\to V$ is injective,
so we can conclude that it is in fact bijective and $\mc{F}_{\min}=\mc{F}_{\max}$.
This proves (a) and (b) in the Extension Theorem.
For (c), we need to show that $a(z)$ is translation covariant $\forall a \in V$ and that each pair $a(z), b(w)\, \forall a, b \in V$ is a local pair. But translation covariance comes from Lemma \ref{L:4} and locality comes from Dong's lemma. 
\end{proof}

\begin{cor}[of the proof]
\label{cor3.1}
\begin{itemize}
\item[(a)] $\mathit{sf}(a^{j_1}_{(n_1)} a^{j_2}_{(n_2)}\cdots a^{j_s}_{(n_s)}\vac)
=(a^{j_1}(z)_{(n_1)}(a^{j_2}(z)_{(n_2)}\cdots(a^{j_s}(z)_{(n_s)}I_{V})\cdots)$.
\item[(b)] $(Ta)(z)=\partial_za(z)$. 
\item[(c)] $(a_{(n)}b)(z)=a(z)_{(n)}b(z)$, which is called
the \emph{$n$-th product identity}.
\end{itemize}
\end{cor}

\begin{proof}
(a) is by definition since $\mathit{sf}$ is the inverse of  $\mathit{fs}$,
while $\mathit{fs}$ is given  by \eqref{3.5}. Letting $s=1$, $n_1=-2$ in (a) we get (b).
Letting $s=2,$ $n_1=n$, $n_2=-1$ in (a) we get (c).
\end{proof}

\begin{rmk}
\label{rem3.1}
Due to Corollary \ref{cor3.1}(b) and Remark \ref{rem1.1}, the Definitions \ref{def1.3} and \ref{def1.5} of a vertex algebra are equivalent. 
\end{rmk}

\begin{rmk}[Special case of (a) in the corollary] 
\label{rem3.3}
For $n_1, \ldots, n_s\in\Z_+$, we have, 
\begin{equation}
\mathit{sf}(a^{j_1}_{(-n_1-1)} a^{j_2}_{(-n_2-1)}\cdots a^{j_s}_{(-n_s-1)}\vac)
=\frac{:\partial_z^{n_1}a^{j_1}(z)\partial_z^{n_2}a^{j_2}(z)\cdots
\partial_z^{n_s}a^{j_s}(z):}{n_1!n_2!\cdots n_s!} \ .
\label{3.17}
\end{equation}
\end{rmk}

\begin{cor}[of the proof]
$\Lie_V:=\Span\{a_{(n)} |\, a\in V,\ n\in\Z\}\subset\End V$ is a subalgebra of
the Lie superalgebra
$\End V$ with the commutator formula
\begin{equation}
\label{e3.19}
[a(z),b(w)]=\sum_{j\geq 0}(a(w)_{(j)}b(w))\frac{\partial_w^j\delta(z, w)}{j!},
\end{equation}
which is equivalent to each of the following two expressions
\begin{align}
[a_{(m)}, b(z)]&=\sum_{j \geq 0}{m \choose j}(a_{(j)}b)(z)z^{m-j}, \label{3.19}\\
[a_{(m)}, b_{(n)}]&=\sum_{j \geq 0}{m \choose j}(a_{(j)}b)_{(m+n-j)}. \label{3.20}
\end{align}
Moreover, $ \Lie_V $ is a regular formal distribution Lie algebra with the data
$(\Lie_V,\mc{F}=\{a(z)\}_{a\in V},\ad T)$.
\end{cor}

\subsection{Borcherds identity and some other properties}
\label{subsec:3.2}


\begin{prp}[Borcherds identity]
\label{BI}
For $n\in\Z$, $a,b\in V$, where $V$ is a vertex algebra, we have
\begin{equation}
a(z)b(w)i_{z,w}(z-w)^n-
(-1)^{p(a)p(b)}
b(w)a(z)i_{w,z}(z-w)^n
=\sum_{j\in\Z_+}(a_{(n+j)}b)(w)\frac{\partial_w^j\delta(z, w)}{j!}.
\label{3.10}
\end{equation} 
\end{prp}

\begin{proof}
The left hand side of \eqref{3.10} is a local formal distribution in $z$ and $w$. Apply to it the Decomposition theorem to get that it is equal to 
\begin{equation}
 \sum_{j \in \Z_+}c^j(w)\partial_w^j\delta(z, w)/j!\,,
\end{equation}
where 
\begin{equation}
\begin{split}
c^j(w)&=\Res\big(a(z)b(w)i_{z,w}(z-w)^n-(-1)^{p(a)p(b)}
b(w)a(z)i_{w, z}(z-w)^n\big)(z-w)^jdz\\
&=\Res\big(a(z)b(w)i_{z,w}(z-w)^{n+j}-(-1)^{p(a)p(b)}
b(w)a(z)i_{w,z}(z-w)^{n+j}\big)dz\\
&=a(w)_{(n+j)}b(w)\\
&=(a_{(n+j)}b)(w).
\end{split}
\end{equation}
The last equality follows from the $n$-th product formula, all other equalities are just by definition.
\end{proof}
\begin{exr}
Prove that a unital $z$-algebra satisfying the Borcherds identity is a vertex algebra.
\end{exr}

\begin{prp}[Skewsymmetry]
\label{SkewSym}
For $a,b\in V$, where $V$ is a vertex algebra, we have: 
\begin{equation}
a(z)b= (-1)^{p(a) p(b)} e^{zT}b(-z)a.
\end{equation}
\end{prp}
\begin{proof}
By locality, we know that, there exists $N\in\Z$, such that 
$$(z-w)^{N}a(z)b(w)=(-1)^{p(a) p(b)}(z-w)^{N}b(w)a(z).$$ 
Apply both sides to $\vac$; by Lemma \ref{L:2}(b) we get
\begin{equation}
(z-w)^{N}a(z)e^{wT}b=(-1)^{p(a) p(b)}(z-w)^{N}b(w)e^{zT}a.
\end{equation}
Now use Lemma \ref{L:6}:
\begin{equation}
RHS= (-1)^{p(a)p(b)} (z-w)^{N}e^{zT}e^{-zT}b(w)e^{zT}a= (-1)^{p(a)p(b)} (z-w)^{N}e^{zT}i_{w,z}b(w-z)a.
\end{equation}
For $N\gg 0$, this is a formal power series in $(z-w)$, so we can set $w=0$
and get 
\begin{equation}
LHS=z^Na(z)b= (-1)^{p(a) p(b)} e^{zT}z^Nb(-z)a=RHS,
\end{equation}
which proves the proposition.
\end{proof}

\begin{prp}
\label{T-Der}
$T$ is a derivation for all $n$-th products, i.e., 
\begin{equation}
T(a_{(n)}b)=(Ta)_{(n)}b+a_{(n)}(Tb),\quad\forall n\in\Z.
\label{3.16}
\end{equation}
\end{prp}
\begin{proof}
It follows from Remark \ref{rem3.1}.
\end{proof}

In view of the $n$-th product identity, we let $:ab:=a_{(-1)}b$ and call this
the normally ordered product of two elements of a vertex algebra. 
\begin{prp} The $n$-th products for negative $n$ are expressed via the normally ordered product: $a_{(-n-1)}b=\,:\frac{T^na}{n!}b:$\,.
\end{prp}

\begin{proof}
We have $(a_{(-n-1)}b)(z)=a(z)_{(-n-1)}b(z)=\,:\frac{\partial_z^na(z)}{n!}b(z):$\,,
where the first equality is the $n$-th product identity and the second equality
is (2.37). But we also have
$T(a)(z)=\partial_za(z)$, hence by induction we have
$:\frac{\partial_z^na(z)}{n!}b(z):\,=\,:\frac{(T^na)(z)}{n!}b(z):\,=
(\frac{(T^na)}{n!}_{(-1)}b)(z)$, 
and by the bijection of the state-field correspondence,
we have $a_{(-n-1)}b=\frac{(T^na)}{n!}_{(-1)}b=\,:\frac{T^na}{n!}b:$\,.
\end{proof}


Now we take care of the $n$-th products $a_{(n)}b$ for $n\in \Z_+$.
For this we define the \emph{$\lambda$-bracket}
\begin{equation} \label{3.30}
[a_\lambda b]=\sum_{j\geq 0}\frac{\lambda^j}{j!}(a_{(j)}b)\in V[\lambda],\quad
\text{for}\quad a,b\in V.
\end{equation}
Thus we get a quadruple $(V,T,:ab:,[a_\lambda b])$, which will be shown in the
next lecture to have a very similar structure to a Poisson Vertex Algebra (PVA).

\subsection{Conformal vector and conformal weight, Hamiltonian operator}

\begin{defn}
\label{conformalvector}
A vector $L$ of a vertex algebra $V$ is called a \emph{conformal vector} if
\begin{itemize}
\item[(i)] $L(z)=\sum_{n\in\Z}L_nz^{-n-2}$, such that,
\begin{equation}
[L_m, L_n]=(m-n)L_{m+n}+\delta_{m,-n}\frac{m^3-m}{12}\,cI_V
\end{equation}
for some $c\in\mathbb{F}$, which is called the 
\emph
{central charge},
\item[(ii)] $L_{-1}=T$,
\item[(iii)] $L_0$ acts diagonalizably on $V$, its eigenvalues are called
\emph{conformal weights}.
\end{itemize}
\end{defn}

\noindent
Since $L_{n-1}=L_{(n)}$, using the commutator formula \eqref{e3.19}, we get
\begin{equation}
\label{e3.34}
[L(z),a(w)]=\sum_{j\geq 0}(L_{j-1}a)(w)\partial_w^j\delta(z,w)/j!,
\end{equation}
which is equivalent to (cf. \eqref{3.20})
\begin{equation}
[L_{(m)}, a_{(n)}]=\sum_{j\geq 0}{m \choose j}(L_{j-1}a)_{(m+n-j)}.
\label{3.33}
\end{equation}
So we have
\begin{equation}
[L_\lambda a]=\sum_{j\geq 0} \frac{\lambda^j}{j!} L_{j-1} a=Ta+\lambda\Delta_aa+o(\lambda).
\label{3.34}
\end{equation}
Here we assume that $a$ is an eigenvector of $L_0$ with the eigenvalue $\Delta _a$.
We call $L_0$ the energy operator. It is a Hamiltonian operator by the definition below and \eqref{3.33} for $ m = 0 $.

\begin{defn}
A diagonalizable operator $H$ is called a \emph{Hamiltonian operator}
if it satisfies the equation 
\begin{equation}
\label{3.24}
[H,a(z)]=(z\partial_z+\Delta_a)a(z) \quad\Longleftrightarrow\quad [H,a_{(n)}]=(\Delta_a-n-1)a_{(n)} 
\end{equation}
for any eigenvector $a$ of $H$ with eigenvalue $\Delta_a$.
\end{defn}

 If we write $a(z)=\sum_{n\in -\Delta_a+\Z}a_nz^{-n-\Delta_a}$,
then due to the equality $a_{(n)}=a_{n-\Delta_a+1}$, we have:
\begin{equation}
\label{e3.38}
[H,a_n]=-na_{n}.
\end{equation}
This is an equivalent definition of a Hamiltonian operator. 
\begin{prp}
If $H$ is a Hamiltonian operator, then we have:
\begin{itemize}
\item[(a)]  $\Delta_{\vac}=0$,
\item[(b)]  $\Delta_{Ta}= \Delta_a+1$,
\item[(c)]  $\Delta_{a_{(n)}b}=\Delta_{a}+\Delta_{b}-n-1$.
\end{itemize}
\end{prp}

\begin{proof}
To prove (a), we just need to know that $\vac(z)=I_V$, and we use
\eqref{3.24} with $a=\vac$.
Since $Ta=a_{(-2)}\vac$, (b) follows from $(a)$ and $(c)$ with $b=\vac$, $n=-2$. For $(c)$, we have 
\begin{equation}
\begin{split}
H(a_{(n)}b)&=[H, a_{(n)}]b+a_{(n)}Hb\\
&=(\Delta_a-n-1)a_{(n)}b+\Delta_ba_{(n)}b\\
&=(\Delta_{a}+\Delta_{b}-n-1)a_{(n)}b .
\end{split}
\end{equation}
\end{proof}

\begin{rmk}
\begin{itemize}
\item[(a)] For a conformal vector $L$, we have
$[L_\lambda L]=(T+2\lambda)L+\frac{\lambda^3}{2}c\vac$, which implies $\Delta_L=2$. That is why we write $ L(z) $ in the form $L(z)=\sum_{n\in \Z}L_nz^{-n-2}$.
\item[(b)] Conformal weight is a good ``book-keeping device'', if we let
$\Delta_\lambda=\Delta_T=1$. Then all summands in the $\lambda$-bracket 
$[a_\lambda b]=\sum_{j\geq 0}\frac{\lambda^j}{j!}(a_{(j)}b)$ have the same
conformal weight $\Delta_a+\Delta_b-1$. 
\end{itemize}
\end{rmk}

\begin{rmk}
\label{remark3.5}
The translation covariance \eqref{1.4} of the quantum field $ a(z) $ is equivalent to the following ``global'' translation covariance:
\[ e^{\epsilon T} a(z) e^{-\epsilon T} = i_{z, \epsilon} a(z+ \epsilon). \]
Likewise, the property \eqref{3.24} of $ a(z)  $ is equivalent to the following ``global'' scale covariance:
\[ \gamma^H a(z) \gamma^{-H} = (\gamma^{\Delta_a} a) (\gamma z), \mbox{ where } Ha = \Delta_a a. \]
The more general property \eqref{e3.34} is called the conformal variance. It is the basic symmetry of conformal field theory.
\end{rmk}

\subsection{ Formal Fourier Transform}

\begin{defn}
The \emph{Formal Fourier Transform} is the map
$F^\lambda_z\colon U[[z,z^{-1}]]\mapsto U[[\lambda]]$ defined by
\begin{equation}
F^\lambda_za(z)=\Res  e^{\lambda z}a(z)dz.
\label{3.28}
\end{equation}
\end{defn}

\begin{prp}
\label{P:FFT}
\begin{itemize}
\item[(a)] $F^\lambda_z\partial_za(z)=-\lambda F^\lambda_za(z)$,
\item[(b)] $F^\lambda_z\partial_w^k\delta(z,w)=e^{\lambda w}\lambda^k$,
\item[(c)] $F^\lambda_za(-z)=-F^{-\lambda}_za(z)$,
\item[(d)] $F^\lambda_z(e^{zT}a(z))=F^{\lambda +T}_za(z)$, where $T\in\End U$,
provided that $a(z)\in U((z))$.
\end{itemize}
\end{prp}

\begin{proof}
\begin{itemize}
\item[(a)]   Assume $a(z)=\sum_{n\in \Z}a_{(n)}z^{-n-1}$, then $\partial_za(z)=\sum_{n\in \Z}(-n-1)a_{(n)}z^{-n-2}$. Now
\begin{equation}
\begin{split}
F^\lambda_za(z)&=\Res e^{\lambda z}a(z)dz\\
&=\Res (\sum_{i\in \Z_+}\dfrac{\lambda^iz^i}{i!})(\sum_{n\in \Z}a_{(n)}z^{-n-1})dz\\
&=\sum_{n\in \Z_+}\dfrac{\lambda^n}{n!}a_{(n)} ,\\
F^\lambda_z\partial_za(z)&= \sum_{n\in \Z_+}\dfrac{\lambda^n}{n!}(-n)a_{(n-1)}\\
&=-\lambda \sum_{n\in \Z_+}\dfrac{\lambda^n}{n!}a_{(n)} .
\end{split}
\end{equation} 

\item[(b)]  Recall that $\frac{\partial^k_w\delta(z,w)}{k!}=\sum_{j\in\Z}{j\choose k}w^{j-k}z^{-j-1}$, so 
\begin{equation}
\begin{split}
F^\lambda_z\partial_w^k\delta(z,w)&=\Res e^{\lambda z}k!\sum_{j\in\Z_+}{j\choose k}w^{j-k}z^{-j-1}dz\\
&=\sum_{j\in\Z_+}\dfrac{\lambda^j}{j!}k!\dfrac{j!}{k!(j-k)!}w^{j-k}\\
&=\lambda^k\sum_{j-k\in\Z_+}\dfrac{\lambda^{j-k}}{(j-k)!}w^{j-k}\\
&=e^{\lambda w}\lambda^k .
\end{split}
\end{equation}
\item[(c)] By definition
\begin{equation}
\begin{split}
F^\lambda_za(-z)&=\Res(\sum_{i\in \Z_+}\dfrac{\lambda^iz^i}{i!})(\sum_{n\in \Z}a_{(n)}(-z)^{-n-1})dz\\
&=\sum_{n\in \Z_+}\dfrac{\lambda^n}{n!}(-1)^{n+1}a_{(n)}\\
&=-\sum_{n\in \Z_+}\dfrac{(-\lambda)^n}{n!}a_{(n)}\\
&=-F^{-\lambda}_za(z) .
\end{split}
\end{equation}
\item[(d)] Since $a(z)\in U((z))$, $e^{zT}a(z)\in U((z))$ is well defined. Now
\begin{equation}
\begin{split}
F^\lambda_z(e^{zT}a(z))&=\Res e^{\lambda z}e^{zT}a(z)dz\\
&=\Res e^{(\lambda+T) z}a(z)dz\\
&=F^{\lambda +T}_za(z) .
\end{split}
\end{equation}
\end{itemize}
\end{proof}

Similarly, we can define the Formal Fourier Transform in two variables.

\begin{defn}
The Formal Fourier Transform in two variables is the map 
\begin{equation}
F_{z,w}^{\lambda} : U[[z,z^{-1},w,w^{-1}]]\rightarrow U[[w,w^{-1}]][[\lambda]],
\end{equation}
defined by 
\begin{equation}
F_{z,w}^{\lambda}a(z, w)=\Res e^{\lambda(z-w)}a(z,w)dz = e^{-\lambda w} F^{\lambda}_z \, a(z,w).
\end{equation}
\end{defn}

\begin{prp}
\label{P:FFT2}
\begin{itemize}

 \item [($\alpha$)] 
 $F_{z,w}^{\lambda}\partial_{z}a(z,w)=-\lambda F_{z,w}^{\lambda}a(z,w)=[\partial_{w},F_{z,w}^{\lambda}]a(z,w)$,
 
 \item [($\beta$)] $F_{z,w}^{\lambda}\partial_{w}^k\delta(z,w)=\lambda^{k}$,
 
 \item [($\gamma$)] $F_{z,w}^{\lambda}a(w,z)=F_{z,w}^{-\lambda-\partial_{w}}a(z,w)$ provided that $a(z,w)$ is local,
 
 \item [($\delta$)] $F_{z,w}^{\lambda}F_{x,w}^{\mu}=F_{x,w}^{\lambda+\mu}F_{z,x}^{\lambda}$.
\end{itemize}
\end{prp}

\begin{proof}
Since $F_{z,w}^{\lambda}=e^{-\lambda w}F_{z}^\lambda$, $(\alpha)$ and $(\beta)$ follow from the properties $(a)$ and $(b)$ in Proposition \ref{P:FFT}. $(\delta)$ holds since
\begin{equation}
\Res\Res e^{\lambda(z-w)+\mu(z-w)}a(z,w,x)dx dz=\Res\Res e^{\lambda(z-x)}e^{(\lambda+\mu)(x-w)}a(z,w,x)dx dz
\end{equation}
Finally, due to the Decomposition theorem, it suffices to check ($\gamma$) (interpretation as before) for $ a(z,w) = c(w) \partial^k_w \, \delta (z,w): $
\begin{align*}
\mathrm{LHS} &= \Res e^{\lambda (z-w)} c(z) \partial^k_z \, \delta (w,z) dz = (-1)^k \, \Res e^{\lambda (z-w)} c(z) \partial_w \, \delta (w,z) dz\\
& = (-1)^k e^{-\lambda w} \partial^k_w \, \Res e^{\lambda z} c(z) \delta(z,w) dz = (-1)^k e^{-\lambda w} \, \partial^k_w e^{\lambda w} c(w)\\
& = (- \lambda -\partial_w)^k \, c(w),
\end{align*}
using the properties (3) and (5) of the delta function.
 
\end{proof}

\newpage

\section{Lecture 4 (December 18, 2014)} 

The Formal Fourier Transform $F_{z}^\lambda$ is very important for us, since the $ \lambda $-bracket \eqref{3.30} is  $[a_\lambda b]=F_{z}^{\lambda}a(z)b$, i.e., the Fourier transform of the $z$-product is the $\lambda$-bracket. We also note that $:ab:\,(=a_{(-1)}b)=\Res \dfrac{a(z)b}{z}dz$. These observations will be important for studying properties of the normally ordered product $:\;:$ and the $\lambda$-bracket. For simplicity we will further consider vertex algebras $ V $ of purely even parity only. The general case follows by the Koszul-Quillen rule. 

\subsection{Quasicommutativity, quasiassociativity and the noncommutative Wick's formula}
\begin{lem}
[Newton-Leibniz (NL) Lemma] For any $a(z)\in U[[z]]$, we have

\begin{equation}
\label{NLlemma}
F_{z}^{\lambda}\frac{a(z)}{z}=\mbox{ \Res } a(z)\frac{dz}{z}+\int_{0}^{\lambda}F_{z}^{\mu}a(z)d\mu .
\end{equation}
\end{lem}

\begin{proof}
Both sides are formal power series in $\lambda$, they are equal at $\lambda=0$, and their derivatives by $\lambda$ are also equal, so they are equal.
\end{proof}

\begin{prp}[Quasicommutativity of :\;:] The commutator for the normally ordered product and $\lambda$-bracket are related as follows
\begin{equation}
\label{quasicommu}
:ab:-:ba:=\int_{-T}^{0}[a_\lambda b]d\lambda .
\end{equation}
\end{prp}

\begin{proof}
Apply $F_{z}^\lambda$ to both sides of skewsymmetry, divided by $z$, and set $\lambda=0$. We get 
\begin{equation}
F_{z}^{\lambda}\frac{a(z)b}{z}\bigg|_{\lambda=0}=F_{z}^{\lambda}\frac{e^{zT}b(-z)a}{z}\bigg|_{\lambda=0} .
\end{equation}
By definition
\begin{equation}
LHS = \,:ab:\, =a_{(-1)}b .
\end{equation} 
Next, using property $(d)$ of the FFT in Proposition \ref{P:FFT}, we have 
\begin{equation}
\begin{split}
RHS &=  F_{z}^{\lambda+T}\frac{b(-z)a}{z}\bigg|_{\lambda=0}\\
(\text{by NL Lemma})\quad    &= \Res_{z}\frac{b(-z)a}{z}+\int_{0}^{\lambda+T}F_{z}^{\mu}b(-z)a \, d\mu\bigg|_{\lambda=0} \\
(\text{by property $(c)$ of FFT in Prop \ref{P:FFT}}) \quad
&= {} :ba:- \int_{0}^{T}F_z^{-\mu}b(z)a \, d\mu\\
&={} :ba: -\int_{0}^{T}[b_{-\mu}a] \, d\mu\\
(\text{by skewsymmetry of the $\lambda$-bracket})\quad 
&= {} :ba:+\int_{0}^{T}[a_{\mu+ T}b] \, d\mu\\
&= {} :ba:+\int_{-T}^{0}[a_{\mu }b] \, d\mu .
\end{split}
\end{equation}
\end{proof}

%
Next we derive the following important identity.

\begin{prp} For $a, b, c$ in a vertex algebra $V$, we have the following identity in $V[[\lambda, w, w^{-1}]]$
\begin{equation}
\label{very important}
[a_\lambda b(w)c]=e^{w\lambda}[a_\lambda b](w)c+b(w)[a_\lambda c].
\end{equation}
\end{prp}

\begin{proof}
The following identity in $ V[[ z^{\pm1}, w^{\pm 1}]] $ is obvious:
\begin{equation}
a(z)b(w)c=[a(z),b(w)]c+b(w)a(z)c .
\end{equation}
Applying to both sides $F_{z}^\lambda=e^{w\lambda}F_{z,w}^\lambda$, we get
\begin{equation}
[a_{\lambda}b(w)c] = e^{w\lambda}F_{z,w}^\lambda [a(z),b(w)]c+ b(w) F_{z}^\lambda a(z)c = e^{w\lambda} [a_{\lambda} b] (w) c+b(w)[a_\lambda c],
\end{equation}
where we have used the $n$-th product formula $a(w)_{(n)}b(w)=(a_{(n)}b)(w),  n \in \mathbb{Z}_+$.

\end{proof}

We have the following two important properties of a vertex algebra.
\begin{prp}
 Assume $a, b, c$ in a vertex algebra  $V$. Then we have
\begin{itemize}

\item[(a)] Quasiassociativity formula
\begin{equation}
\label{quasiassoci}
::ab:c:-:a:bc::=:\left(\int_{0}^{T}d\lambda a\right)[b_{\lambda}c]:+:\left(\int_{0}^{T}d\lambda b\right)[a_{\lambda}c]: .
\end{equation}

\item[(b)] Non-commutative Wick's formula
\begin{equation}
\label{noncommwick}
[a_{\lambda}:bc:]=:[a_{\lambda}b]c:+:b[a_{\lambda}c]:+\int_{0}^{\lambda}
[[a_{\lambda}b]_{\mu}c]d\mu .
\end{equation}
\end{itemize}
\end{prp}

\begin{proof}

\begin{itemize}

\item[(1)]
Apply $\Res\frac{1}{z}dz$ to the -1st product identity:
\begin{equation}
:ab:(z)c=:a(z)b(z):c=a(z)_+b(z)c+b(z)a(z)_-c ,
\end{equation}
 and use that
\begin{equation}
\begin{split}
\Res\frac{1}{z}(:ab:(z)c)dz&=(:ab:)_{(-1)}c=::ab:c: \, ,\\
\Res\frac{1}{z} (a(z)_+b(z)c)dz&=:a:bc::+\sum_{j\in \Z_+}a_{(-j-2)}b_{(j)}c\\
&=:a:bc::+:(\int_0^T d\lambda a)[b_\lambda c]: \, , \\
\Res\frac{1}{z}(b(z)a(z)_-c)dz&=\sum_{j\in \Z_+} b_{(-j-2)}a_{(j)}c\\
&=:(\int_0^T d\lambda b)[a_\lambda c]: \,.
\end{split}
\end{equation}

\item[(2)] 
Take $\Res\frac{1}{w}dw$ of both sides of formula \eqref{very important}:
\begin{equation}
\label{very important recall}
\Res\frac{1}{w}[a_\lambda b(w)c]dw=\Res\frac{1}{w}(e^{w\lambda}[a_\lambda b](w)c+b(w)[a_\lambda c])dw.
\end{equation}
Since  $\Res\frac{b(w)c}{w}dw=b_{(-1)}c=:bc:$, we have 
\begin{equation}
\Res\frac{1}{w}[a_\lambda b(w)c]dw=[a_{\lambda}:bc:] .
\end{equation}
For the second term of the right-hand side of \eqref{very important recall},
\begin{equation}
\Res\frac{1}{w}b(w)[a_\lambda c] dw =:b[a_\lambda c]: \,.
\end{equation} 
Using the NL Lemma \ref{NLlemma} for the first term in the right hand side of \eqref{very important recall}, we have 
\begin{equation}
\begin{split}
\Res e^{w\lambda}\dfrac{[a_\lambda b](w)c}{w}dw&=F_{w}^\lambda \dfrac{[a_\lambda b](w)c}{w}\\
&=\Res \dfrac{[a_\lambda b](w)c}{w}dw + 
\int_{0}^{\lambda}F_{w}^{\mu}[a_\lambda b](w)c d\mu\\
&=:[a_\lambda b]c:+\int_{0}^{\lambda}[[a_\lambda b]_\mu c]d\mu .
\end{split}
\end{equation} 

\end{itemize}

\end{proof}

\begin{rmk}
The expression $:(\int_0^T d\lambda a)[b_\lambda c]:$ should be understood in the following way. We know that $[b_\lambda c]=\sum_{j\in \Z_+} b_{(j)}c \dfrac{\lambda ^j}{j!}$, so $:a[b_\lambda c]:=\sum_{j\in \Z_+}a_{(-1)}b_{(j)}c\dfrac{\lambda ^j}{j!}$. We have $\int_0^T \dfrac{\lambda ^j}{j!} d\lambda= \dfrac{T^{j+1}}{(j+1)!}$; letting  $\dfrac{T^{j+1}}{(j+1)!}$ act just on $a$ we get 
$\left( \dfrac{T^{j+1}a}{(j+1)!}  \right) _{(-1)}
=a_{(-j-2)}$, so $:(\int_0^T d\lambda a)[b_\lambda c]:=\sum_{j\in \Z_+}a_{(-j-2)}b_{(j)}c$ .
\end{rmk}

\subsection{Lie conformal algebras vs vertex algebras.}

Let $\mathfrak{g}$ be a Lie algebra, and let $\mathfrak{g}[[w,w^{-1}]]$ be the space of all $\mathfrak{g}$-valued formal
distributions. This space is an $\mathbb{F}[\partial]$-module by defining 
\begin{equation}
\partial a(w):=\partial_{w}a(w) .
\end{equation}
It is closed under the following (formal) $\lambda$-bracket: for $a=a(w),b=b(w) \in \mathfrak{g}[[w,w^{-1}]]$. Let 
\begin{equation}
[a_{\lambda}b](w):=F_{z,w}^{\lambda}[a(z),b(w)].
\end{equation}
Indeed, by definition of $F_{z,w}^{\lambda}$ and its property $(\beta)$, we have:
\begin{equation}
\begin{array}{lll}
 [a_{\lambda}b](w) & = & \Res e^{\lambda(z-w)}[a(z),b(w)]dz\\
     & = & \underset{j\in \mathbb{Z_{+}}}{\sum}\frac{\lambda^j}{j!}\Res (z-w)^j[a(z),b(w)]dz\\
     & = & \underset{j\in \mathbb{Z_{+}}}{\sum}\frac{\lambda^j}{j!}(a_{(j)}b)(w)\in \mathfrak{g}[[w,w^{-1}]][[\lambda]].
\end{array} 
\end{equation}
Thus $[a_{\lambda}b](w)$ is a generating series for  $j$-th products of $a(w)$ and $b(w)$. It is a formal power series in $\lambda$ in general, but if the pair $(a(w), b(w))$ is local, $[a_{\lambda}b]\in \mathfrak{g}[[w,w^{-1}]][\lambda]$ is polynomial in $\lambda$.

\begin{prp}
\label{prop4.4}
Assume $a(w), b(w), c(w)\in \mathfrak{g}[[w,w^{-1}]]$ for some Lie algebra $\mathfrak{g}$ with $\partial=\partial_w$ defined as above. Denote $a=a(w), b=b(w), c=c(w)$. Then the $\lambda$-bracket defined as above satisfies the following properties: 
\begin{equation}
\begin{array}{ll}
\mbox{(sesquilinearity)}& [\partial a_{\lambda}b]=-\lambda[a_{\lambda}b],\quad [a_{\lambda}\partial b ]=(\lambda+\partial)[a_{\lambda}b],\\

\mbox{(skewsymmetry)} & [b_{\lambda}a]=-[a_{-\lambda-\partial}b]\quad \mbox{if $a, b$ is a local pair}, \\

\mbox{(Jacobi identiy)} & [a_{\lambda}[b_{\mu}c]]=[[a_{\lambda}b]_{\lambda+\mu}c]+[b_{\mu}[a_{\lambda} c]].
\end{array} 
\end{equation}

\end{prp}

\begin{proof}
The sesquilinearity comes from $(\alpha)$ and the skewsymmetry comes from $(\gamma)$ in Proposition \ref{P:FFT2} about properties of formal Fourier transform in two variables.
For the Jacobi identity we have:
\begin{equation}
\begin{split}
[a_{\lambda}[b_{\mu}c]](w) & :=  F_{z,w}^{\lambda} [a(z),F_{x,w}^{\mu}[b(x),c(w)]] \\
 & =  F_{z,w}^{\lambda} F_{x,w}^{\mu} [a(z),[b(x),c(w)]]   \\ 
& =  F_{z,w}^{\lambda} F_{x,w}^{\mu} [[a(z),b(x)],c(w)]]+ F_{z,w}^{\lambda} F_{x,w}^{\mu} [b(x),[a(z),c(w)]] .
\end{split}
\end{equation}
The last equality comes from the Jacobi identiy in the Lie algebra $\mathfrak{g}$. By property $(\delta)$ of the formal Fourier transform in Proposition \ref{P:FFT2}, we have:
\begin{equation}\begin{array}{lll}
F_{z,w}^{\lambda} F_{x,w}^{\mu} [[a(z),b(x)],c(w)]] & =  & F_{x,w}^{\lambda+\mu} F_{z,x}^{\lambda} [[a(z),b(x)],c(w)]\\ 
   & =  & F_{x,w}^{\lambda+\mu} [F_{z,x}^{\lambda} [[a(z),b(x)],c(w)]]\\
   &= &[[a_{\lambda}b]_{\lambda+\mu}c](w),
\end{array}
\end{equation}
while $ F_{z,w}^{\lambda} F_{x,w}^{\mu} [b(x),[a(z),c(w)]]  =  [b_{\mu}[a_{\lambda}c]](w)$ is just by definition.
\end{proof}
 
\begin{defn}
A \emph{Lie conformal algebra} (LCA)  is an $\mathbb{F}[\partial]$-module $R$ endowed with an $\mathbb{F}$-bilinear $\lambda$-bracket $[a_{\lambda}b]\in R[\lambda]$ for $a, b\in R$, which satisfies the axioms of sesquilinearty, skewsymmetry and the Jacobi identity. 
\end{defn}
\begin{exm}
\label{ex4.1}
The Virasoro formal distribution Lie algebra from Example \ref{ex2.1} gives rise, by Proposition \ref{prop4.4}, to the Virasoro Lie conformal algebra
\begin{equation}
\label{e4.23}
\Vir = \mathbb{F} [\partial] L \oplus \mathbb{F}C 
\end{equation}
with $ \lambda $-bracket
\[ [L_\lambda L] = (\partial + 2 \lambda)L + \frac{\lambda^3}{12}C, \quad [C_\lambda \Vir ] = 0 \, . \]
Replacing $ L $ by $ L - \tfrac{1}{2} \alpha C, $ where $ \alpha \in \mathbb{F}, $ we obtain a $ \lambda $-bracket with a trivial cocycle added:
\begin{equation}
\label{e4.24}
[L_\lambda L] = (\partial + 2 \lambda)L + \alpha \lambda C +  \frac{\lambda^3}{12}C, \quad [C_\lambda \Vir ] = 0.
\end{equation}
\end{exm}

\begin{exm}
\label{ex4.2}
The Kac-Moody formal distribution Lie algebra from Example \ref{ex2.2} gives rise to the Kac-Moody Lie conformal algebra
\begin{equation}
\label{e4.25}
\mathrm{Cur} \, \mf{g} = \mathbb{F}[\partial] \otimes \mf{g} + \mathbb{F}K
\end{equation}
with $ \lambda $-bracket $ (a,b \in \mf{g}): $
\[ [a_\lambda b] = [a,b]+ \lambda (a|b)K, \quad [K_\lambda \mathrm{Cur} \, \mf{g}] = 0. \]
Fix $ s \in \mf{g}; $ replacing $ a $ by $ a - (a|s)K,  $ we obtain a $ \lambda $-bracket with a trivial cocycle added:
\begin{equation}
\label{e4.26}
[a_\lambda b] = [a,b]+ \lambda (a|b)K +(s|[a,b])K, \quad [K_\lambda \mathrm{Cur} \, \mf{g}] = 0.
\end{equation}

Of course, adding a trivial cocycle doesn't change the Lie conformal algebra. However this will become crucial in the proof of the integrability of the associated integrable systems. 
\end{exm}

Due to the $n$-th product identity in a vertex algebra (Corollary 3.1(c)), we derive from the last proposition the following.

\begin{prp}
A vertex algebra $V$ is a Lie conformal algebra with $\partial = T$, the translation operator, and  $\lambda$-bracket 
\begin{equation}
\label{lambdabra}
[a_{\lambda}b]=\sum_{n\geq 0} \frac{\lambda^{n}}{n!}a_{(n)}b,\, a, b\in V. 
\end{equation} 
\end{prp}

\begin{proof}
The $\lambda$-bracket defined by \eqref{lambdabra} is the formal Fourier transform of the $z$-product in $V$. $V$ is obviously an $\F[T]$-module. Moreover, the Fourier coefficients of the formal distributions $\{a(w) | a\in V\} \subset \End V [[w, w^{-1}]]$ span a Lie subalgebra of $\Lie_V$ of $ \End \, V $ (Corollary 3.2), and they are pairwise local, hence the skewsymmetry is always satisfied. Thus, $ (\Lie_V,  \{ a(w)\}_{a \in V} ) $ is a formal distribution Lie algebra. Hence, by Proposition 4.4, the formal distributions $ \{ a(w)\}_{a \in V} $ satisfy all axioms of a Lie conformal algebra. Due to the $ n $-th product identity, Proposition 4.5 follows. 
\end{proof}
We thus obtain the following
\begin{thm}
\label{quantumpoisson}
Let $V$ be a vertex algebra. Then the quintuple $(V,\vac,T, ::, 
[._\lambda .])$ satisfies the following properties of a ``quantum Poisson vertex algebra''.
\begin{itemize}
\item [(a)] $(V, T, [ \cdot_\lambda \cdot]) $ is a Lie conformal algebra.
\item [(b)] $(V, \vac, T, ::)$ is a quasicommutative, quasiassociative unital differential algebra.
\item [(c)] The normally order product :: and the $ \lambda $-bracket $[\cdot_\lambda \cdot]$ are related by the  noncommutative Wick formula \eqref{noncommwick}.
\end{itemize}
\end{thm}


\begin{rmk}
In fact, properties  (a), (b), (c) of Theorem 4.1 characterize a vertex algebra structure, i.e., a quintuple $(V, \vac, T, ::, [ \cdot_\lambda \cdot ])$ satisfying the above ``quantum Poisson vertex algebra'' properties, is a vertex algebra. This is proved in \cite{BK03}.
\end{rmk}

\begin{exm}
\label{freebosonconformal}
(Computation with the non-commutative Wick's formula)
The simplest example is a free boson. Recall Example \ref{Ex:Free-Boson} in Lecture 1. For a free boson field $a(z)$, we have 
\begin{equation}
[a(z),a(w)]=\partial_w\delta(z, w).
\end{equation}
In the language of $\lambda$-brackets this means   for $ a = f\!s(a(z)): $
\begin{equation}
\label{4.19}
[a_{\lambda}a]=\lambda \vac,
\end{equation}
i.e., $a_{(1)}a=1$ and 
$a_{(n)}a=0$ for $n=0$ or $n\geq 2$. 

Now let $L:=\frac{1}{2}:aa:$, then 
\begin{equation}
\label{4.20}
[L_{\lambda}a]=(T+\lambda)a, \qquad [L_{\lambda}L]=(T+2\lambda)L+\frac{\lambda^{3}}{12}\vac.
\end{equation}
Indeed,
\[ 2[a_{\lambda} L] = [a_\lambda : aa:] = :[a_\lambda a] a: + :a[a_\lambda a]: + \int_{0}^{\lambda} [[a_\lambda a]_\mu a] d \mu. \]
Using \eqref{4.19}, we obtain $ [a_{\lambda}L] = \lambda a $ (since $ [\vac_\lambda a ] = 0$). By the skewsymmetry of the $\lambda$-bracket, the first equation in \eqref{4.20} follows. 
\end{exm}

Next we have:
\begin{align*}
[L_{\lambda} L] & = \frac{1}{2} [L_{\lambda} :aa:] \\
& = \frac{1}{2} :[L_{\lambda} a]a: + \frac{1}{2} :a[L_{\lambda}a]: + \int_{0}^{\lambda} [[L_{\lambda}a]_\mu a] \, d \mu \\
& = \frac{1}{2} :((T + \lambda) a)a: + \frac{1}{2} :a (T + \lambda)a : + \int_{0}^{\lambda} [(T + \lambda) a_{\mu} a] \, d \mu  \\
& = \frac{1}{2} T(:aa:) + \lambda :aa: + \int_{0}^{\lambda} (\lambda - \mu) \mu \, d \mu \ \vac \\
& = (T + 2 \lambda) L + \frac{\lambda^3}{12} \ \vac,
\end{align*}
\noindent proving the second equation in \eqref{4.20}.

Of course, there is a simpler way of manipulating with free quantum fields, see Theorem 3.3 in \cite{VAB}. However, exactly the same method as above works well for arbitrary quantum fields (like currents, discussed below).

The following proposition tells us how to prove that a vector $ L $ is a conformal vector, hence how to construct a Hamiltonian operator $ H = L_0 $. 
  
\begin{prp}
\label{P4.6}
 Let $(V,\vac, T, \mathcal{F})$ be a pre-vertex algebra and let $L\in V$ be such that for $a(z)\in \mathcal{F}$,
 \begin{itemize}
    \item [(i)] $[L_{\lambda}a]=(T+\triangle_{a}\lambda)a+o(\lambda)$
    \item [(ii)] $L(z)$ satisfies the Virasoro relation: $[L_{\lambda}L]=(T+2\lambda)L+
    \frac{\lambda^{3}}{12}c\vac$.
 \end{itemize}
Then $L$ is a conformal vector of the corresponding (by the Extension theorem) vertex algebra.
\end{prp} 
 
\begin{proof}
$L(z)$ is already a Virasoro field, so we only need to prove that $L_{-1}=T$ and that $L_0$ acts diagonalizably on $V$. By completeness, $V$ is spanned by $a_{(k_{1})}^{j_1}\cdots a_{(k_{s})}^{j_s}\vac$, 
where $a^{j_i}(z)\in \mathcal{F}$.
Furthermore, property $(i)$ tells us 
\begin{equation}
\label{l0l1}
[L_{-1},a_{(n)}]=-na_{(n-1)}\quad\text{and}\quad [L_{0},a_{(n)}]=(\Delta_{a}-n-1)a_{(n)}.
\end{equation}
Moreover, letting $a=\vac$ in $(i)$, we get
\begin{equation}
L_{-1}\vac=0 \quad\text{and}\quad L_{0}\vac=0.
\end{equation}
Remember that $T$ also satisfies the first equation in \eqref{l0l1}, so $[L_{-1}-T, a_{(k)}]=0$ for all $k\in \Z$. Moreover $(L_{-1}-T)\vac =0$, so $L_{-1}-T$, being a derivation of all $ n $-th products, is zero, i.e., $L_{-1}=T$. $L_{0}$ is diagonizable by \eqref{l0l1}.
\end{proof}



It follows from Proposition \ref{P4.6} and \eqref{4.20} that $ L $ is a conformal vector for the free boson vertex algebra, the free boson $ a  $ being primary  of conformal weight 1. Exactly the same method works for the affine vertex algebras. 

\begin{exr}
Let  $V^k(\mathfrak{g})$ be the universal affine vertex algebra of level $ k $ associated to a simple  Lie algebra $\mathfrak{g}$. Let $a^{i},b^{i}$ be dual bases of $\mathfrak{g}$, i.e., $(b^{i}|a^{j})=\delta_{ij}$ with respect to the Killing form. Assume that $ k \neq -h^{\vee}, $ where $ 2h^{\vee} $ is the eigenvalue of the Casimir element of $U(\mf{g}) $ in the adjoint representation 
($h^\vee$ is called the dual Coxeter number). Let $L=\frac{1}{2(k+h^{\vee})}\sum_i:{a^ib^{i}}:$ (the so called Sugawara construction). Show that $ L $ is a conformal vector  with central charge $c=\frac{k \dim \mathfrak{g} }{2(k+h^{\vee})}$, all $a\in \mathfrak{g}$ being primary of conformal weight $1$.
\end{exr}


\subsection{Quasiclassical limit of vertex algebras.}

Suppose we have a family of vertex algebras, i.e. a  vertex algebra $(V_{\hbar}, T_\hbar, \vac_\hbar, ::_{\hbar}, [\cdot_\lambda \cdot]_{\hbar})$ over $\F[[\hbar]]$, such that
 \begin{itemize}
  \item [(i)] for $v\in V_\hbar$, $\hbar v=0$ only if $v=0$ (e.g. if $V_{\hbar}$ is a free $\F[[\hbar]]$-module),
  \item [(ii)] $[a_{\lambda}b]_{\hbar}\in \hbar V_{\hbar}$ for $a, b\in V_\hbar$.
 \end{itemize}

Given a vertex algebra $(V_{\hbar}, T_\hbar, \vac_\hbar, ::_{\hbar}, [_\lambda]_{\hbar})$ over $\F[[\hbar]]$, satisfying the above two conditions, let $\mathcal{V}:=V_{\hbar}/\hbar V_{\hbar}$. This is a vector space over $ \mathbb{F}. $ Denote by 1 the image of $ \vac_{\hbar} \in V_{\hbar} $ in $ \mathcal{V} $ and by $ \partial  $ the operator on $ \mathcal{V}, $  induced by $ T \in \End \, V_{\hbar} $ ($ \hbar V_{\hbar} $ is obviously $ T $-invariant). The subspace (over $ \mathbb{F} $) $ \hbar V_{\hbar} $ is obviously an ideal for the product $ :v:_{\hbar}, $ hence we have the induced product $ \cdot  $ on $ \mathcal{V}, $ which is bilinear over $ \mathbb{F}. $ Finally, define a $ \lambda $-bracket $ \{ a_{\lambda} b \} $ on $ \mathcal{V} $ as follows. Let $ \tilde{a} $ and $ \tilde{b}  $ be preimages in $ V_{\hbar} $ of $ a $ and $ b; $ then we have
\[ [\tilde{a}_{\lambda} \tilde{b}]_{\hbar} = \hbar [\tilde{a}_{\lambda} \tilde{b}]' \]
where $ [\tilde{a}_{\lambda} \tilde{b}]' $ is uniquely defined due to (i) and (ii). We let 
\[ \{ a_{\lambda} b \} = \mbox{ image of } [\tilde{a}_{\lambda} \tilde{b}]' \mbox{ in } \mathcal{V}. \]
Obviously this $ \lambda $-bracket is independent of the choices of the preimages of $ a $ and $ b. $

\begin{defn}
\label{D4.3}
The \textit{quasiclassical limit} of the family of vertex algebras $ V_{\hbar} $ is the quintuple $ (\mathcal{V}, 1, \partial, \cdot, \{ \cdot_\lambda \cdot \} ). $
\end{defn}

 

\begin{defn}
\label{poissonva}
A \emph{Poisson vertex algebra} is a quintuple $(\mathcal{V}, \vac, \partial, \cdot, \{\cdot_{\lambda} \cdot \})$ which satisfies the following axioms,

\begin{itemize}
  \item [(A)] $(\mathcal{V}, \partial, \{_{\lambda}\})$ is a Lie conformal algebra,
  \item [(B)] $(\mathcal{V}, 1, \partial, \cdot)$ is a commutative associative unital differential algebra,
  \item [(C)] $\{a_{\lambda} bc \}=\{a_{\lambda}b\}c+b\{a_{\lambda}c\}$ for all $ a, b,c \in \mathcal{V} $ (left Leibniz rule).
\end{itemize}
\end{defn}

\begin{thm}
The quasiclassical limit $\mathcal{V}$ of the family 
of vertex algebras $V_\hbar$ is a Poisson vertex algebra.
\end{thm}

\begin{proof}
Since $V_\hbar$ is a vertex algebra over $\F[[\hbar]]$, due to Theorem  \ref{quantumpoisson} we have the quasicommutativity formula, the quasiassociativity formula and the non-commutative Wick formula for representatives in $ V_\hbar $ of elements of 
$\mathcal{V}$
. After taking the images of these formulas in $\mathcal{V}$, the ``quantum corrections'' disappear, hence $\mathcal{V}$ satisfies properties 
(B) and (C) of a PVA. Property (C) is satisfied as well since the axioms of a Lie conformal algebra are homogeneous in its elements.
\end{proof}

\begin{exr}
\label{exr4.2}
Deduce from the left Leibniz rule and the skewcommutativity of the $ \lambda $-bracket of a Poisson vertex algebra, the right Leibniz rule:
\[ \{ ab_\lambda c \} = \{ b_{\lambda + \partial} c \}_\rightarrow a + \{ a_{\lambda + \partial} c \}_\rightarrow b. \]
\end{exr}

Given a Lie algebra $\mathfrak{g}$, we can associate to it two structures: the universal enveloping algebra $U(\mathfrak{g})$ and the Poisson algebra $S(\mathfrak{g})$. The Poisson bracket on $S(\mathfrak{g})$ is the extension of $\{a,b\}=[a,b]$ for all $a,b\in\mathfrak{g}$ by left and right Leibneiz rules. In fact, $S(\mathfrak{g})$ is the quasiclassical limit of $U(\mathfrak{g}_{\hbar})$, where $\mathfrak{g}_\hbar$ is the Lie algebra $\mathbb{F}[[\hbar]]\otimes\mathfrak{g}$ over $\mathbb{F}[[\hbar]]$ with bracket $[a,b]_{\hbar} =\hbar[a,b]$  for $a, b \in\mathfrak{g}$. Indeed it is easy to see that the ordered monomials in a basis of $\mathfrak{g}$ form a basis of $U(\mathfrak{g_\hbar})$ over $\mathbb{F}[[\hbar]]$. Hence $U(\mathfrak{g}_{\hbar})/\hbar  U(\mathfrak{g}_{\hbar})=S(\mathfrak{g})$ as associative algebras, and $\{a,b\}=\dfrac{[\tilde{a}, \tilde{b}]_\hbar}{\hbar} \bigg{|}_{\hbar =0} = [a,b]$ for all $a,b\in\mathfrak{g}$ defines the Poisson structure on $S(\mathfrak{g})$.

Similar picture holds if in place of a Lie algebra $\mathfrak{g}$ we take a Lie conformal algebra $R$, and in place of $U(\mathfrak{g})$ we take $V(R),$ its universal enveloping vertex algebra. Recall its construction. We have the ``maximal'' formal distribution Lie algebra $ (\Lie R, R) $, associated to $ R,  $ which is regular (see \cite{VAB}, Chapter 2). Then $ V(R) = V(\Lie R, R) $ (for another construction, entirely in terms of $ R, $ see \cite{DSK06}).
 
Consider the vertex algebra  $V(R_{\hbar})$ over $\mathbb{F}[[\hbar]]$, where $R_{\hbar}=R[[\hbar]]$ for the Lie conformal algebra $R$ over $\F$, with $\lambda$-bracket defined by $[a_{\lambda}b]_{\hbar}= \hbar[a_{\lambda}b]$ for $a, b\in R$. In the same way as in the Lie algebra case, the quasiclassical limit is the Poisson vertex algebra, which, as a differential algebra, is $ S(R)$ 
(the symmetric algebra of the $ \mathbb{F} $-vector space $ R $) with $ \partial, $ extended as its derivation, endowed with the $ \lambda $-bracket $ \{a_\lambda b \} = [a_\lambda b] $ on $ R, $ which is extended to $ S(R)  $ by the left and right Leibniz rules. 

\subsection{Representations of vertex algebras and Zhu algebra} 
We have the following diagram
\begin{center}
\begin{tikzpicture}
     \node (D) at (0,0) {$VA$};
     \node[left=of D] (F) {$PVA$};
     \node[below=of D] (A) {$AA$};
     \node[below=of F] (B) {$PA$};
      draw arrows and text between them
     \draw[->] (D)--(F) node [midway,above] {q.lim};
        node [midway,above] { };
    
     \draw[->] (F)--(B) node [midway,left] {\Zhu} 
                node [midway, left] {};
     \draw[->] (A)--(B) node [midway,below] {} 
               node [midway,above] {q.lim};
     \draw[->] (D)--(A) node [midway,right] {\Zhu};     
\end{tikzpicture}
\end{center}

\noindent In the diagram, $AA$ means associative algebras, $PA$ means Poisson algebras, $VA$ means vertex algebras and $PVA$ means Poisson vertex algebras; q.lim means the quasiclassical limit and Zhu means a functor from vertex algebras to associative algebras (resp. from Poisson vertex algebras to Poisson algebras), explained below. 

Let $V$ be a vertex algebra with a Hamiltonian operator $H$. Throughout this section we will assume (for simplicity) that all eigenvalues of $H$ are integers. Recall the Borcherds identity from Section 3.2. For $ a $ and $ b \in V $ with eigenvalues of $ H $ equal $ \Delta_a $ and $ \Delta_b $ respectively, we write 

\[ a(z) =\underset{n\in \mathbb{Z}}{\sum}a_{n}  z^{-n-\Delta_a }, \quad  b(z)  = \sum_{a \in \mathbb{Z}} b_n z^{-n -\Delta_b } \ .  \]
Then, comparing the coefficients of monomials in $ z $ and $ w $ in the Borcherds identity we have, for $ m, n, k\in \Z $:
\begin{equation}
\label{4.29}
\sum_{j\geq 0}{k \choose j}(-1)^{j}(a_{m+k-j}b_{n+j}-(-1)^{n} b_{n+k-j}a_{m+j})=\underset{j\geq0}{\sum} {m+\triangle_{a}-1 \choose j}(a_{(k+j)}b)_{m+n+k} \, .
\end{equation}

\begin{defn}
A \emph{representation} of the vertex algebra $V$ in a vector space $M$ is a linear map 
\begin{equation}
\label{4.30}
 V\longrightarrow (\End M) [[z,z^{-1}]], \quad a\longmapsto a^{M}(z)=\sum_{n\in \Z} a_{n}^{M}z^{-n-\triangle_{a}},
\end{equation}
defined for eigenvectors of $ H $ and then extended linearly to  $V$, such that,
\begin{itemize}
   \item [(i)]  $a^{M}(z)$ is an $\End M$-valued quantum field for all $a \in V$ (i.e., given $m\in M$, $a_{(n)}^Mm=0$ for $n\gg 0$), 
   \item [(ii)] $\vac^{M}(z)=I_M$,
   \item [(iii)] Borcherds identity holds, i.e., for $a, b \in V, \ c \in M, \  m, n, k \in \Z$ we have (cf. \eqref{4.29}):
   \begin{equation}
   \label{4.31}
   \sum_{j\geq 0}{k \choose j}(-1)^{j}(a_{m+k-j}^{M}b_{n+j}^{M} c-(-1)^{n} b_{n+k-j}^{M}a_{m+j}^{M} c)=\underset{j\geq0}{\sum} {m+\triangle_{a}-1 \choose j}(a_{(k+j)}b)_{m+n+k}^M c \, .
   \end{equation}
\end{itemize}
\end{defn}
  
\begin{rmk}
Note that $(Ta)_n=(-n-\Delta_a)a_n$ and $Ha=\Delta_a a$, hence, $((T+H)a)_0=0.$
\end{rmk}  
 
Now assume that our vertex algebra $V$ contains a conformal vector $L$ of central charge $ c \in \mathbb{F} $ (see Definition \ref{conformalvector}), so that $ L_{-1} = T $ and $ L_0 = H $ is a Hamiltonian operator. Then we have $L^M(z)=\sum_{n\in\Z}L^M_nz^{-n-2}$, and
$[L^M_m, L^M_n]=(m-n)L^M_{m+n}+\delta_{m,-n}\frac{m^3-m}{12}c I_M$ .

\begin{defn}
A \emph{positive energy representation} $M$ of $V$ is a representation with $L^M_0$ acting diagonalizably on $M$ with spectrum bounded below, i.e.,  $M=\oplus_{j\geq h}M_{j}$ for some $h$, where $M_j=\{m\in M| L^M_0m=jm\}$.
\end{defn}
  
By \eqref{e3.38} (which follows from the Borcherds identity) we have 
\begin{equation}
\label{4.32}
a_{n}^{M} M_h=0 \mbox{ for $n>0$}, \ a_{0}^{M} M_h\subset M_h \, .
\end{equation}
So we have a linear map with $(H+T)V$ contained in the kernel (by Remark 4.3):
\begin{equation}
\label{4.33}
\pi_{M}:V \longrightarrow \End M_h, \ a \longmapsto a_0^M|_{M_h} \, .
\end{equation}
Taking $m=1,k=-1,n=0$ in Borcherds identity \eqref{4.31} for $ c \in M_h $, we get, by \eqref{4.32},
\[ \pi_M(a)\pi_M(b)c=\pi_{M}(a\ast b)c, \mbox{ for } a,b \in V,  \]
where 
\begin{equation}
\label{4.34}
a*b:=\sum_{j\geq 0} {\Delta_a \choose j}a_{(j-1)}b \ .
\end{equation}
Thus we get a representation of the algebra $(V, \ast)$ in the vector space $M_h$. The multiplication $ \ast $ on $ V $ is not associative. However, we have the following remarkable theorem.

\begin{thm}[\cite{Zhu}]
\begin{itemize}
 \item [(a)] $ J(V): = ((T + H)V) \ast V $ is a two-sided ideal of the algebra $(V,*)$.
 \item [(b)] $\Zhu V: =(V/J(V),*)$ is a unital associative algebra with 1 being the image of $ \vac $. 
 \item [(c)] The map $M \rightarrow M_h$ induces a map from the equivalence classes of positive energy $V$-modules to the equivalence classes of $\Zhu V$-modules, which is bijective on irreducible modules.
\end{itemize}
\end{thm}

\begin{proof}
We refer for the proof to the original paper \cite{Zhu} or to \cite{DSK06} for a simpler proof of a similar result without the assumption that the eigenvalues of
$H$ are integers.
\end{proof}

\begin{exr}
\label{exer4.2}
Prove the commutator formula in Zhu algebra:

\begin{equation}
[a,b]:=a*b-b*a={\sum}_{j\geq 0} {\Delta_a-1 \choose j }a_{(j)}b
\end{equation}
\end{exr}

\begin{exr}
Let 
$ \mathcal{V}$ 
be a Poisson vertex algebra and let $ H $ be a diagonalizable operator on 
$ \mathcal{V}$, such that
\[ \Delta_{a_{(n)}b} = \Delta_a + \Delta_b - n - 1, \quad \Delta_{\partial a} = \Delta_a+1, \quad \Delta_{ab} = \Delta_a + \Delta_b,  \]
\noindent where $ \Delta_a $ is the eigenvalue of $ a, $ and 
\[ \{ a_{\lambda} b \} = \sum_{n \in \mathbb{Z}_+} \frac{\lambda^n}{n!} a_{(n)}b. \]
Show that $ \Zhu \mathcal{V} := \mathcal{V} / ((\partial + H)\mathcal{V})\mathcal{V} $ is a unital commutative associative algebra with the well defined Poisson bracket (cf. Exercise \ref{exer4.2})
\begin{equation}
 \{a,b \} = \sum_{j \geq 0} \begin{pmatrix}
\Delta_a -1 \\j
\end{pmatrix} a_{(j)} b. 
\end{equation}
\end{exr}
\begin{exr}
Let $V$ (resp. $ \mathcal{V}$) be a vertex algebra (resp. Poisson vertex algebra). Then 
$J=:(TV)V:$ (resp. $J=(\partial\mathcal{V})\cdot \mathcal{V}$) is a two-sided ideal of the algebra 
$(V,::)$ (resp. $(\mathcal{V},\cdot)$), and 
$V/J$ (resp. $\mathcal{V}/J$) is a Poisson algebra with the product, induced
by $::$ (resp $\cdot$), and the well defined bracket, induced by the 
$0$-th product of the $\lambda$-bracket.
 
\end{exr}
Of course, Zhu's theorem is just the beginning of the representation theory of vertex algebras, which has been a rapidly developing field in the past twenty years. Some of the most remarkable results of this theory are presented in the beautiful lecture course by T. Arakawa in this school. 
\newpage

\section{Lecture 5 (January 14, 2015)}
\label{sec:5}

Given a vertex algebra $ V $, one can construct its \textit{quasiclassical limit}. As a result we get a Poisson vertex algebra (PVA). This can be done both considering a filtration of the vertex algebra V or by constructing a one parameter family of vertex algebras $V_{\hbar}$, as previously done in Lecture $ 4 $. This construction resembles the way a Poisson algebra arises as a quasiclassical limit of a family of associative algebras, hence the name ``Poisson'' vertex algebra. The reason we are interested in such structures is that the theory of Poisson vertex algebras has  important relation with the theory of integrable systems of PDE's. This relation is parallel to (but a bit different from) the relation of Poisson algebras with the theory of integrable systems of ODE's. 

\subsection{From finite-dimensional to infinite-dimensional Poisson structures}
Let us start by recalling the definition of a Poisson vertex algebra:

\begin{defn}
A PVA is a quintuple $(\mathscr{V},\,\partial,\,1
,\,\cdot\,,\{\cdot_{\lambda}\cdot\})$ such that:
\begin{enumerate}
\item $(\mathscr{V},\,\partial,\,1,\,\cdot\,) $ is a differential algebra;
\item $(\mathscr{V},\,\partial,\{\cdot_{\lambda}\cdot\})$ is a Lie conformal algebra, whose $\lambda$-bracket satisfies the following axioms:
\begin{itemize}
\item[(i)](sesquilinearity) $\{\partial a_{\lambda}b\} = - \lambda\{a_{\lambda}b\}$, \quad\quad $\{a_{\lambda}\partial b\} =(\partial + \lambda)\{a_{\lambda}b\}$;
\item[(ii)](skewsymmetry) $\{b_{\lambda}a\} = - \{a_{-\partial - \lambda}b\}$;
\item[(iii)](Jacobi identity) $\{a_{\lambda}\{b_{\mu}c\}\} - \{b_{\mu}\{a_{\lambda}c\}\} = \{\{a_{\lambda}b\}_{\lambda + \mu}c\}$;
\end{itemize}
\item $\{\cdot_{\lambda}\cdot\} $ and $ \cdot $ are related by the following Leibniz rules:
\begin{itemize}
\item[(i)](left Leibniz rule) $ \{a_{\lambda}bc\} = \{a_{\lambda}b\}c + b\{a_{\lambda}c\}$;
\item[(ii)](right Leibniz rule) $ \{ab_{\lambda}c\} = \{a_{\lambda + \partial}c\}_{\rightarrow}b + \{b_{\lambda + \partial}c\}_{\rightarrow}a$.
\end{itemize}
\end{enumerate}
\end{defn}

\begin{rmk} We use the following notation: if $ \{a_{\lambda}b\} = \sum\limits_{n \in \bZ_+} \frac{{\lambda}^n}{n!} a_{(n)}b $, then when a right arrow appears it means that $ \lambda + \partial $ has to be moved to the right: $ \{a_{\lambda + \partial}b\}_{\rightarrow}c = \sum\limits_{n \in \bZ_+} \frac{a_{(n)}b}{n!}(\lambda+ \partial)^n c$. However, if no arrow appears we just have $ \{a_{-\partial - \lambda}b\} = \sum\limits_{n \in \bZ_+} \frac{(-\lambda - \partial)^n}{n!} a_{(n)}b $.
\end{rmk}


In the theory of Hamiltonian ODEs the key role is played by the Poisson bracket on the space of smooth functions $\mathcal{F}$ on a manifold. Choosing local coordinates $ u_1,\ldots,u_\ell $ on the manifold, we can endow $\mathcal{F}$ with a structure of Poisson algebra, letting
\begin{equation}
\label{5.1}
\{u_j,u_i\} = H_{ij} \in \mathcal{F}.
\end{equation}
By the Leibniz rule this extends to polynomials in the variables $ u_i $ as follows:
\begin{equation}\label{masterformulazero}
\{f,g\} = \frac{\partial g}{\partial u} \cdot H\frac{\partial f}{\partial u},
\end{equation}
where $\frac{\partial f}{\partial u} = \left( \begin{array}{c}
\frac{\partial f}{\partial u_1} \\
\vdots \\
\frac{\partial f}{\partial u_\ell} \end{array} \right)$, $u =\left( \begin{array}{c}
 u_1 \\
\vdots \\
 u_\ell \end{array} \right)$,  
$ H = (H_{ij})^\ell_{\substack{i,j=1}} $ is an $ \ell \times \ell $ matrix with coefficients in $ \mathcal{F} $, and $\cdot$ is the usual dot product of vectors from $ \mathcal{F}^\ell $ with values in $ \mathcal{F}. $ Formula \eqref{masterformulazero} extends to arbitrary functions $ f, g \in \mathcal{F} $. This bracket obviously satisfies the Leibniz rule, but it is not necessarily skewsymmetric, neither it satisfies the Jacobi identity. If the matrix $ H $ is skewsymmetric (i.e. $ H^{T} = - H$), then the bracket (\ref{masterformulazero}) is skewsymmetric. If, in addition, it satisfies the Jacobi identity (which happens iff $[H,H]=0 $, where $ [ \cdot , \cdot ] $ is the Schouten bracket), then the matrix $ H $ is called a \textit{Poisson structure}
on $ \mathcal{F}$.
 
 \begin{defn}
 The \emph{Hamiltonian ODE} associated with this Poisson structure is
 \begin{equation}\label{eq:HamiltonianODE}
 \frac{du}{dt} = \{h,u\} = H\frac{\partial h}{\partial u},
 \end{equation}
where the second equality follows from (\ref{masterformulazero}). The function $ h  \in \mathcal{F} $ is called the \textit{Hamiltonian} of this equation. 
\end{defn}
This is a special case of what is called an \textit{evolution ODE}, that is 
\[ \frac{du}{dt} = F(u),\, \hbox{for some}\,\, F\in \mathcal{F}^{\ell}.
\]

In the theory of Hamiltonian PDEs a similar role is played by PVAs. Let us now see how to construct a similar machinery.

First of all we need to define which kind of differential algebra $\mathscr{V}$ we want for our PVA. The basic example is the \textit{algebra of differential polynomials} in $ \ell $ variables $ 
\mathscr{P}_\ell 
= \F[u_i^{(n)} \,\vert\, i \in I=\{1,\ldots,\ell\},\, n \in \bZ_+] $, which is a differential algebra with derivation $ \partial $, called the \textit{total derivative}, such that $ \partial u_i^{(n)} = u_i^{(n+1)} $. 

\begin{defn}\label{algebradifffunctions}
An \emph{algebra of differential functions} in $ \ell $ variables $ \mathscr{V} $ is a differential algebra with a derivation $\partial$, which is an extension of the algebra of differential polynomials $\mathscr{P}_\ell $, endowed with linear maps $ \frac{\partial}{\partial u_i^{(n)}}: \mathscr{V} \rightarrow \mathscr{V} $ for all $ i \in I $, $n \in \bZ_+ $, which are commuting derivations of $ \mathscr{V}$, extending the usual partial derivatives in $ \mathscr{P}_\ell $, and satisfying the following axioms:
\begin{itemize}

\item[(i)] given $f\in \mathscr{V}$, $\frac{\partial f}{\partial u^{(n)}_i} = 0 $ for all but finitely many pairs $ (i,n) \in I \times \bZ_{+} $;
\item[(ii)] $ [\frac{\partial}{\partial u_i^{(n)}}, \partial ] = \frac{\partial}{\partial u_i^{(n-1)}} $ \quad (where the RHS is considered to be zero if $ n = 0 $).
\end{itemize}
\end{defn}

Which differential algebras are algebras of differential functions? The algebra of differential polynomials $ \mathscr{P}_\ell$ itself clearly satisfies these axioms (it suffices to check (ii) on the generatots 
$u_i^{(n)}$). One can as well consider the corresponding field of fractions 
$ \mathscr{Q}_\ell  = \F(u_i^{(n)} \,\vert\, i \in I,\, n \in \bZ_+) $, or any algebraic extension of $ \mathscr{P}_\ell  $ or $\mathscr{Q}_\ell $, obtained by adding a solution of a polynomial equation. However, if we want both axioms to hold, we can not add a solution of an arbitrary differential equation: for example, we can add $e^u$, solution of $f' = fu' $, but we can not add a 
non-zero solution of $ f'= fu $.
\begin{exr}
Let $\mathscr{V}=\mathscr{P}_1 [v]$ with the derivation $\partial$, extended from $\mathscr{P}_1 $
by $\partial v=vu_1$ or by $\partial v=u_1$. Show that the structure of an algebra of differential functions cannot be extended from $\mathscr{P}_1$ to $\mathscr{V}$.  
\end{exr}
The reasons why we want both properties (i) and (ii) to hold will soon be clear.\\

We also want an analogue of the bracket given by (\ref{masterformulazero}) and to understand what a Poisson structure is in the infinite-dimensional case. Recall the following (non-rigorous) formula which appears in any textbook on integrable
Hamiltonian PDE, cf. \cite{TF86}. It defines the Poisson bracket on generators ($i,j \in I $) as
\begin{equation}\label{physicistsformula}
\{u_i(x),u_j(y)\} = H_{ji}(u(y),u'(y),\ldots,u^{(n)}(y);\frac{\partial}{\partial y}) \delta (x-y),	
\end{equation}
where $ H = (H_{ji})_{i,j=1}^\ell $ is an $ \ell \times \ell $ matrix differential operator on $\mathscr{V}^\ell$, the $u_i$'s are viewed as 
functions in $x$ on a one-dimensional manifold, and $ \delta (x-y) $ is the usual delta function. 

\begin{exm}\label{ex:gfz}
The first example is given by the Gardner-Faddeev-Zakharov (GFZ) bracket, for $ \mathscr{V} = \mathscr{P}_1 $, and it goes back to 1971:
\begin{equation}
\{u(x),u(y)\} = \frac{\partial}{\partial y} \delta (x-y).
\end{equation}
\end{exm}

As in the ODE case, we can extend the bracket defined in (\ref{physicistsformula}) by the Leibniz rule. Then, for arbitrary $f,\,g \in \mathscr{V} $ we have
\begin{equation}\label{physicistsleibniz}
\{f(x),g(y)\} = \sum_{i,j\in I,\, p,q \in \bZ_+ } \frac{\partial f}{\partial u_i^{(p)}}\frac{\partial g}{\partial u_j^{(q)}} \partial_x^p \partial_y^q \{u_i(x),u_j(y)\}.
\end{equation}
The basic idea is to introduce the $ \lambda $-bracket by application of the Fourier transform
\begin{equation}
F(x,y) \mapsto \int e^{\lambda (x-y)}F(x,y)\,dx
\end{equation}
to both sides of 
(\ref{physicistsleibniz}):
\begin{equation}
\{f_{\lambda} g\} := \int e^{\lambda (x-y)} \{f(x),g(y)\} dx.
\end{equation}
Thus, for arbitrary $f,g \in \mathscr{V} $, we get a rigorous formula, called
the \textit{Master Formula}:
\begin{equation}\label{masterformula}
\{f_{\lambda}g\} = \sum\limits_{i,j\in I,\, p,q \in \bZ_+} \frac{\partial g}{\partial u_j^{(q)}}(\partial + \lambda)^q\{{u_i}_{\,\partial + \lambda}u_j\}_{\rightarrow}(-\partial - \lambda)^p\frac{\partial f}{\partial u_i^{(p)}} \tag{MF}.
\end{equation}
Here, $ \{{u_j}_{\,\partial + \lambda}u_i\} = H_{ij}(\partial + \lambda) $, where $ H(\partial) = (H_{ij}(\partial) )_{i,j\in I} $ is a matrix differential operator with coefficients in $\mathscr{V} $ for which the $\lambda$-bracket is its symbol.

\begin{exr} Derive (\ref{masterformula}) from (\ref{physicistsleibniz}).
\end{exr}

Note that (\ref{masterformula}) is similar to the formula for the Poisson bracket defined by Equation (\ref{masterformulazero}). In fact, to go from the former to the latter we just put $ \lambda $ and $ \partial $ equal to $ 0 $.
%
%

\begin{thm}[\cite{BDSK}]
\label{thm:bdsk}
Let $ \mathscr{V} $ be an algebra of differential functions in the variables $ \{u_i\}_{i \in I}$. For each pair $ i,j \in I $ choose $ \{{u_i}_\lambda u_j\} = H_{ji}(\lambda) \in \mathscr{V}[\lambda] $. Then
\begin{enumerate}
\item The Master Formula (\ref{masterformula}) defines a $ \lambda$-bracket on $ \mathscr{V} $ which satisfies sesquilinearity, the left and right Leibniz rules, and extends the given $ \lambda$-bracket on the variables 
 $u_i$'s. Consequently, any $\lambda$-bracket on the algebra of differential polynomials, satysfying these properties, is given by the Master Formula.
\item This $\lambda$-bracket is skewsymmetric provided skewsymmetry holds for every pair of variables:
\begin{equation}
\label{e5.10}
\{{u_i}_{\lambda}u_j\} = - \{{u_j}_{-\lambda - \partial}u_i\}, \quad \forall \,i,j\in I.
\end{equation}
\item If this $ \lambda$-bracket is skewsymmetric, then it satisfies the Jacobi identity, provided Jacobi identity holds for every triple of variables:
\begin{equation}
\label{e5.11}
\{{u_i}_{\lambda}\{{u_j}_{\mu} u_k\}\} - \{{u_j}_{\mu}\{{u_i}_{\lambda} u_k\}\} = \{{\{{u_i}_{\lambda} u_k\}}_{\lambda + \mu} u_j\}, \quad \forall\, i,j,k \in I.
\end{equation}
\end{enumerate}
\end{thm}

It follows from Theorem \ref{thm:bdsk} that, if the corresponding conditions 
on the variables $u_i$'s hold, the $\lambda$-bracket defined by the Master Formula (\ref{masterformula}) endows $\mathscr{V} $ with a structure of PVA. As in the finite-dimensional case, this structure is completely defined by $ H(\lambda)=
(H_{ij}(\lambda))\in \text{Mat}_{\ell \times \ell}\mathscr{V}[\lambda]  $. 

\begin{defn}
We say that the matrix differential operator $ H(\partial) 
\in \text{Mat}_{\ell \times \ell}\mathscr{V}[\partial]$ with the symbol
$H(\lambda)$ is a \emph{Poisson structure} if the corresponding $ \lambda$-bracket defines a PVA structure on $ \mathscr{V} $.
\end{defn}

\begin{exr}
The $\lambda$-bracket, given by the Master Formula, is skewsymmetric
if and only if the matrix differential operator $ H(\partial)$
is skewadjoint.
\end{exr}

\begin{exm}
Let $ \mathscr{V}=\mathscr{P}_1 = \F[u,u',u'',\ldots]$. From the GFZ bracket defined in Example \ref{ex:gfz} we get the following $\lambda$-bracket: $\{ u_{\lambda} u\} = \lambda $. The skewsymmetry and the Jacobi identity for the $\lambda$-bracket, given by the Master Formula, are immediate by Theorem \ref{thm:bdsk}. The associated Poisson structure is 
$ H(\partial) = \partial $. This PVA is the quasi-classical limit of the 
family of free boson vertex algebras $ B_{\hbar} $.
\end{exm}

\begin{exm}
\label{ex5.3}
Let $ \mathscr{V}=\mathscr{P}_1 = \F[u,u',u'',\ldots]$. The Magri-Virasoro PVA with central charge $ c \in \F $ is defined by the following $\lambda $-bracket: 
\begin{equation}
 \{u_\lambda u\} = (\partial + 2\lambda)u + c\lambda^3 + \alpha\lambda.
\end{equation} 
Of course, it is straightforward to check that the pair $ u,u $ satisfies \eqref{e5.10} and the triple $ u,u,u $ satisfies \eqref{e5.11}, hence, by Theorem \ref{thm:bdsk}, we get a PVA. It is instructive, however, to give a more conceptual proof. Consider the Lie conformal algebra $ \Vir $ from Example \ref{ex4.1}. Then by Theorem 5.1, $ S(\Vir) $ is a PVA, hence its quotient $ \mathcal{V}^c $ by the ideal, generated by $ C-c, $ is a PVA, which is obviously isomorphic to the Magri-Virasoro PVA. 
The corresponding family of Poisson structures is 
\begin{equation}
\label{MV}
 H(\partial) = u' + 2u\partial + c\partial^3 + \alpha\partial. 
\end{equation} 
These Poisson structures were discovered by Magri; the name is due to its connection to the Virasoro algebra. 
Note that $ \mathcal{V}^c $ is the quasiclassical limit of the family of universal Virasoro vertex algebras $ V^{12c}_\hbar. $
\end{exm}
The following exercise shows that the discrete series vertex algebras $V_c$
with $c$ given by (\ref{2.13}) is a purely quantum effect.
\begin{exr} 
Show that the PVA $\mathcal{V}^c$ is simple 
if $c\neq 0$.
\end{exr} 
\begin{exm}
\label{ex5.4}
Given a vector space $ U, $ denote by $ \mathcal{P}(U) = S(\mathbb{F}[\partial] \otimes U) $ the algebra of differential polynomials over $U. $ Let $ \mf{g}, (.\, | \, .) $ be as in Example \ref{ex2.2}, let $ k \in \mathbb{F}, $ and fix $ s \in \mf{g}. $ Then the associated {\it affine PVA} 
$ \mathcal{V}^k (\mf{g},s)$ is defined as the algebra of differential  
polynomials $ \mathcal{P} (\mf{g})$, endowed with the $ \lambda $-brackets $ (a,b \in \mf{g}): $
\begin{equation}
\label{e5.13}
\{ a_\lambda b \} = [a,b] + \lambda (a|b) k + (s|[a,b])1.
\end{equation}
The two proofs from Example \ref{ex5.3} apply to show that 
$ \mathcal{V}^k (\mf{g},s)$  
is a PVA. Of course, up to isomorphism, it is independent of $ s, $ but the trivial cocycle is important for the associated integrable system, since we get a multiparameter family of Poisson structures. Note that $ \mathcal{V}^k (\mf{g},s) $ is the quasiclassical limit of $ V^k_\hbar(\mf{g}). $
\end{exm}

%
Now we recall how one passes from the definition of a Hamiltonian ODE to that of a Hamiltonian PDE. 
The following idea goes back to the $1970's$: in order to get an ``honest''
Lie algebra bracket, we should not consider the whole algebra of differential functions $ \mathscr{V}$, but its quotient $ \mathscr{V}/\partial\mathscr{V} $,
which is not an algebra anymore, just a vector space. Denote by $ \int $ the 
quotient map $ \int: \mathscr{V} \rightarrow 
\mathscr{V}/\partial\mathscr{V} $. 
The corresponding bracket is defined by
\begin{equation}\label{eq:liebracket1}
\{\smallint f, \smallint g\} = \int \frac{\delta g}{\delta u}\cdot H(\partial) \frac{\delta f}{\delta u},
\end{equation}
where $ \frac{\delta f}{\delta u} $ is the vector of \textit{variational derivatives} of $ f $: 
\[\frac{\delta f}{\delta u_i} = \sum\limits_{n \in \bZ_+} (-\partial)^n \frac{\partial f}{\partial u_i^{(n)}}.\]
Elements $\smallint f \in\mathscr{V}/\partial\mathscr{V} $ are called \textit{local functionals}.
 
Equation (\ref{eq:liebracket1}) is analogous to Equation (\ref{masterformulazero}), with variational derivatives instead of partial derivatives, and a matrix differential operator $H(\partial) $ instead of a matrix of functions. It is rather difficult to prove directly that
(\ref{eq:liebracket1}) 
is a Lie algebra bracket on $ \mathscr{V}/\partial\mathscr{V}$. 
The connection to the PVA theory, explained further on, makes it very easy.

The following exercise shows that (\ref{eq:liebracket1}) is well defined.
\begin{exr} 
The variational derivative $ \frac{\delta f}{\delta u} $ depends only on the 
image of $ f \in \mathscr{V} $ in the quotient space $ \mathscr{V}/\partial\mathscr{V} $, since $ \frac{\delta}{\delta u} \circ \partial = 0 $. 
Deduce the latter fact from axiom (ii) in the Definition \ref{algebradifffunctions} of an algebra of differential functions. 
\end{exr}
Given a local functional $\smallint h$, in analogy with (\ref{eq:HamiltonianODE}), one defines the associated \textit{Hamiltonian PDE} as the following 
evolution PDE:  
\begin{equation} \frac{du}{dt} 
= H(\partial) \frac{\delta \smallint h}{\delta u}.
\end{equation}
The local functional $\smallint h$ is called the \textit{Hamiltonian}
of this equation.

We shall explain further on how these classical definitions fit nicely in the 
framework of Poisson vertex algebras.

\subsection{Basic notions of the theory of integrable equations.}

An evolution equation in the infinite-dimensional case is quite the same as in the finite-dimensional case, except it is a partial differential equation.

\begin{defn}
Let $\mathscr{V}$ be an algebra of differential functions in $\ell $ variables $ u_1,\ldots,u_\ell $. An \emph{evolution PDE} is
\begin{equation}\label{evolutionPDE}
\frac{du}{dt} = F(u,u',\ldots, u^{(n)}),
\end{equation}
where $u =\left( \begin{array}{c}
 u_1 \\
\vdots \\
 u_\ell \end{array} \right)$ and $F = \left( \begin{array}{c}
 F_1 \\
\vdots \\
 F_\ell \end{array} \right) \in \mathscr{V}^\ell $. Here, $ u_i = u_i(x,t) $ is a function in one independent variable $ x $, and the parameter $ t $ is called \textit{time}.
\end{defn}

Given an arbitrary differential function $ f \in \mathscr{V} $, by the chain rule we have
\begin{equation}
\frac{df}{dt} = \sum\limits_{i\in I,\, n \in \bZ_+} \frac{d( u_i^{(n)})}{dt} \frac{\partial f}{\partial u_i^{(n)}}.
\end{equation}
Since, by (\ref{evolutionPDE}), we have $ \frac{d( u_i^{(n)})}{dt} = \partial^{n} F_i $, the function $f$  evolves in virtue of Equation (\ref{evolutionPDE}) as
\[ \frac{df}{dt} = X_{F} f, \] 
where
\begin{equation}
X_{F} = \sum\limits_{i \in I,\, n \in \bZ_+}(\partial^{n}F_i) \frac{\partial}{\partial u_i^{(n)}}
\end{equation}
is a derivation of the algebra $ \mathscr{V} $, called the \textit{evolutionary vector field} with characteristic $ F \in \mathscr{V}^\ell $.
It is now clear why Axiom ($i$) in Definition \ref{algebradifffunctions} is 
important: otherwise, the evolutionary vector field would give a divergent sum when applied to arbitrary functions $ f \in \mathscr{V} $.

An important notion in the theory of integrable systems is 
\textit{compatibility} of evolution equations:

\begin{defn}
Equation (\ref{evolutionPDE}) is called \textit{compatible} with the evolution PDE
\begin{equation}\label{complatiblePDE}
\frac{du}{d\tau} = G(u,u',\ldots, u^{(m)}) \in \mathscr{V}^\ell
\end{equation}
where, as before, $u =\left( \begin{array}{c}
 u_1 \\
\vdots \\
 u_\ell \end{array} \right)$ and $G = \left( \begin{array}{c}
 G_1 \\
\vdots \\
 G_\ell \end{array} \right) \in \mathscr{V}^\ell $, if the corresponding flows commute, that is if $ \frac{d}{dt}\frac{d}{d\tau}f = \frac{d}{d\tau}\frac{d}{dt}f $ holds for every function $ f \in \mathscr{V} $.
\end{defn}

By the above discussion, the compatibility of evolution equations
(\ref{evolutionPDE}) and (\ref{complatiblePDE}) is equivalent to the 
property that the corresponding evolutionary vector fields commute:
$ [X_{F},X_{G}] = 0 $, which is a purely Lie algebraic condition. In fact, we can easily see that the commutator of two evolutionary vector fields is again 
an evolutionary vector field. This follows from the next exercise. 
\begin{exr}
Prove that $ [X_{F},X_{G}]= X_{[F,G]} $, where $ [F,G]:= X_{F}G - X_{G}F $.
\end{exr}
Thus, the bracket $ [F,G] = X_{F}G - X_{G}F  $ endows $ \mathscr{V}^\ell $ with a Lie algebra structure, called the \textit{Lie algebra of evolutionary vector fields}.

If two evolutionary vector fields commute, then each of them is called a \textit{symmetry} of the other. So if $ [X_{F}, X_{G}] = 0 $, $ F $ is a symmetry of $ G $ and  $ G $ is a symmetry of $ F $. Note that every evolutionary vector field commutes with $\partial = X_{u'} = \sum\limits_{ i\in I,\, n \in \bZ_+}u_i^{(n+1)} \frac{\partial}{\partial u_i^{(n)}} $.

Let us now introduce the notion of \textit{integrability} for an evolution equation.
\begin{defn}
Equation (\ref{evolutionPDE}) is called 
\emph{Lie integrable} if $ X_{F} $ is contained in an infinite-dimensional abelian subalgebra of the Lie algebra $ \mathscr{V}^\ell $.
\end{defn}
\begin{rmk}
Informally, one says that equation (\ref{evolutionPDE}) is Lie integrable if it admits infinitely many commuting symmetries. 
\end{rmk}


\begin{exm}
\label{ex5.5}
The linear equations over $ \mathscr{P}_1 $:
\[ u_t = u^{(n)}, \ n \in \mathbb{Z}_+, \]
are Lie integrable. Indeed, $ X_{u^{(m)}} (u^{(n)}) = u^{(m+n)} $ is symmetric in $ m $ and $ n, $ hence the corresponding evolutionary vector fields commute.
\end{exm}
\begin{exm}
\label{ex5.6}
The dispersionless equations over $\mathscr{P}_1 $:
\[ u_t = f(u)u', \,\,\,\, f(u) \in \mathscr{P}_1, \]
are Lie integrable, since
\[ X_{f(u)u'} (g(u)u') = \frac{\partial}{\partial u} (f(u)g(u))u'^{2} + f(u)g(u) u'' \]
is symmetric in $ f $ and $ g, $ hence the corresponding evolutionary vector fields commute. 
\end{exm}

The motivation for the definition of Lie integrability of PDE's comes 
from a theorem of Lie in the theory of ODE's, saying that if the evolution
ODE in $\ell$ variables  $\frac{du}{dt} = F(u) $ posesses $\ell $ commuting symmetries with a non-degenerate Jacobian, then it can be solved in quadratures. Of course, in the PDE case the number of coordinates is infinite, therefore we need to require infinitely many commuting symmetries.

There has been a lot of work trying to establish integrability of various partial differential equations. One well-known method of constructing symmetries of an evolution equation is called \textit{recursion operator}; however, in all examples the recursion operator is actually a pseudodifferential operator (which is an element of $ \mathscr{V}(({\partial}^{-1}))$), hence it can not be applied to functions, as Exercise 5.1 demonstrates. We will discuss a different approach, the \textit{Hamiltonian} approach, which is completely rigorous.

We shall deduce Lie integrability from the stronger \textit{Liouville integrability} of Hamiltonian PDE, which, analogously to the definition for ODEs, requires the existence of infinitely many integrals of motion in involution. 

\subsection{Poisson vertex algebras and Hamiltonian PDE}

In order to translate the traditional language of Hamiltonian PDE's, discussed above, to the
language of PVA's, and also, to connect the two notions of integrability, 
the following simple lemma is crucial.

\begin{lem}[Basic lemma]
Let $ \mathscr{V} $ be a PVA. Let $ \bar{\mathscr{V}} := \mathscr{V}/\partial \mathscr{V} $ and let $ \int: \mathscr{V} \rightarrow  \bar{\mathscr{V}} $ be the corresponding quotient map. Then we have the following well-defined brackets:
\begin{itemize}
\item[(i)] $ \bar{\mathscr{V}} \times \bar{\mathscr{V}} \longrightarrow \bar{\mathscr{V}}, \qquad \{\int a, \int b\} := \int \{a_{\lambda} b\}_{\lambda = 0} $,
\item[(ii)] $ \bar{\mathscr{V}} \times \mathscr{V} \longrightarrow \mathscr{V}, \qquad \{\int a, b\} := \{a_{\lambda} b\}_{\lambda=0} $.
\end{itemize}
Moreover, ($ i $) defines a Lie algebra bracket on $ \bar{\mathscr{V}} $, and ($ ii $) defines a representation of the Lie algebra $ \bar{\mathscr{V}} $ on $ \mathscr{V} $ by derivations of the product and the $\lambda$-bracket of $ \mathscr{V} $, commuting with $ \partial$.
\end{lem}
\begin{proof}
It all follows directly when we put $ \lambda = 0 $ in the axioms for the $ \lambda$-bracket $ \{\cdot_{\lambda} \cdot\} $ of a PVA. First, both brackets are well defined since sesquilinearity holds for $ \{\cdot_{\lambda} \cdot\} $: for every $ a, b \in \mathscr{V} $ we have $ \{\partial a, b\} =-\lambda\{ a_{\lambda} b \}_{\lambda = 0}=0 $ and $ \{a, \partial b\} = \{a_{\lambda} \partial b\}_{\lambda = 0 } = \partial \{a_{\lambda} b\} \in \partial \mathscr{V} $.

Let us now verify the Lie algebra axioms for the first bracket: note that $ \int \{b_{-\lambda - \partial}a\}_{\lambda = 0} = \int \{b_{\lambda}a\}_{\lambda = 0 } $ since only the coefficients of the $0$-th power of $-\lambda - \partial$ and $ \lambda $ respectively survive in $ \bar{\mathscr{V}} $, and they obviously coincide. By skewsymmetry of $ \{\cdot_{\lambda} \cdot\} $ we have
\begin{equation}
\{\smallint a, \smallint b\} = \smallint \{a_{\lambda} b \}_{\lambda = 0 } = - \smallint \{b_{-\lambda - \partial} a\}_{\lambda = 0 } = - \smallint \{b_{\lambda} a\}_{\lambda = 0 } = - \{ \smallint b, \smallint a \}.
\end{equation}
Hence, skewsymmetry holds for ($i$). Similarly, the Jacobi identity for $ \{\cdot_{\lambda} \cdot\} $ provides that the Jacobi identity holds for this bracket as well, just putting $ \lambda = \mu = 0 $ in the corresponding definitions:
\begin{equation}
\{\smallint a, \{\smallint b, \smallint c\}\} = \{\smallint b, \{\smallint a, \smallint c\}\} + \{\{\smallint a, \smallint b\}, \smallint c \}.
\end{equation}
Therefore, $ \bar{\mathscr{V}} $ is endowed with a Lie algebra structure with the Lie bracket defined by ($i$).

Next, we have to check that (ii) is a representation of 
$\bar{ \mathscr{V}} $ on $ \mathscr{V} $, i.e., that
\begin{equation}
\{\{\smallint f, \smallint g\},a\} = \{\smallint f, \{\smallint g, a\}\} - \{ \smallint g, \{\smallint f, a\}\}
\end{equation}
holds for all $ \int f, \int g \in \bar{\mathscr{V}}$, $ a \in \mathscr{V} $. Again, this is due to the Jacobi identity. Then we have to check that  $\bar{\mathscr{V}}$ acts on $ \mathscr{V} $ as derivations of the product. For $ a, b \in \mathscr{V} $ and $ \int h \in \bar{\mathscr{V}} $ we have, by the left Leibniz rule:
\begin{align}
\{\smallint h, ab\} = \{h_{\lambda}ab\}_{\lambda = 0} & = (\{h_{\lambda}a\}b)_{\lambda = 0 } + (\{h_{\lambda}b\}a)_{\lambda = 0 } = \nonumber\\
& = \{h_{\lambda}a\}_{\lambda = 0 }\,b + \{h_{\lambda}b\}_{\lambda = 0 }\,a = \{\smallint h,a\}b +\{\smallint h,b\}a.
\end{align}
Similarly, by the Jacobi identity, we check that it acts by derivations of the 
$\lambda$-bracket.
Finally, we have to check that the derivations $ \{\int h, \cdot\,\} $ commute with $ \partial $. For every $ a \in \mathscr{V} $ we have 
\begin{equation}
(\{\smallint h, \cdot\, \} \circ \partial) a = \{ \smallint h, \partial a \} = \{h_{\lambda} \partial a \}_{\lambda = 0} = ((\lambda+ \partial)\{ h_{\lambda} a\})_{\lambda = 0} = \partial \{h_{\lambda}a\}_{\lambda = 0 } = (\partial \circ \{\smallint h, \cdot\,\})a
\end{equation}
due to the sesquilinearity of $ \{\cdot_{\lambda} \cdot\} $.
\end{proof}

\begin{defn}
Given a PVA 
$\mathscr{V}$
and a local functional $\smallint h \in\bar{\mathscr{V}}$, the associated 
 \emph{Hamiltonian PDE} is
\begin{equation}\label{HamiltonianPDE}
\frac{du}{dt} = \{ \smallint h, u \}.
\end{equation}
The local functional $ \int h$ 
is called the \textit{Hamiltonian} of this equation. 
\end{defn}

In the case when the PVA $\mathscr{V} $ is an algebra of differential functions in the variables $\{u_i\}_{i \in I} $ and the $\lambda$-bracket is given by the Master Formula (\ref{masterformula}), we reproduce the traditional definitions:
\begin{itemize}
\item[(i)] Hamiltonian PDE:  $ \frac{du}{dt} = \{ \int h, u \} = H \frac{\delta \int h}{\delta u} $;
\item[(ii)] Poisson bracket on $\bar{\mathscr{V}}$: $\{\int f, \int g\} =  \int \frac{\delta g}{\delta u} \cdot H \frac{\delta f}{\delta u}$.
\end{itemize}
The first claim is obvious, and the second is obtained by integration by parts. 

It follows that in this case $ \bar{\mathscr{V}} $ acts on $\mathscr{V} $ by evolutionary vector fields: $ \int f \mapsto  X_{H \frac{\delta f}{\delta u}} $,
and that the following holds. 
\begin{cor}
\label{cor:basiclemma}
We have a Lie algebra homomorphism $ \bar{\mathscr{V}} \rightarrow 
\mathscr{V}^\ell $, $ \int f \mapsto X_{H \frac{\delta f}{\delta u}} $.
\end{cor}

Thus, in the case when $\mathscr{V}$ is an algebra of differential functions
with the Poisson $\lambda$-bracket given by the Master formula,
the Hamiltonian equation is a special case of the evolution equation with RHS $ H \frac{\delta \int h}{\delta u} $ and the corresponding evolutionary vector field is $ X_{H\frac{\delta h}{\delta u}}$. 

\begin{defn}
A local functional $ \int f \in \bar{\mathscr{V}} $ is called an \emph{integral of motion} of the evolution equation (\ref{evolutionPDE}) and $f$ is called
a \emph{conserved density}, 
if $ \int \frac{df}{dt} = 0 $, or, equivalently, if $ \int X_F f = 0 $. 
Integrating by parts, this, in turn, is equivalent to
\begin{equation}
\int \frac{\delta f}{\delta u}\cdot F =0.
\end{equation}
Hence, $ \int f $ is an integral of motion of the Hamiltonian equation (\ref{HamiltonianPDE})
 if and only if $ f $ and $ h $ are in \textit{involution}, that is if $ \{ \int f, \int h \} = 0 $.
\end{defn}

So, we have completely translated the language of Hamiltonian PDEs into the language of PVAs.

\begin{defn}
The Hamiltonian PDE (\ref{HamiltonianPDE}) is called \emph{Liouville integrable} if $ \int h $ is contained in an infinite-dimensional abelian subalgebra of the Lie algebra $ \bar{\mathscr{V}} $. That is, if there exists an infinite sequence of linearly independent local functionals $\int h_n $, 
such that $ \int h_0 = \int h $ and $ \{\int h_n, \int h_m \} = 0 $ for all $ n,m \in \bZ_+ $.
\end{defn}

By Corollary \ref{cor:basiclemma}, integrals of motion in involution go to commuting evolutionary vector fields $ X_{H \frac{\delta \int h}{\delta u}} $. Hence Liouville integrability usually implies Lie integrability (provided we make some weak assumption on $ H(\partial) $, such as $ H(\partial) $ 
is non-degenerate). In fact, in order to check that the local functionals are linearly independent, it is usually easier to check that the corresponding evolutionary vector fields are linearly independent.

\begin{exr}
Show that the equation 
$\frac{du}{dt}=u''$ 
is Lie integrable, but has no non-trivial integrals of motion, hence is not Hamiltonian. On the other hand the equation $\frac{du}{dt}=u'''$ is Hamiltonian
with $H=\partial$, 
$h=-\frac{1}{2} (u')^2$, 
and it is both Lie and Liouville integrable. 
\end{exr}

\begin{rmk}
Let $ F, G,\ldots $ be a sequence of elements of $ \mathscr{V}^\ell $, such that the corresponding evolutionary vector fields commute, i.e. the 
corresponding evolution equations are compatible. Then we have a \textit{hierarchy} of evolution equations
\begin{equation}
\frac{du}{dt_0} = F,\quad \frac{du}{dt_1} = G, \,\ldots,
\end{equation}
so that the solution of this hierarchy depends now on $ x $ and on infinitely many times: $ u = u(x,t_0,t_1,t_2,\ldots ) $.
\end{rmk}

\subsection{The Lenard-Magri scheme of integrability}
There is a very simple scheme to prove integrability, called the \textit{Lenard-Magri scheme}. Although it is not a theorem, it always works in practice.

Let $\mathscr{V}$ be an algebra of differential functions in $\ell$ variables
$u_1,...,u_\ell$. First of all, introduce the following symmetric bilinear forms on $ \mathscr{V}^\ell $:
\begin{equation}\label{bilinearform1}
(\cdot\vert\cdot): \mathscr{V}^\ell \times \mathscr{V}^\ell \longrightarrow \bar{\mathscr{V}}, \quad (F \vert G ) = \smallint F \cdot G.
\end{equation}
Given a matrix differential operator $ H(\partial) \in 
\text{Mat}_{\ell \times \ell} \mathscr{V}[\partial] $
\begin{equation}\label{bilinearform2}
\langle \cdot, \cdot \rangle_{H}: \mathscr{V}^\ell \times \mathscr{V}^\ell \longrightarrow \bar{\mathscr{V}}, \quad \langle F, G \rangle_{H} = (H(\partial) F \vert G ).
\end{equation}
Note that $ (H(\partial) F \vert G ) = (F \vert H^{\ast}(\partial) G) $, where $ H^{\ast}(\partial) $ is the adjoint differential operator of 
$ H(\partial) $. Indeed, defining $ \ast$ on $ \mathscr{V}[\partial] $ as an anti-involution such that $ \ast(f) = f $ and $\ast(\partial) = - \partial $, 
we get $ (\partial f \vert g ) = - ( f \vert \partial g ) $ because $ (\partial f \vert g ) + ( f \vert \partial g ) = \int\,\partial (f g ) = 0 $ in $ \bar{\mathscr{V}} $.     
Hence, if $ H(\partial) $ is skewadjoint, then the bilinear form (\ref{bilinearform2}) is skewsymmetric.

Proof of Liouville integrability is based on the following result.
\begin{lem}[Lenard lemma]
\label{LMlemma}
Let $ H(\partial) $ and $ K(\partial) $ be skewadjoint differential operators on $ \mathscr{V}^\ell $. Suppose elements 
$ \xi_0,\ldots,\xi_N \in\mathscr{V}^\ell $ satisfy the following Lenard-Magri relation:%
\begin{equation}
\label{LMrelation}
K(\partial)\xi_{n+1} = H(\partial) \xi_n,\quad n=0,\ldots, N-1.
\end{equation}
Then, the  $ \langle \xi_m,\xi_n \rangle =0 $  for all $ m,\,n=0,\ldots, N $, whenever we consider it with respect to $ H $ or $ K $: $ \langle \xi_m,\xi_n \rangle_{H,K} = 0 $.
\end{lem}

\begin{proof}
Proceed by induction on $ i = \vert m - n \vert$. If $ i =0 $, then $ m = n $ and we get $ \langle \xi_n, \xi_n\rangle_{H,K} = - \langle \xi_n, \xi_n\rangle_{H,K} $ because the form is skewsymmetric, therefore it is equal to zero. Now let $i > 0 $; by skewsymmetry we may assume $ m > n $. We have
\begin{equation}
\langle \xi_m, \xi_n\rangle_{H} = (H(\partial)\xi_m \vert \xi_n) = - (\xi_m \vert H(\partial)\xi_n) = - (\xi_m \vert K(\partial) \xi_{n+1}) = (K(\partial)\xi_m \vert \xi_{n+1}) = \langle \xi_m, \xi_{n+1} \rangle_{K}, 
\end{equation}
and, by the induction hypothesis, the RHS is zero, since $ \vert m - (n+1) \vert < \vert m -n \vert $. Similarly we have, assuming $ n > m $:
\begin{equation}
\langle \xi_m, \xi_n \rangle_K = ( K(\partial) \xi_m \vert \xi_n ) = - ( \xi_m \vert K(\partial) \xi_n ) = - ( \xi_m \vert H(\partial) \xi_{n-1} ) = (H(\partial) \xi_m \vert \xi_{n-1}) = \langle \xi_m, \xi_{n-1} \rangle_H
\end{equation}
and again, by induction hypothesis the RHS is zero since $ \vert n - 1 - m \vert < \vert n - m \vert $.
\end{proof}

This lemma is important since, if we can prove that the elements $ \xi_m \in \mathscr{V}^\ell $ are variational derivatives, i.e. $ \xi_m = \frac{\delta \int h_m}{\delta u} $ for some local functionals $ \int h_m $, it guarantees 
that 
$\int h_m$ and $ \int h_n $ are in involution with respect to both brackets on
 $\bar{\mathscr{V}}$
. Indeed, we know that the bracket on  $\bar{\mathscr{V}}$ for the Poisson structure $H$ is given by
\begin{equation}
\{\smallint f, \smallint g\}_H 
= \int \frac{\delta g}{\delta u} \cdot H(\partial) \frac{\delta f}{\delta u} = \bigg( \frac{\delta g}{\delta u} \vert H(\partial)  \frac{\delta f}{\delta u}\bigg) = \bigg\langle \frac{\delta f}{\delta u}, \frac{\delta g}{\delta u} \bigg\rangle_H,
\end{equation}
therefore, if $ \xi_n, \xi_m $ are variational derivatives, then by Lemma \ref{LMlemma} we get
\begin{equation}
\{\smallint h_m, \smallint h_n \}_H = \bigg\langle\frac{\delta  \int h_m}{\delta u }, \frac{\delta \int h_n}{\delta u}\bigg\rangle_H = \langle \xi_m, \xi_n \rangle_H = 0 ,
\end{equation}
and the same holds for $K$. In other words, we have the following corollary of 
Lenard's lemma.
\begin{cor}
\label{LMcor} Let $ H(\partial) $ and $ K(\partial) $ be skewadjoint 
differential operators on $ \mathscr{V}^\ell $. Suppose that the local functionals $\smallint h_0,...,\smallint h_N$  
satisfy the following relation:
\begin{equation}
\label{LM}
K(\partial)\frac{\delta\smallint h_{n+1}}{\delta u} =H(\partial)\frac{\delta\smallint h_{n}}{\delta u} 
,\quad n=0,\ldots, N-1.
\end{equation}
Then all these local functionals  
are in involution with respect to both brackets 
$\{., .\}_H$ and $ \{., .\}_K$ on 
$\bar{\mathscr{V}}$.
\end{cor}
In the case when (\ref{LM}) holds, and $K, H$ are Poisson structures, one
says that the evolution equations
\[\frac{du}{dt_n}
=K(\partial)\frac{\delta\smallint h_{n+1}}{\delta u} =H(\partial)\frac{\delta\smallint h_{n}}{\delta u}\]
form a  hierarchy of {\it bi-Hamiltonian} equations. Note that if the right-hand sides of these equations span an infinite-dimensional subspace in the space of evolutionary vector fields, then all of these equations are both Lie and Liouville integrable. 

We now must address two issues:
\begin{enumerate}
\item How can we construct vectors $ \xi_n $'s satisfying equation (\ref{LMrelation})?
\item How can we prove that such $\xi_n $'s are variational derivatives?
\end{enumerate}
Although the second issue has been completely solved considering some reduced de Rham complex, called the \textit{variational complex}, discussed in the next lecture, the first and basic issue is far from being  resolved, though there are some partial results.

We will now see how to construct a sequence of vectors $ \xi_n$'s satisfying the Lenard-Magri relation.

\begin{lem}[Extension lemma]\cite{BDSK}
\label{lem:extension}
Suppose that, in addition to the hypothesis of Lemma \ref{LMlemma}, we also have the following \textit{orthogonality condition}: assume to have vectors $ \xi_0,\ldots,\xi_N \in \mathscr{V}^\ell$, satisfying the Lenard-Magri relation (\ref{LMrelation}),
such that 
\[ {\text{Span}\{\xi_0,\ldots,\xi_N\}^{\perp} \subseteq \mathrm{Im}\,K(\partial)},\]
where
$ \text{Span}\{\xi_0,\ldots,\xi_N\}^{\perp}$  is the orthogonal complement with respect to the symmetric bilinear form (\ref{bilinearform1}).
Then we can extend the given sequence to an infinite sequence of vectors satisfying the Lenard-Magri relation 
(\ref{LMrelation}) for any $ n \in \bZ_{+}$.
\end{lem}

\begin{proof}
It suffices to construct $ \xi_{N+1} $ such that equation (\ref{LMrelation}) holds for $ n = N $. In fact, the orthogonal complement to $\text{Span}\{\xi_0,\ldots, \xi_{N+1}\} $ is contained in the orthogonal complement to $ \text{Span}\{\xi_0,\ldots,\xi_N\} $, hence the orthogonality condition would hold for the extended sequence. By Lemma \ref{LMlemma}, $ H(\partial)\xi_N \perp \xi_n $ for every $ n=0,\ldots, N $. Hence, by the orthogonality condition, $ H(\partial)\xi_N \subset \mathrm{Im}\,K(\partial) $. Therefore, $ H(\partial)\xi_N = K(\partial) \xi_{N+1} $ for some element $ \xi_{N+1} \in \mathscr{V}^\ell $. We can now iterate this procedure to construct an infinite sequence of vectors.
\end{proof}


Now, let us address the question why the $ \xi_n$'s, satisfying 
equation  
(\ref{LMrelation}),   are variational derivatives. Note that so far we only have used the fact that $ H $ and $ K $ are skewadjoint, but none of their other properties as Poisson structures. However, we will need these properties in order  to prove that the $\xi_n$'s are variational derivatives. Moreover, we will need the notion of \textit{compatibility of Poisson structures}:

\begin{defn}[Magri compatibility]
Given two Poisson structures $ H $ and $ K $, they (and the corresponding $\lambda$-brackets) are \textit{compatible} if any their linear combination $ \alpha H + \beta K $ is again a Poisson structure.
\end{defn}
Examples \ref{ex5.3} and \ref{ex5.4} provide multiparameter families of compatible Poisson structures.

The importance of compatibility of Poisson structures is revealed by the following theorem.
\begin{thm}[see \cite{O}, Lemma 7.25]\label{thm:olver}
Suppose that $ H,\,  K \in \text{Mat}_{\ell \times \ell}\mathscr{V}[\partial]$
   are compatible Poisson structures, with $ K $ non-degenerate (i.e. $ KM = 0 $ implies $ M = 0 $ for any differential operator 
$ M \in \text{Mat}_{\ell \times \ell}\mathscr{V}[\partial]  $).  
Suppose, moreover, that the Lenard-Magri relation $ K(\partial)\xi_{n+1} = H(\partial) \xi_n $ holds for $ n =0,\,1 $, and that $ \xi_0,\,\xi_1 $ are variational derivatives: $ \xi_0 = \frac{\delta\int  h_0}{\delta u} $, $\xi_1 = \frac{\delta \int h_1}{\delta u}$ for some $ \int h_0, \int h_1 \in \bar{\mathscr{V}} $. Then $ \xi_2 $ is closed in the variational complex (discussed in the next lecture).

\end{thm}

Theorem \ref{thm:olver} 
allows us to construct an infinite series of integrals of motion in involution. In fact, if we are given a pair of compatible Poisson structures $H,K $ with $ K $ non-degenerate and we know that the first two vectors $\xi_0$ and $\xi_1$, satisfying the Lenard-Magri relation, are exact in the variational complex
(i.e. they are variational derivatives), it would follow that, whenever we can construct an extending sequence of $ \xi_n$'s, then all of them would be closed, and hence exact in some extension 
$\widetilde{\mathscr{V}}$
of the algebra of differential functions 
$\mathscr{V}$
(i.e. $\xi_n = \frac{\delta h_n}{\delta u} $ for some 
$ h_n \in \widetilde{\mathscr{V}}$).
This is a consequence of the theory of the variational complex, discussed in the next lecture. Note, however, that $\xi_n\in\mathscr{V}^\ell$ for all $n$.

\begin{rmk}
Let $\xi_{-1} = 0 = \frac{\delta}{\delta u} 1 $. If $ K(\partial)\xi_0 = 0 $, then for Theorem \ref{thm:olver} to hold it suffices to have only $ \xi_1 $ such that $ K(\partial)\xi_1 = H(\partial)\xi_0 $, since the first step is trivial.
\end{rmk}

\begin{prp}
Suppose we have two compatible Poisson structures $ H $ and $ K $ on $\mathscr{V} $, with $K$ non-degenerate, and consider a basis $ \xi_0^1,\ldots,\xi_0^s $ of $ \Ker\,K $ (it is finite dimensional since $ K $ is non-degenerate). Suppose that each $ \xi_0^i $ can be extended to infinity so that equation (\ref{LMrelation}) holds for all $ n \in \bZ_+ $, and hence we have $ \xi_n^i$ for all $ n \in \bZ_+ $. Assume moreover that all vectors ${\xi_0^i}$ are exact: $ \xi_0^i = \frac{\delta \int h_0^i}{\delta u}  $ for some local functional $ \int h_0^i $. Then all the $ {h_n^i}'s$ are in involution. Hence, we have constructed canonically an abelian subalgebra of the Lie algebra $ \bar{\mathscr{V}} $, corresponding to the pair of compatible Poisson structures.
\end{prp}

This proposition holds by the following result.

\begin{lem}[Compatibility lemma]\label{lem:compatibility}
Let $ H(\partial) $ and $ K(\partial) $ be skewadjoint differential operators. Suppose we have vectors $ \xi_0,\ldots, \xi_N \in \mathscr{V}^\ell $ such that $ K(\partial)\xi_0 = 0 $ and equation (\ref{LMrelation}) holds for $ n =0,\ldots, N$. Suppose moreover to have an infinite sequence of vectors $ \xi_0',\ldots,\xi_M',\ldots$ satisfying equation (\ref{LMrelation}).
Then all $ \xi_i$'s are in involution with all $\xi_j'$'s 
with respect to both Poisson structures $H$ and $K$.
\end{lem}

\begin{proof}
Proceed by induction on $ i $. The induction basis follows by the fact that for $ i = 0 $ we have $ K(\partial)\xi_0 = 0 $:
\begin{equation}
\langle\xi_0, \xi_j' \rangle_K = ( K(\partial) \xi_0 \vert \xi_j' ) = ( 0 \vert \xi_j') = 0
\end{equation}
and
\begin{align}
\langle\xi_0, \xi_j' \rangle_H = (H(\partial) \xi_0 \vert \xi_j' ) = ( \xi_0 \vert H^{\ast}(\partial) \xi_j' ) & = - (\xi_0 \vert H(\partial) \xi_j') =\nonumber\\
& = - (\xi_0 \vert K(\partial) \xi_{j+1}' ) = ( K(\partial) \xi_0 \vert \xi_{j+1}') = ( 0 \vert \xi_{j+1}' ) = 0.
\end{align}
Now let $N> i > 0 $ and suppose $ \langle \xi_h, \xi_j' \rangle_{H,K} = 0 $ for all $ h \leq i $. We want to prove that $ \langle \xi_{i+1}, \xi_j' \rangle_{H,K} = 0 $. We have
\begin{equation}
\langle \xi_{i+1}, \xi_j' \rangle_K = ( K(\partial) \xi_{i+1} \vert \xi_j' ) = ( H(\partial)\xi_i \vert \xi_j' ) = \langle \xi_i, \xi_j' \rangle_H = 0
\end{equation}
and
\begin{align}
\langle \xi_{i+1}, \xi_j' \rangle_H = ( H(\partial) \xi_{i+1} \vert \xi_j' ) = - (\xi_{i+1} \vert H(\partial) \xi_j') & = - (\xi_{i+1} \vert K(\partial) \xi_{j+1}' ) = \nonumber\\
& = (K(\partial)\xi_{i+1} \vert \xi_{j+1}') = (H(\partial)\xi_i \vert \xi_{j+1}' ) = \langle \xi_i, \xi_{j+1}' \rangle_H = 0, 
\end{align}
where in both cases the last equality is given by the induction hypothesis. Hence $ \langle \xi_i,\xi_j'\rangle_{H,K} = 0 $ for all $ i,j $ in question.
\end{proof}


In the next lecture we will demonstrate the Lenard-Magri method on the example of the KdV hierarchy, hence establishing its integrability.
\newpage

\section{Lecture 6 (January 15, 2015)}

\subsection{An example: the KdV hierarchy}
We begin this lecture with an example.

Consider the PVA 
$\mathscr{P}_1 = \F[u, u', u'', \ldots ]$ with two compatible compatible $\lambda$-brackets: one is the Gardner-Faddeev-Zakharov (GFZ) $\lambda$-bracket $\{u_\lambda u\}_K= \lambda$, and the other one is the Magri-Virasoro (MV) $\lambda$-bracket $\{u_\lambda u\}_H=(\partial+2\lambda)u+c\lambda^3$ for some $c\in \F$. The corresponding compatible Poisson structures are $K(\partial)=\partial$ and $H(\partial)= u'+2u\partial +c\partial^3$ respectively (see Example \ref{ex5.3}). Note that $\Ker\,\partial =\F$.

We shall use the Lenard-Magri scheme discussed in the previous lecture to construct an infinite hierarchy of integrable Hamiltonian equations: we want to construct an infinite sequence of vectors $ \xi_n \in \mathscr{P}_1 $, such that $ K(\partial)\xi_{n+1} = H(\partial)\xi_n $, $ n \in \bZ_+ $, and $ \xi_0 \in \Ker\,K(\partial) $. We also want to compute the conserved densities $h_n$, such that 
$ \xi_n = \frac{\delta h_n}{\delta u} $.

We can take $ \xi_0 = 1$ and, consequently, $h_0= u$. Taking 
$\xi_{-1}=0$, $h_{-1}=0$, we can apply Theorem \ref{thm:olver} to establish by induction on $n$ that all the $\xi_n$'s, satisfying the Lenard-Magri relation,
are closed in the variational complex. Since, by Corollary \ref{Cor6.1} from the next section,
every closed $1$-form is exact over the algebra of differential polynomials,
we conclude that there exist $h_n\in\mathscr{P}_1$, such that 
$ \xi_n = \frac{\delta h_n}{\delta u} $.  

The first step of the Lenard-Magri scheme:
\begin{equation}
H(\partial)\xi_0= K(\partial)\xi_1 \,\Longrightarrow\, u'= \xi_1' \,\Longrightarrow\, \xi_1=u \,\Longrightarrow\, h_1=\dfrac{1}{2} u^2.
\end{equation}
The second step of the Lenard-Magri scheme:
\begin{equation}
H(\partial)\xi_1= K(\partial) \xi_2 \,\Longrightarrow\, 3uu'+cu'''=\xi_2' \,\Longrightarrow\, \xi_2 =\dfrac{3}{2} u^2 +cu'' \,\Longrightarrow\, \xi_2 = \frac{\delta}{\delta u} h_2,\,  h_2=\dfrac{1}{2}(u^3 +cuu'').
\end{equation}
And so on.
\begin{rmk}
All $ \xi_n $'s are defined up to adding an element of $ \Ker\,K(\partial) $, 
hence, in this case, up to adding a constant.
\end{rmk}
The corresponding KdV hierarchy of Hamiltonian equations is given by $\dfrac{du}{dt_n}= K(\partial) \xi_{n+1} = \partial \xi_{n+1}$, namely:
\begin{equation}
\frac{du}{dt_0}=u', \quad 
\frac{du}{dt_1}=3uu' + cu''', \quad\ldots.
\end{equation}
Note that for $ n = 1 $ we get the classical KdV equation, which is the simplest dispersive equation (cf. Example 5.6) .

The hierarchy can be extended to infinity because the orthogonality condition 
$(\xi_0)^\perp \subset \mathrm{Im}\,K(\partial)$ holds (see the Extension 
Lemma \ref{lem:extension}):
since $\xi_0=1$ and $1^\bot= \partial \mathscr{P}_1=\mathrm{Im}\,K(\partial)$ (equivalently, if $ P \in (\xi_0)^\perp $ then $\int 1\cdot P=0 \Leftrightarrow P\in\partial \mathscr{P}_1$, therefore $P\in\mathrm{Im}\,K(\partial)$). It is easy to show by induction that the differential order of $ K(\partial)\xi_n $ is $2n+1$ (if $ c \neq 0 $), hence they are linearly independent, and we consequently have Lie integrability. Then automatically all the $\smallint h_n$'s are linearly independent, and we have Liouville integrability as well. Therefore the KdV equation is integrable, as are all the other equations $\dfrac{d u}{d t_n} = K(\partial)\xi_{n+1}= H(\partial)\xi_n$.

\begin{exr}
Show that the the next equation of the KdV hierarchy is 
\[\dfrac{du}{dt_2}=\partial\xi_3=\frac{15}{2}u^2u'+10cu'u''+5cuu''' ,\]
and the next conserved density is
\[ h_3=\frac{5}{8}u^4+  \frac{5}{3}cu^2u''+  
\frac{5}{6}cuu'^2+  \frac{1}{2}c^2uu^{(4)}.                                
\]
\end{exr}

\subsection{The variational complex}

As professor S.\@ S.\@ Chern used to say, ``In life both men and women are important; likewise, in geometry both vector fields and differential forms are important''. In our theory vector fields are evolutionary vector fields over an
algebra of differential functions $\mathscr{V}$:
\begin{equation}
X_P = \sum_{\substack{i=1,\ldots,\ell \\ n\in\bZ_+}} \partial^n P_i \frac{\partial}{\partial u_i^{(n)}},  \quad P \in \mathscr{V}^\ell,
\end{equation}
and, as we have already seen in Lecture \ref{sec:5}, they commute with $\partial = X_{u'} $. Differential forms in our theory are ``variational differential forms'' which are obtained by the reduction of the de Rham complex over 
$\mathscr{V}$ by the image of $\partial$.

Let $J=\{1,...,N\}$, where $N$ can be infinite. Given a unital commutative associative algebra $A$, containing the algebra of polynomials 
$\F[x_j|\,j\in J]$ and endowed with $N$ commuting derivations $\dfrac{\partial}{\partial x_j}$, extending those on the subalgebra of polynomials,
the \emph{de Rham complex} $\widetilde\Omega(A)$ over $A$ consists of finite
linear combinations of the form 
\begin{equation}
\sum_{i_1< \cdots < i_k} f_{i_1,\ldots,i_k} \, d x_{i_1} \wedge \ldots \wedge d x_{i_k} \in \widetilde\Omega^k(A), \quad f_{i_1,\ldots,i_k}  \in A,
\end{equation}
so that we have the decomposition
\[ \widetilde\Omega(A) = \bigoplus_{k\in\Z_+} \widetilde\Omega^k(A). 
\] 
Moreover, $\widetilde\Omega(A)$ is a $\Z_+$-graded
 associative commutative superalgebra 
with parity given by $p(A)=\bar 0$ and $p(dx_i)=\bar 1$. This is a complex 
with the usual de Rham differential, namely an odd derivation $ d : \widetilde\Omega^k(A) \longrightarrow \widetilde\Omega^{k+1}(A)$ of $ \widetilde\Omega(A) $ such that
\begin{equation}
 d(dx_i) = 0; \,\, df = \sum_{j \in J} \frac{\partial f}{\partial x_j} d x_j\,\hbox{for}\,\, f\in A.
\end{equation}
It is easily checked that $ d $ is a differential, namely that $ d^2 =0 $. We will denote this complex by $ (\widetilde\Omega(A),d)$.

Let us define now an increasing filtration on $ A $ by subalgebras:
\begin{equation}
A_j=\{a\in A \;\vert\; \dfrac{\partial a}{\partial x_i}=0, \;\forall\;i \geq j\}.
\end{equation}
We call $A_0 $  the subalgebra of {\it quasiconstants}. 
If $\dfrac{\partial}{\partial x_j}A_j=A_j$ for all $j \in J $,
we call $A$ \emph{normal}. Obviously, the algebra of polynomials in any (including infinite) number of variables is normal. 


\begin{lem}[Algebraic Poincar\'{e} Lemma]\label{lem:algebraicpoincare}
Let $A$ be a normal commutative associative algebra as above, and let $(\widetilde{\Omega}(A), d)$ be its de Rham complex. Then 
\begin{equation}
H^k(\widetilde{\Omega}(A), d)=0, \,\, k> 0 ;\,\,\,\,
H^0(\widetilde{\Omega}(A), d)= A_0.
\end{equation}
\end{lem}
 
\begin{proof}
Extend the filtration of $A$ to $\widetilde\Omega(A)$ by letting $\widetilde\Omega_j(A)$ be the subalgebra, generated by $A_j$ and $dx_1,\ldots ,dx_j$. Introduce ``local'' {\it homotopy operators} $K_m : \widetilde\Omega_m^k(A) \to \widetilde\Omega_m^{k-1}(A)$ by
\begin{equation}
K_m(f\,dx_{i_1}\wedge \cdots \wedge dx_{i_s} \wedge dx_m)=
\begin{cases}
(-1)^s(\int f\,dx_m)\, dx_{i_1}\wedge\cdots\wedge dx_{i_s}  \\
0, \quad\text{if}\; dx_m \,\text{does not occur}
\end{cases}
\end{equation}
where $ i_1 < \ldots < i_s < m $. Here
the integral $ \int f\,dx_m $ is a preimage in $A_m$ of $ f\in A_m $ under the
map $ \frac{\partial}{\partial x_m} $, which exists by normality of $A$.

Let $\omega\in \widetilde{\Omega}^k_m(A)$. Then it is straightforward to check that 
\begin{equation}
K_md\omega + dK_m\omega-\omega \in \widetilde{\Omega}^k_{m-1}(A), \quad \text{for}\quad m\geq 1.
\end{equation}
Hence, if $\omega\in \widetilde{\Omega}^k_m(A)$ is closed, then 
\begin{equation}
d(K_m\omega)-\omega \in \widetilde{\Omega}^k_{m-1}(A).
\end{equation}
Equivalently $ \omega \in \widetilde{\Omega}^k_{m-1}(A) + d\widetilde{\Omega}(A)$, i.e.\@ we may assume that $\omega\in \widetilde{\Omega}^k_{m-1}(A)$ modulo adding an exact tail. Repeating the same argument we proceed downward in the filtration, and after a finite number of steps we get $0$, hence $ \omega \in d\widetilde{\Omega}(A)$.
\end{proof}


Let $\mathscr{V}$ be an algebra of differential functions. Consider the lexicographic order on pairs $ (m,i) \in \bZ_+ \times I $, and consider the corresponding filtration by subalgebras as above:
\begin{equation}
\mathscr{V}_{m,i} = \{ f \in \mathscr{V} \,\vert\, \frac{\partial f}{\partial u^{(n)}_j} = 0, \;\forall\, (n,j)\geq (m,i) \}.
\end{equation}
Hence we can define normality of $ \mathscr{V} $ as above.
\begin{exm}
The algebra of differential polynomials in $ \ell $ variables 
$\mathscr{P}_\ell$ is normal.
\end{exm}

The derivation $\partial$ of $\mathscr{V}$ extends to an even derivation
of the superalgebra $ \widetilde{\Omega}(\mathscr{V}) $
by letting $ \partial(du^{(n)}_i) =d u^{(n+1)}_i $.
\begin{exr}
Show that $ d\partial = \partial d $. (Hint: use Axiom $(ii)$ in Definition \ref{algebradifffunctions}.)
\end{exr}
Due to this exercise, we may consider the reduced complex 
\[( \Omega(\mathscr{V}),d) = (\widetilde{\Omega}(\mathscr{V})/\partial\widetilde{\Omega}(\mathscr{V}) ,d),
\] 
called the \emph{variational complex} over the algebra of differential functions $\mathscr{V} $.
\begin{exr}
\label{exr:partialinjective}
Show that $ \partial$ is injective on $ \widetilde{\Omega}^k(\mathscr{V}) $ for $ k \geq 1 $.
\end{exr}

\begin{thm}[\cite{BDSK}]
\label{Th6.1}
Let $ \mathscr{V} $ be a normal algebra of differential functions. Then
\begin{equation}
H^k(\Omega(\mathscr{V}), d)=0, \quad k> 0;\,\,\,\,H^0(\Omega(\mathscr{V}), d)=\mc{F}/\partial \mc{F},
\end{equation}
where $\mc{F}\subset\mathscr{V}$ is the subalgebra of quasiconstants.
\end{thm}

\begin{proof}
We have a short exact sequence of complexes
\begin{equation}
0 \longrightarrow \partial \widetilde{\Omega}(\mathscr{V}) \longrightarrow \widetilde{\Omega}(\mathscr{V}) \longrightarrow \Omega(\mathscr{V}) \longrightarrow 0,
\end{equation}
which induces a long exact sequence in cohomology:
\begin{equation}
 \label{6.15} 
\text{H}^0(\partial\widetilde{\Omega}(\mathscr{V})) \longrightarrow \text{H}^0(\widetilde{\Omega}(\mathscr{V})) \longrightarrow \text{H}^0(\Omega(\mathscr{V})) \longrightarrow \text{H}^1(\partial\widetilde{\Omega}(\mathscr{V})) \longrightarrow \text{H}^1(\widetilde{\Omega}(\mathscr{V})) \longrightarrow \text{H}^1\Omega(\mathscr{V}))\longrightarrow\ldots\,.
\end{equation}
Since $ \mathscr{V} $ is normal, by Lemma \ref{lem:algebraicpoincare} we get 
\[ \text{H}^k(\widetilde{\Omega}(\mathscr{V}))= 0 \,\, \hbox{for}\,\, k>0.\] 

Now note that $ \text{H}^k(\partial\widetilde{\Omega}(\mathscr{V}),d) = 0 $ for $ k > 0 $. Indeed, take $ \tilde{\omega} \in \widetilde{\Omega}^k(\mathscr{V}) $ with $ k \geq 0 $. If $ d(\partial \tilde{\omega}) = 0 $, then $ \partial(d\tilde{\omega}) = 0 $ since $ d $ and $ \partial $ commute. So, thanks to Exercise \ref{exr:partialinjective}, we have $ d\tilde{\omega} = 0 $. By Lemma \ref{lem:algebraicpoincare}, since $ \tilde{\omega} $ is closed, it is exact: $ \tilde{\omega} = d \tilde{\eta} $ for some $ \tilde{\eta} \in \widetilde{\Omega}(\mathscr{V}) $, hence $ \partial\tilde{\omega} = \partial(d\tilde{\eta}) = d(\partial\tilde{\eta}) $, and
\[ \text{H}^k(\partial\widetilde{\Omega}(\mathscr{V}))= 0 \,\, \hbox{for}\,\, k>0.\] 
Therefore the long cohomology exact sequence (\ref{6.15}) becomes
\begin{equation}
\ldots \longrightarrow  \text{H}^0(\Omega(\mathscr{V})) \longrightarrow 0 \longrightarrow 0 
\longrightarrow 
\text{H}^1(\Omega(\mathscr{V})) 
\longrightarrow 0 \longrightarrow 0 \longrightarrow \text{H}^2(\Omega(\mathscr{V}))\longrightarrow 0 \longrightarrow 0 \ldots,
\end{equation}
so $H^k(\Omega(\mathscr{V}), d)=0 $ for $ k> 0 $. When $ k = 0 $ we obviously get $ \text{H}^0(\Omega(\mathscr{V})) \cong \text{H}^0(\widetilde{\Omega}(\mathscr{V}))/\text{H}^0(\partial\widetilde{\Omega}(\mathscr{V}))= \mc{F}/\partial \mc{F} $. 
\end{proof}

Let us study the variational complex more closely. We can write down explicitly the first terms of the complex $ \Omega(\mathscr{V}) $:
\begin{itemize}
\item $\Omega^0(\mathscr{V})={\mathscr{V}}/{\partial \mathscr{V}}$;
\item $ \Omega^1(\mathscr{V}) = \mathscr{V}^\ell $;
\item $ \Omega^2(\mathscr{V}) = \{ \text{skewadjoint } \ell \times \ell \text{ matrix differential operators over } \mathscr{V}\} $.
\end{itemize}
The corresponding maps are
\begin{equation}
\Omega^0(\mathscr{V}) \overset{d}{\longrightarrow} \Omega^1(\mathscr{V}) \overset{d}{\longrightarrow} \Omega^2(\mathscr{V}) \longrightarrow \ldots,\,\,\,\,
\int f \overset{d}{\mapsto} \frac{\delta}{\delta u}\int f \overset{d}{\mapsto} \frac{1}{2}(D_F - D_F^\ast) \mapsto \dots,
\end{equation}
where $ F = \frac{\delta}{\delta u} \int f $ and $ D_F = \sum\limits_{i \in I,n \in \bZ_+} \frac{\partial F}{\partial u^{(n)}_i} \partial^n $ is the {\it Fr\'{e}chet derivative}.

The first identification is clear since $ \widetilde{\Omega}(\mathscr{V}) = \mathscr{V} $. Let us explain how to obtain the identification 
$\Omega^1(\mathscr{V}) = \mathscr{V}^\ell $. We have
\begin{equation}
\Omega^1(\mathscr{V}) = \bigg\{ \sum_{i \in I, n \in \bZ_+} f_{i,n} du^{(n)}_i \bigg\} / \partial\widetilde{\Omega}^1(\mathscr{V}) = \bigg\{\int \sum_{i \in I, n \in \bZ_+} f_{i,n} du^{(n)}_i \bigg\} = \bigg\{ \int \sum_{i \in I, n \in \bZ_+} f_{i,n}\partial^n du_i \bigg\},
\end{equation}
where last equality is due to the fact that $ d $ and $ \partial $ commute. Integrating by parts, we get
\begin{equation}
\int \sum\limits_{i \in I, n \in \bZ_+} f_{i,n}\partial^n du_i = \sum\limits_{i=1}^\ell \left(\int \sum\limits_{n \in \bZ_+} (-\partial)^n f_{i,n}\right)du_i.
\end{equation}
Thus the identification $\Omega^1(\mathscr{V}) \overset{\sim}{\longrightarrow} \mathscr{V}^\ell$ is given by 
\begin{equation}\label{eq:omega1}
\int \sum_{i \in I, n \in \bZ_+}f_{i,n}du_i^{(n)} \mapsto \left(\sum_{n\in \bZ_+} (-\partial)^n f_{i,n}\right)_{i \in I}
\end{equation}
and this is an isomorphism of vector spaces. In particular, we conclude that if we take $ df = \sum\limits_{i \in I,n\in \bZ_+} \frac{\partial f}{\partial u^{(n)}_i} du^{(n)}_i $ (in this case $ f_{i,n} = \frac{\partial f}{\partial u^{(n)}_i}$), then the RHS of (\ref{eq:omega1}) becomes exactly the vector of variational derivative of $ f $. It also explains the action of the first differential $ d : \Omega^0(\mathscr{V}) \longrightarrow \Omega^1(\mathscr{V}) $. Moreover, it is clear that a $ 1$-form $ \xi \in \mathscr{V}^\ell $ is exact iff $ \xi = \frac{\delta f}{\delta u} $, and it is closed iff $ D_\xi $ is self-adjoint.

\begin{exr}
Show that the algebra of differential functions 
$\mathscr{P}_1[u^{-1}, \log u]$ is normal.
Show that any algebra of differential functions $\mathscr{V}$ can be included
in a normal one.
\end{exr}

Since the algebra of differential polynomials $\mathscr{ P}_\ell  $ is 
normal, we obtain the following corollary of Theorem \ref{Th6.1}.
\begin{cor}
\label{Cor6.1} 
Let $\mathscr{V} = \mathscr{P}_\ell $ be an algebra of differential polynomials. Then
\begin{enumerate}
\item[(a)] $ \Ker \frac{\delta}{\delta u} = \mathbb{F} + \mathrm{Im} \, \partial \, . $
\item[(b)] $ \mathrm{Im} \frac{\delta}{\delta u} = \{ F \in \mathcal{V}^\ell \ | \ D_F \mbox{ is selfadjoint} \}. $
\item[(c)]$ \omega \in \Omega_k (\mathcal{V}), k \geq 1, $ is closed if and only if it is exact.
\end{enumerate}
\end{cor}
Claim (a) is usually attributed to a paper by Gelfand-Manin-Shubin from the 70's, though it is certainly much older. Claim (b) is called the Helmholz criterion, and apparently, it was first proved by Volterra in the first half of the $ 20^{\mathrm{th}} $ century. 

If we know that $ \xi \in  \sV^\ell $ is a variational derivative: $ \xi = \frac{\delta h}{\delta u} $ for some $ h \in \sV $ (which is not unique since we can add to $ h $ elements from $ \partial \, \sV $), there is a simple formula to find one of such $ h: $
\begin{exr}
  \label{exercise 6.5}
  Let
  \[ \Delta = \sum_{i \in I, n \in \mathbb{Z}_+} u_i^{(n)} \ \frac{\partial}{\partial u_i^{(n)}}  \]
  be the degree evolutionary vector field, and suppose that $ \xi \in \sV^\ell  $ is such that $ \Delta (u \cdot \xi) \neq 0 $. Let $ h \in \Delta^{-1} (u \cdot \xi). $
  Show that
  \[ \frac{\delta h}{\delta u_i} - \xi_i \in \Ker (\Delta +1) \mbox{ for all } i \in I.\] 
Consequently, if $ \Ker (\Delta + 1) = 0,  $ then $ \frac{\delta h}{\delta u} = \xi $.\end{exr}


\subsection{Homogeneous Drinfeld-Sokolov hierarchy and the classical affine Hamiltonian reduction.}
The method of constructing solutions of the Lenard-Magri relation, described in Section 5.4, uses Theorem \ref{thm:olver}, which assumes that $ K $ is non-degenerate. In this section I will describe the direct method of Drinfeld and Sokolov on the example of the so called homogeneous hierarchy, which avoids the use of Theorem \ref{thm:olver}.

Consider the affine PVA $ \sV = \sV^1 (\mf{g}, s), $ where $ \mf{g} $ is a reductive Lie algebra with a non-degenerate invariant symmetric bilinear form 
$ (. \, | \, .) $ and $ s $ is a semisimple element of $ \mf{g}, $ with compatible Poisson $ \lambda $-brackets $ (a, b \in \mf{g}): $
\begin{equation}
\label{e6.21}
\{ a_\lambda b \}_H = [a,b] + (a|b) \lambda, \quad \{ a_\lambda b \}_K = (s| [a,b]),
\end{equation}
see Example \ref{ex5.4}. Note that the Poisson structure $ K $ is degenerate, as $ s $ is a central element of the corresponding $\lambda$-bracket. 

The Lenard-Magri relation \eqref{LM} for infinite $ N $ 
can be rewritten as follows:
\begin{equation}
\label{e6.22}
\{ \smallint h_n, u \}_H = \{  \smallint h_{n+1}, u   \}_K, \ n \in \mathbb{Z}_+, \ u \in \mf{g}.
\end{equation}
The Drinfeld-Sokolov method of constructing solutions to this equation is as follows, see \cite{DS85} and \cite{DSKV13}.
Choosing dual bases $ \{ u_i\}_{i \in I} $ and $ \{ u^i\}_{i \in I} $ 
of $\mf{g}$, let 
\[ L(z) = \partial + \sum_{i \in I} u^i \otimes u_i - z (s \otimes 1) \in \mathbb{F} \partial \ltimes (\mf{g} \otimes \sV )[z] \, . \]
The first step consists of finding a solution $ F(z) = \sum_{n\geq 0} F_n z^{-n} \in (\mf{g} \otimes \sV) [[z^{-1}]] $ of the following equations in $ \mathbb{F} \partial \ltimes (\mf{g} \otimes \sV) ((z^{-1})): $
\begin{equation}
\label{e6.23}
[L(z), F(z)] = 0, \ [s \otimes 1, F_0] = 0 \, .
\end{equation}
\begin{thm}
\label{Th6.3.1}
Assume that the element $ s $ is semisimple, and let 
$ \mf{h}$  
be the centralizer of $s$ in $ \mf{g} $, so that 
$ \mf{g}= \mf{h}\oplus \mf{h}^\perp$. Then
\begin{enumerate}
\item[(a)] There exist unique $ U (z) \in z^{-1} (\mf{h}^\perp \otimes \sV) [[z^{-1}]] $ and $f(z) \in (\mf{h} \otimes \sV) [[z^{-1}]]$,  such that 
\[ e^{ad \,  U (z)} L(z) = \partial + f(z)- z(s \otimes 1).\]  
The coefficients of $ U(z) $ and $ f(z) $ can be recursively computed. 
\item[(b)] Let $ a $ be a central element of $ \mf{h} $. 
Then $ F^a(z) = e^{-ad \, U (z)} (a \otimes 1) $ satisfies equations \eqref{e6.23}.
\end{enumerate}
\end{thm}
Define the variational derivative of $ \smallint f \in \sV / \partial \sV $ in invariant form:
\[ \frac{\delta \smallint f}{\delta u} = \sum_{i \in I} u^i \otimes \frac{\delta \smallint f}{\delta u_i} \, . \]
The second step is given by the following.

\begin{thm}
\label{Th6.3.2}
Let $ f(z), a $ and $ F^a(z) $ be as in Theorem \ref{Th6.3.1}. Let $ h^a (z) = (a \otimes 1 | f(z)). $
Then 
\begin{enumerate}
\item[(a)] $ F^a (z) = \frac{\delta \smallint h^a (z)}{\delta u} \, . $
\item[(b)] The coefficients of $ h^a (z) = \sum_{n \geq 0}\smallint h^a_n z^{-n} $ satisfy the Lenard-Magri relation (\ref{e6.22}). 
\item[(c)] All the elements $ \int h^a_n \in \sV / \partial \sV  $,
where $n\in\Z_+$ and $a$ is a central 
element of $\mf{h}$, are in involution with respect to both Poisson structures $ H $ and $ K. $
\end{enumerate}
\end{thm}
For proofs of these theorems we refer to \cite{DSKV13}. Note that the claim (c) of Theorem \ref{Th6.3.2} follows from claim (b) and Lemma \ref{lem:compatibility}.

\begin{exm}
\label{example6.3.1} 
Let $ s $ be a regular semisimple element of $ \mf{g}, $ so that $ \mf{h} $ is a Cartan subalgebra, and let $ a \in \mf{h} $.
Then the above procedure gives the following sequence of densities of local functionals in involution, satisfying the Lenard-Magri relation:
\[ 
\begin{split}
& h_{-1} = 0, \ h_0 = a, \ h_1 = \frac{1}{2} \sum_{\alpha \in \Delta} \frac{\alpha (a)}{\alpha (s)} e_{-\alpha} e_\alpha, \\
& h_2 = \frac{1}{2} \sum_{\alpha \in \Delta} \frac{\alpha (a)}{\alpha (s)} e_{-\alpha} e'_\alpha + \frac{1}{2} \sum_{\alpha \in \Delta} \frac{\alpha (a)}{\alpha (s)^2} e_{-\alpha} e_\alpha  [e_{-\alpha}, e_{\alpha}] + \frac{1}{3} \sum_{\substack{\alpha, \beta \in \Delta \\ \alpha \neq \beta}} \frac{\alpha (a)}{\alpha (s) \beta(s)} e_{-\beta} e_\alpha [e_{-\alpha}, e_\beta], \ldots, 
\end{split}
 \]
where $ \Delta  $ is the set of roots of $ \mf{g} $ and the root vectors $ e_{\alpha} $ are chosen such that $ (e_{\alpha} | e_{-\alpha}) = 1. $

The corresponding Hamiltonian equations are:
\begin{equation}
\label{e6.24}
\frac{d b}{dt_n} = 0 \mbox{ for } b \in \mf{h}, 
\ n \in \mathbb{Z}_+, \,\,\, \frac{d e_\alpha}{dt_0} = \alpha (a) e_\alpha,
\end{equation}
\begin{equation}
\label{e6.25}
\frac{d e_\alpha}{dt_1} = \frac{\alpha (a)}{\alpha (s)} e'_\alpha + \sum_{\beta \in \Delta} \frac{\beta(a)}{\beta(s)} e_{- \beta} [e_\beta, e_\alpha].
\end{equation}
The next equation is more complicated, so we give it only for 
$ \mf{g} = s \ell_2,\, a = s,\, \alpha(s) = 1:  $
\begin{equation}
\label{e6.26}
\frac{d e_\alpha}{dt_2} = e''_\alpha - (2 e'_\alpha \alpha + e_\alpha \alpha') - (\alpha | \alpha) e^2_\alpha e_{-\alpha} \, .
\end{equation}
Note that the elements of $ \mf{h} $ do not evolve since they are central for the Poisson structure $ K. $
\end{exm}

In order to construct new PVAs from existing ones, we can use the classical affine Hamiltonian reduction of a PVA $\mathscr{V}$.

The classical affine Hamiltonian reduction of a PVA $\sV$, associated to a triple $(\sV_0,I_0,\varphi)$, where $\sV_0$ is a PVA, $I_0\subset\sV_0$ is a PVA ideal and $\varphi : \sV_0 \to \sV$ is a PVA homomorphism, is
\begin{equation}
\sW=\sW(\sV,\sV_0,I_0,\varphi) = (\sV/\sV\varphi(I_0))^{\text{ad}_\lambda \varphi(\sV_0)},
\end{equation}
where $ \text{ad}_\lambda \varphi(\sV_0)$ means that we are taking the adjoint action of $ \sV\varphi(I_0) $ on $ \sV $ with respect to the 
$ \lambda$-bracket.

\begin{rmk}
$ \sV/\sV\varphi(I_0) $ is a differential algebra, but the $ \lambda$-bracket is not well defined on this quotient. However, the $ \lambda$-bracket is well defined on the subspace of invariants $  (\sV/\sV\varphi(I_0))^{\text{ad}_\lambda \varphi(\sV_0)} $. 
\end{rmk}

\begin{thm}
The $\lambda$-bracket on $\sW$ given by
\begin{equation}
\{{f+\sV\varphi(I_0)}_\lambda\, g+\sV\varphi(I_0)\} = \{f{}_\lambda g\} + \sV\varphi(I_0)
\end{equation}
is well defined and it endows the differential algebra $\sW$ with a structure of a PVA.
\end{thm}

\begin{proof}
Let $\widetilde\sW=\{f\in\sV \,|\, \{\varphi(\sV_0) {}_\lambda f\} \subset \sV[\lambda]\varphi(I_0)\}$, so that $\sW=\widetilde\sW/\sV\varphi(I_0)$. It is a 
subalgebra 
of the differential algebra $\sV$, and $\sV\varphi(I_0)$ is its differential ideal.

Check that $\widetilde\sW$ is closed under the $\lambda$-bracket of $\sV$ (i.e.\@ $\widetilde\sW$ is a PVA subalgbra): let $h\in I_0$, $f,g\in\widetilde\sW$, then by the Jacobi identity
\begin{align}
\{h_\lambda \{ f_\mu g\}\} &=\{\{h_\lambda f\}_{\lambda+\mu} g\} + \{ f_\mu \{ h_\lambda g\}\} \subset \{\sV[\lambda]\varphi(I_0)_{\lambda + \mu} g \} + \{f_\mu \sV[\lambda]\varphi(I_0)\} \subset \nonumber\\
& \subset \{\sV[\lambda]\varphi(\sV_0)_{\lambda + \mu} g \} + \{f_\mu \sV[\lambda]\varphi(\sV_0)\} \subset \sV[\lambda, \mu] \varphi(I_0).
\end{align}
Finally, by the right Leibniz rule, $\sV\varphi(I_0)$ is a Poisson ideal of $\widetilde\sW$: for $f\in\widetilde\sW$ we have 
\begin{equation}
\{f_\lambda \sV\varphi(I_0)\} \subset \sV\{f_\lambda \varphi(I_0)\} + \{f_\lambda \sV\}\varphi(I_0) \subset \sV\{f_\lambda \varphi(\sV_0)\} + \sV[\lambda]\varphi(I_0) \subset \sV[\lambda]\varphi(I_0).
\end{equation}
\end{proof}

The main example of this construction is the classical affine $W$-algebra. 

\begin{exm}
\label{example6.3.4}
Consider the affine PVA $ \sV = \sV^1 (\mf{g}, s) $ with compatible $ \lambda $-brackets \eqref{e6.21}. Let $ f $ be a nilpotent element of $ \mf{g} $ and let $ \{ f, h, e \} $ be an $ s \ell_2 $-triple, containing $ f. $ Let
\[ \mf{g} =  \underset{j \in \frac{1}{2} \mathbb{Z}}{\bigoplus} \mf{g}_j \]
be the $ \frac{1}{2} ad \, h $ eigenspace decomposition (so that $ f \in \mf{g}_{-1} $). Assume that $ s \in \mf{g}_d,  $ where $ d = $ max$ \{ j | \mf{g}_j \neq 0 \}.  $ Let $ \sV = \mathcal{P} (\mf{g}_{>0}), $ let $ \varphi: \sV_0 \rightarrow \sV  $ be the inclusion homomorphism, and let $ I_0 \subset \sV_0 $ be the differential ideal, generated by the set 
\[ M = \{ m -(f|m) |\, m \in \mf{g}_{\geq 1} \}. \] 
It is easily checked that $ I_0 $ is a PVA ideal of $ \sV_0 $ with respect to both $ \lambda $-brackets \eqref{e6.21}. Then the \emph{classical affine $ W $-algebra} is the corresponding classical Hamiltonian reduction for both $ \lambda $-brackets:
\[ \sW (\mf{g}, f, s) = \sW (\sV, \sV_0, I_0, \varphi). \]
\end{exm}


\begin{rmk}
The same construction does not work for vertex algebras because of the presence of quantum corrections. However, for the usual associative algebras it actually works. The \emph{quantum (finite) Hamiltonian reduction} of an associative algebra $ A $ is $W=W(A,A_0,I_0,\varphi)$ constructed as above, where $ A_0 $ is an associative algebra, $\varphi : A_0 \hookrightarrow A$ is a homomorphism of associative algebras and $I_0 \subset A_0$ is a two-sided ideal. Thus, taking $A=U(\mf{g})$, $A_0=U(\mf{g}_{>0})$ and $I_0$ the two-sided ideal generated by the above set $M$, we get the \emph{quantum finite $W$-algebra}
\begin{equation}
W(\mf{g}, f) = W(A,A_0,I_0,\varphi).
\end{equation} 
\end{rmk}

\begin{thm}[\cite{DSKV13}]
As a differential algebra, the $ W$-algebra 
$\sW(\mf{g},f, s)$ 
is isomorphic to the algebra of differential polynomials on $\mf{g}^f$, the centralizer of $ f $ in $ \mf{g} $.
\end{thm}
In particular, for $ f $ principal nilpotent we get the classical Drinfeld-Sokolov reduction. 

The problem is, for which nilpotent elements $f$ can one construct the associated with $\sW(\mf{g},f, s)$  integrable hierarchy of Hamiltonian PDE's? 
Drinfeld and Sokolov constructed such hierarchy in \cite{DS85} for the
principal nilpotent $ f $. (For $\mf{g}=s\ell_n$ it coincides with the
Gelfand-Dickey $n$-th  KdV hierarchy, 
$n=2$ being the KdV hierarchy.) 
Their method is similar to that in the homogeneous case. The same method can be extended, but unfortunately not for all nilpotent elements.
\begin{defn}
A nilpotent element $f\in\mf{g}$ is called of \emph{semisimple type} if $f+s$ is a semisimple element of $\mf{g}$ for some $s\in\mf{g}_d$.
\end{defn}
These elements are classified for all simple Lie algebras $\mf{g}$ \cite{EKV}. For example, principal, subprincipal and minimal nilpotent elements are of semisimple type. In exceptional Lie algebras about one third of the nilpotent elements are of the semisimple type. In $\mf{s\ell}_N$ only those elements corresponding to partitions $(n,\ldots,n,1,\ldots,1)$ of $ N $ are of semisimple type.

\begin{thm} [\cite{DSKV13}]
Let $\mf{g} $ be a simple Lie algebra.  If $f\in \mf{g}$ is a nilpotent element, such that $ f + s $ 
is semisimple for $s\in \mf{g}_d$, then there exists a bi-Hamiltonian hierarchy associated to $\sW(\mf{g},f, s)$, which is both Lie and Liouville is integrable. 
\end{thm}

\begin{rmk}
In the recent paper \cite{DSKV15} for any nilpotent element $ f $ of $ g \ell_N $ and non-zero $s\in\mf{g}_d$ an integrable
hierarchy associated to $ W (g \ell_N, f, s) $ is constructed.
\end{rmk}

\subsection{Non-local Poisson structures and the Dirac reduction}

Unfortunately in  many important examples the PVA structure is not enough to deal with integrable systems, as it is in the case of the KdV equation, since in practice most of the Poisson structures are non-local. Thus we need to consider \emph{non-local PVAs}, for which the $ \lambda$-bracket takes value in $\mathscr{V}((\lambda^{-1}))$. Equivalently, the associated operator $ H(\partial) \in \text{Mat}_{\ell \times \ell} \mathscr{V}((\partial^{-1}))$ is now a matrix \textit{pseudodifferential} operator.

Still, we can work with these structures, but we have to check that the axioms for a PVA bracket still make sense when the $ \lambda$-bracket is a map $ \{ \cdot_\lambda \cdot \} : \mathscr{V} \otimes \mathscr{V} \rightarrow \mathscr{V}((\lambda^{-1})) $. Sesquilinearity and the left and right Leibniz rules are clear. For skewsymmetry we have to make sense of $ (\lambda + \partial)^{-1} $: write $ (\lambda +  \partial)^{-1} = \lambda^{-1}( 1 + \frac{\partial}{\lambda})^{-1} $ and then expand in the geometric progression, so we get a Laurent series in $ \lambda $. More generally, for an $ n \in \mathbb{Z} $ we let
\begin{equation}
(\lambda + \partial)^n = \sum\limits_{k \in \bZ_+} {{n}\choose{k}} \lambda^{n-k} \partial^k.
\end{equation}
We only have problems with the Jacobi identity, and in order for it to make sense we need the $ \lambda$-bracket to satisfy an additional property, called {\it admissibility}:
\begin{equation}
\{\{a_\lambda b\}_\mu c\} \subset \mathscr{V}[[\lambda^{-1},\mu^{-1},(\lambda-\mu)^{-1}]][\lambda,\mu].
\end{equation}
The fact is that when we consider a term like $ \{a_\lambda \{b_\mu c \} \} $ we have to take Laurent series in $ \lambda $ and then Laurent series in $ \mu $ and these can not be interchanged, since what we get are completely different spaces. So, two different terms of the Jacobi identity cannot a priori be compared, and we need this admissibility property in order to do so. If $H(\partial)=A(\partial)\circ B(\partial)^{-1}$ is a rational matrix pseudodifferential operator (that is, both $ A(\partial) $, $ B(\partial) $
 are $\ell \times \ell$ matrix differential operators and $ B(\partial)$ is non-degenerate), then the $ \lambda$-bracket defined by the Master Formula (\ref{masterformula}) is admissible.

\begin{exm}
 For $ \mathscr{V} = \mathscr{P}_1 = \F[u,u',u'',\ldots] $ examples of non-local Poisson structures are:
\begin{itemize}
\item $H(\partial) = \partial^{-1}$ (Toda)
\item $H(\partial) =u' \partial^{-1} u'$ (Sokolov).
\end{itemize}
\end{exm}

More information about non-local PVA can be found in \cite{DSK13}. In particular, it is shown there that the Lenard-Magri scheme can be applied if both $ K(\partial) $ and $ H(\partial) $ are rational pseudodifferential operators. 
One of the most important examples is the following pair of compatible
non-local Poisson structures on the algebra of differential polynomials
in $u$ and $v$, where $\kappa \in \F$: 
\begin{equation}
\label{e.kappa1}
K = \begin{pmatrix}
0 & -1 \\1 & 0 \\
\end{pmatrix}, \qquad H = \begin{pmatrix}
0 & \partial \\
\partial & 0 
\end{pmatrix} \ + 2 \kappa \begin{pmatrix}
u \partial^{-1} \circ u & - u \partial^{-1} \circ v \\
-v \partial^{-1} \circ u & v \partial^{-1} \circ v \\
\end{pmatrix}, 
\end{equation}
which produces the non-linear Schr\"{o}dinger (NLS) equation:
\begin{equation}
\label{e.kappa2}
\begin{split}
\frac{du}{dt} &= u'' + \kappa u^2v \\
\frac{dv}{dt} &= -v'' - \kappa uv^2.
\end{split}
\end{equation}

An important construction, leading to non-local PVA's, is the Dirac 
reduction for PVA's, introduced in \cite{DSKV14}, which generalizes the classical Dirac reduction for Poisson algebras \cite{Dir50}.

\begin{thm}[\cite{DSKV14}]
\label{Dirac}
Let $ (\sV, \{ ._\lambda .\}, \cdot ) $ be a (possibly non-local) PVA. Let $ \theta_1, \ldots, \theta_m \in \sV $ be some elements (constraints) such that
\[  C(\partial) = (( \{\theta_{\beta \ \partial} \ \theta_\alpha \})^m_{\alpha, \beta = 1})_\rightarrow  \]
is a non-degenerate matrix pseudodifferential operator. For $ f, g \in \sV $ let 
\begin{equation}
\label{e6.36}
\{ f_\lambda g \}^D = \{ f_\lambda g \} - \sum_{\alpha, \beta = 1}^{m} \{ \theta_{\beta \ \lambda + \partial} \ g \}_\rightarrow \ (C^{-1})_{\beta \alpha} (\lambda + \partial ) \{ f_\lambda \theta_\alpha \} \, .
\end{equation}
Then this modified $\lambda $-bracket provides $ \sV  $ with a structure of a non-local PVA, such that all elements $ \theta_\alpha $ are central. Consequently, the differential ideal of the PVA $ \sV^D = (\sV, \{ ._\lambda .\}^D, \cdot ), $ generated by the $ \theta_\alpha $'s is a PVA ideal, so that the quotient 
of $ \sV^D $ by this ideal is a non-local PVA. 
\end{thm}
\begin{proof}
Formula (\ref{e6.36}) defines the only $\lambda$-bracket, which satisfies sesquilinearity and skewsymmetry, and for which all the $\theta_i$ are central.
The proof of Jacobi identity is a long, but straightforward, calculation.
\end{proof}

\begin{exm}
\label{ex6.5}
Consider the affine PVA $ \sV = \sV^1 (s\ell_2, s) $ with the two compatible Poisson $ \lambda $-brackets $ \{ ._\lambda .  \}_H $ and 
$ \{ ._\lambda .  \}_K $, given by \eqref{e6.21}. As in Example \ref{example6.3.1}, choose a basis $ e_\alpha, e_{-\alpha}, s $ of $ s\ell_2, $ such that 
\[ [ e_\alpha, e_{-\alpha}] = s, \quad [s, e_{\pm \alpha}]=\pm e_{\pm\alpha}, \]
and the invariant bilinear form, such that $ (e_\alpha | e_{-\alpha}) = 1, (\alpha | \alpha ) = - \kappa. $ 

Consider the constraint $ \theta = s $ (=a  multiple of $ \alpha $). This constraint is central with respect to the $ \lambda $-bracket $ \{ ._\lambda .  \}_K. $ The quotient of $ \sV $ by the differential ideal, 
generated by $ \theta $, is the algebra of differential polynomials $ \mathscr{P}_2 $ in the indeterminates $ u = e_\alpha, v = e_{-\alpha}. $ The induced on $\mathscr{ P}_2$  $\lambda $-bracket by $ \{ ._\lambda .  \}_K $ is given by the matrix $ K $ in \eqref{e.kappa1}, and the Dirac reduced $ \lambda $-bracket $ \{ ._\lambda .  \}_H $ on $\mathscr{P}_2 $ is given by the matrix $ H $ in \eqref{e.kappa1}. The reduced by the constraint $ \theta $ evolution equation \eqref{e6.26} is the NLS equation (\ref{e.kappa2}). 
\end{exm}
This approach establishes integrability of the NLS equation, see \cite{DSKV14a} for details. For other approaches see \cite{TF86} and \cite{DSK13}.

\begin{exr} 
Dirac reduction of the affine PVA $\sV^1(\mf{g},s)$ by a basis of $\mf{h}$,
applied to equation (\ref{e6.25}), gives an integrable Hamiltonian equation
on root vectors of the reductive Lie algebra 
$\mf{g}$:
\[\frac{d e_\alpha}{dt} = \frac{\alpha (a)}{\alpha (s)} e'_\alpha + \sum_{\beta \in \Delta, \beta\neq -\alpha} \frac{\beta(a)}{\beta(s)} e_{- \beta} [e_\beta, e_\alpha],\]
where $a$ and $s$ are some fixed elements of $\mf{h}$, $s$ being regular.
Find its Poisson structures.
\end{exr}


\end{document}